\documentclass[12pt]{article}

\usepackage[letterpaper,top=2.54cm,bottom=2.54cm,left=2.54cm,right=2.54cm,marginparwidth=1.75cm]{geometry}

\usepackage[T1]{fontenc}
\usepackage{authblk}
\usepackage{lmodern}
\usepackage{textcomp}
\usepackage{amsmath}
\usepackage{amssymb}
\usepackage{eqlist} 
\usepackage[nosepfour,warning,np,debug,autolanguage]{numprint}
\usepackage{acro}
\usepackage{booktabs}
\usepackage[final]{graphicx}
\usepackage[dvipsnames]{color}
\DeclareGraphicsExtensions{.pdf, .jpg}
\usepackage{setspace}

\usepackage{url}
\usepackage[shortlabels]{enumitem}
\usepackage{mathtools}

\usepackage{titlesec}

\usepackage[colorlinks = true, 
            urlcolor  = black,
            citecolor = black,
            linkcolor = black]{hyperref}

\usepackage[numbers,sort&compress,super]{natbib}

\makeatletter
\renewcommand{\@biblabel}[1]{#1.}
\makeatother

\usepackage{amsthm}
\usepackage{thmtools}

\usepackage[singlelinecheck=false]{caption}
\usepackage{subcaption}
\newcounter{lofdepth}
\newcounter{lotdepth}
\usepackage{comment}
\usepackage{tabularray}
\usepackage{multirow}
\usepackage[table,xcdraw]{xcolor}
\usepackage{enumitem}

\usepackage[utf8]{inputenc}
\usepackage{csquotes}
\usepackage[english]{babel}
\usepackage{soul}
\usepackage{chngcntr}

\declaretheorem[name=Theorem, numberwithin=]{theorem}
\declaretheorem[sibling=theorem,name=Definition,numberwithin=]{definition}
\declaretheorem[sibling=theorem,name=Lemma,numberwithin=]{lemma}
\declaretheorem[sibling=theorem,name=Proposition,numberwithin=]{prop}
\declaretheorem[sibling=theorem,name=Corollary,numberwithin=]{cor}

\makeatletter
\renewenvironment{proof}[1][\proofname]{%
  \par\pushQED{\qed}%
  \normalfont
  \topsep6pt \partopsep0pt
  \trivlist
  \item[\hskip\labelsep\bfseries #1\@addpunct{.}]%
}{%
  \popQED\endtrivlist\@endpefalse
}
\makeatother

\begin{document}

\title{\vspace{-1em}An MCMC-Based Method for Dynamic Causal Modeling of Effective Connectivity in Functional MRI} 

\author[1]{Kaitlyn R. Fales\thanks{Corresponding author: kfales@psu.edu}}
\author[1]{Hyebin Song}
\author[1,2]{Nicole A. Lazar}

\affil[1]{Department of Statistics, Pennsylvania State University}
\affil[2]{Huck Institutes of the Life Sciences, Pennsylvania State University}

\date{}

\maketitle

\vspace{-1em}

\noindent\textbf{Shortened Title:} MCMC-Based Dynamic Causal Modeling for fMRI Connectivity

\bigskip

\begin{singlespace}

\noindent\textbf{Abstract:} Effective connectivity analysis in functional magnetic resonance imaging (fMRI) studies directional interactions among brain regions and experimental stimuli. Dynamic causal modeling (DCM) is a widely used method to estimate effective connectivity, based on a state-space representation consisting of a latent neural signal model and an observation model transforming the neural signal into the observed blood-oxygen–level-dependent (BOLD) response. A standard DCM combines ordinary differential equation (ODE) dynamics for the latent signal with a complex neural-hemodynamic system for the observation model, and typically uses variational Bayes for parameter estimation. While physically well-motivated, this approach can lead to practical challenges such as inexact solutions and underestimated uncertainty. We introduce Canonical DCM (CDCM), a Markov chain Monte Carlo (MCMC)-based method that adopts a simpler observation model and the No-U-Turn Sampler for posterior sampling. The simpler observation model admits a piecewise analytic solution to the neural ODE, increasing computational efficiency and enabling explicit derivation of sufficient conditions for parameter identifiability. The results indicate that CDCM provides reliable uncertainty quantification and consistent estimation of parameters related to experimental inputs for simulated and real data. We use publicly available data from the Wellcome Centre for Human Neuroimaging and the Human Connectome Project (HCP) to benchmark CDCM against standard DCM methods and examine replicability of estimated connectivity patterns in small- and large-scale neuroimaging settings.

\bigskip

    
\end{singlespace}

\section{Introduction}
State-space models defined by systems of ordinary differential equations (ODEs) are widely used to represent dynamic processes across scientific domains. Parameter estimation in these models is often challenging due to the presence of latent states, nonlinear dynamics, and partial or indirect observations.\cite{duan-etal-2020,stanhope-etal-2014,ljung-2010,schon-etal-2011} In such settings, distinct parameter sets may produce identical model outputs, raising the issue of parameter identifiability. Parameter identifiability generally refers to the ability to uniquely identify model parameters from observed data.\cite{wang-etal-2024,miao-etal-2011,villaverde-2019} Determining identifiability is a central problem in dynamical systems, with a large body of existing results, but one that remains inherently model-dependent. Identifiability is closely related to observability in the control systems literature, where observability is introduced first for linear systems,\cite{kalman-1960} and later extended to nonlinear systems.\cite{hermann-krener-1977} Bellman and {\AA}str{\"o}m\cite{bellman-astrom-1970} formalize the concept of structural identifiability for dynamical systems as the property that model parameters can be uniquely determined from noiseless observations of the system output. Without structural identifiability, reliable parameter estimation is fundamentally impossible, but even when parameters are identifiable under idealized conditions, finite and noisy data still complicate estimation in practice.\cite{raue-etal-2009,miao-etal-2011}

A large body of existing work studies identifiability and its implications for parameter estimation in dynamical systems. In system identification, identifiability of linear and affine state-space models depends heavily on input design and parameterization,\cite{duan-etal-2020,wang-etal-2024,stanhope-etal-2014,kailath-1980} while work in systems biology highlights the gap between theoretical identifiability and reliable estimation in finite or partially observed data.\cite{chis-etal-2011,miao-etal-2011,villaverde-2019,wieland-etal-2021,raue-etal-2009} The previous literature reflects a fundamental trade-off between model complexity and parameter identification.\cite{ljung-2010} More recent advances emphasize identifying estimable combinations of parameters \cite{ovchinnikov-etal-2021} and reparameterizing models to improve identifiability.\cite{ovchinnikov-etal-2025} 

These identifiability challenges are pronounced in studies of effective, or directed, brain connectivity in functional MRI (fMRI), where latent neural activity is inferred from indirect measurements. In neuroimaging, dynamic causal modeling (DCM) is a popular state-space framework in which latent neural dynamics are represented through a system of ODEs.\cite{friston-etal-2003} DCM models directed interactions between brain regions and their modulation by experimental inputs. DCM is inherently confirmatory: the state-space structure is specified \textit{a priori}, with many parameters fixed to zero according to hypothesized connections.\cite{friston-etal-2003,zeidman-etal-2019} Consequently, the framework does not support discovery of previously unmodeled interactions, increasing the importance of reliable and identifiable parameter estimation. Moreover, experimental design, including choices regarding the number, timing, and variation of stimuli, directly influences the ability to recover model parameters and test hypotheses.\cite{arand-etal-2015,zeidman-etal-2019} Although general design recommendations exist,\cite{friston-etal-2003, stephan-etal-2010, zeidman-etal-2019, kahan-foltynie-2013} they are often qualitative, leaving open questions about how design impacts identifiability and estimation in practice. Reliable parameter identification is also fundamental to replicability and reproducibility: without consistent recovery of model parameters across datasets or experimental realizations, inferred system dynamics may fail to generalize, limiting the scientific utility of the model.

Despite widespread use, theoretical understanding of parameter identifiability in DCM is limited. Most work relies on empirical evaluation through simulations and model recovery studies,\cite{marreiros-etal-2008, daunizeau-etal-2011, zeidman-etal-2019} or profile likelihood assessments.\cite{arand-etal-2015} Recent methodological developments aimed at improving computational scalability, including regression DCM \cite{frassle-etal-2021} and large-scale optimization strategies,\cite{zhuang-etal-2021} focus on efficient inference but do not establish formal identifiability guarantees. Similarly, spectral and frequency-domain extensions of DCM \cite{novelli-etal-2024} emphasize interpretability and links to functional connectivity but rely on empirical validation rather than theoretical analysis. Alternative causal modeling approaches, such as CausalMamba,\cite{bae-cha-2025} are aimed at addressing the computational intractability challenges of DCM. These methods rely on machine learning techniques to infer network connectivity and typically lack interpretability and explicit guarantees of parameter identifiability, further highlighting the need for theoretically grounded approaches to dynamical connectivity modeling.

The difficulty of formal identifiability analysis appears to arise, at least in part, from the complex structure of standard DCM, which couples the neural ODE with a nonlinear hemodynamic model of four additional differential equations.\cite{friston-etal-2000,friston-2002,stephan-etal-2007,zeidman-etal-2019} This complexity is emblematic of nonlinear state-space ODE models more generally, where latent dynamics interacting with observation models complicate theoretical analysis. Adding further model flexibility exacerbates the challenge; in the case of DCM, other modeling variations such as the two-state,\cite{marreiros-etal-2008} nonlinear,\cite{stephan-etal-2008} and stochastic DCM\cite{daunizeau-etal-2012} all substantially increase complexity. On the other hand, work in system identification and dynamical systems establishes conditions for parameter identifiability in linear and affine models.\cite{kailath-1980, duan-etal-2020, wang-etal-2024,stanhope-etal-2014} Similarly, systems biology literature develops frameworks for assessing structural and practical identifiability in nonlinear ODE models.\cite{arand-etal-2015,chis-etal-2011, miao-etal-2011,villaverde-2019,wieland-etal-2021,raue-etal-2009} Together, these lines of work suggest that, under appropriate simplifications, complex models such as DCM may be recast in a tractable form that permits analysis of parameter identifiability. Simplifications that admit piecewise or analytic solutions further facilitate both theoretical understanding and computational efficiency.\cite{chis-etal-2011,khalil-2002,sontag-1998,villaverde-2019,jacquez-simon-1993}

Most existing DCM implementations use variational Laplace (VL) for parameter estimation.\cite{zeidman-etal-2019,friston-etal-2003,marreiros-etal-2008,stephan-etal-2008,daunizeau-etal-2014} VL replaces the intractable posterior with a Gaussian approximation centered at the posterior mode and obtained through free energy optimization. While computationally efficient, the approximation imposes a unimodal, locally quadratic structure on the posterior that is a strong assumption for state-space models such as DCM.\cite{friston-etal-2006,zeidman-etal-2023} Furthermore, the free energy may have multiple local maxima, or even when the global maximum is reached, there is still a risk of biased inference.\cite{daunizeau-etal-2011}  

By contrast, Markov chain Monte Carlo (MCMC) methods avoid restrictive assumptions on posterior form and instead target the full posterior distribution, enabling more accurate representations of uncertainty and parameter dependencies. Although prior work shows that VL optimization and sampling-based approaches, such as MCMC, agree on model selection,\cite{friston-etal-2006} agreement does not imply equivalence of the inferred posteriors. Variational inference literature demonstrates these approaches often underestimate posterior uncertainty, particularly when the true posterior deviates from Gaussianity.\cite{blei-etal-2017} This limitation is especially problematic in fMRI, where inference is already challenged by indirect, noisy, and temporally coarse measurements of neural activity. 

Relevant work explores sampling-based approaches for DCM parameter estimation, particularly in neural mass models. Gradient-based MCMC methods are shown to more effectively explore the complex posterior geometries arising in DCM compared to gradient-free approaches.\cite{sengupta-etal-2015,sengupta-etal-2016} However, these approaches are applied to the full nonlinear DCM and therefore remain computationally demanding in practice, taking approximately 46 hours for 20,000 Hamiltonian Monte Carlo samples, but with much larger effective sample sizes than Langevin Monte Carlo, taking just over one hour for the same number of iterations.\cite{sengupta-etal-2016} While sampling-based methods are also used for estimating model evidence in DCM to perform Bayesian model selection, for example via thermodynamic integration,\cite{aponte-etal-2022} they similarly require repeated evaluation of the nonlinear state-space model. These limitations motivate the development of more computationally efficient approaches to sampling-based inference in DCM.

In this paper, we introduce Canonical DCM (CDCM), a simplified representation of the classical DCM from Friston et al\cite{friston-etal-2003} that retains its essential dynamical structure while reducing model complexity. Through the simplified framework, we make three main contributions. First, we show that the proposed model admits a piecewise analytic solution to the neural ODE, bypassing the need for numerical integration. The solution is exact and more computationally efficient than numerical integration for this class of state-space ODE models. Second, we derive sufficient conditions for parameter identifiability by exploiting the piecewise affine ODE solution made possible by the simplification. We can then connect the DCM setting to existing identifiability results from linear and affine dynamical systems. To the best of our knowledge, this is the first work to investigate and establish formal parameter identifiability in the context of DCM. Finally, we implement an exact parameter estimation procedure that allows reliable uncertainty quantification via MCMC, while improving computational scalability with the use of the piecewise analytic ODE solution. 

The remainder of the paper is organized as follows. In Section \ref{sec:motivation}, we present a motivating large-scale fMRI application that serves as a natural setting for evaluating the consistency and replicability of effective connectivity estimates. Section \ref{sec:methods} describes the proposed method, and Section \ref{sec:identification} provides sufficient conditions for unique parameter identification. In Section \ref{sec:simulation}, we present two simulation studies demonstrating the practical utility of the proposed approach under a range of settings. Section \ref{sec:application} presents two fMRI applications: the first compares CDCM with the classical DCM of Friston et al\cite{friston-etal-2003} using data from a single subject, and the second revisits the motivating example from Section \ref{sec:motivation}. We conclude with a discussion in Section \ref{sec:discussion}. 

\section{Motivating Data Example}\label{sec:motivation}

As a motivating example, we consider task fMRI data from the Human Connectome Project – Young Adult (HCP-YA) 2025 Release,\cite{vanessen-etal-2013,glasser-etal-2013} which highlights the practical challenges of assessing stability and replicability of effective connectivity estimates in DCM-based analyses. The HCP provides high-quality, publicly available fMRI data for over 1,000 healthy adults, with the 2025 release incorporating updated preprocessing, including spatial and temporal independent component analysis, multi-run denoising, and spin-echo intensity bias correction.\cite{hcpya2025} We focus on the social cognition task from the 100 Unrelated Subjects dataset, a Theory of Mind paradigm based on animated shapes.\cite{castelli-etal-2000, castelli-etal-2002, wheatley-etal-2007, white-etal-2011} In the task, participants view short clips of geometric shapes exhibiting either random (inanimate) or socially meaningful (animate) interactions and classify each clip accordingly. The design consists of alternating blocks of these conditions interleaved with fixation, and participants typically perform with high accuracy.\cite{barch-etal-2013,hillebrandt-etal-2014}

A prior study by Hillebrandt et al\cite{hillebrandt-etal-2014} applies a two-state DCM to the Q2 HCP data release to estimate directed connectivity between motion-sensitive (V5) and social cognition (posterior superior temporal sulcus; pSTS) regions.  The model includes two inputs corresponding to the experimental design: the ``All Motion'' input reflects the presence of any motion stimulus relative to fixation, while the ``Animate $>$ Inanimate'' contrast isolates the additional effect of viewing socially meaningful interactions. Together, the inputs define how effects of social meaning are modeled separately from basic motion processing.

In the two-state formulation, each region is represented by interacting excitatory and inhibitory neural subpopulations\cite{marreiros-etal-2008} rather than a single neural state, as in traditional DCM.\cite{friston-etal-2003} Using post-hoc Bayesian model selection over a large model space, they identify a fully connected model as having the highest posterior evidence and report consistent connectivity patterns across sessions and hemispheres. The selected model, with results as reported by Hillebrandt et al\cite{hillebrandt-etal-2014} is shown in Figure \ref{fig:hillebrandt-dcm}. The model includes bidirectional connectivity between V5 and pSTS; the ``All Motion'' input drives activity in both regions, and the ``Animate $>$ Inanimate'' contrast modulates all intra- and inter-regional connections. The consistency Hillebrandt et al\cite{hillebrandt-etal-2014} find in connectivity patterns may suggest that effective connectivity estimates for the social cognition task are stable and well-characterized.

\begin{figure}[t!]
    \centering
    \includegraphics[width=0.7\linewidth]{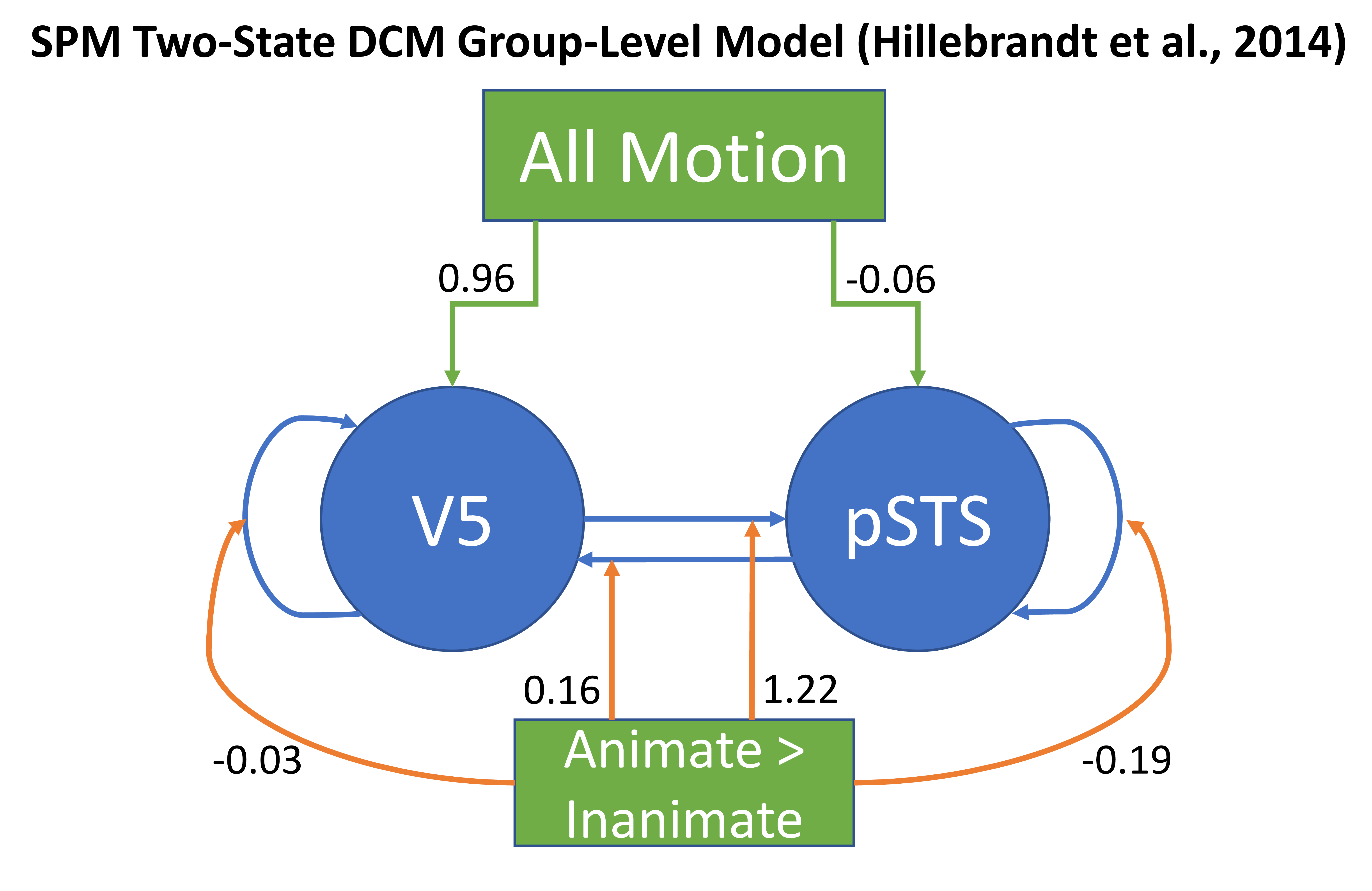}
    \caption{\protect\input{Figure-Captions/hillebrandt-dcm}}
    \label{fig:hillebrandt-dcm}
\end{figure}

However, several aspects of the previous analysis limit the strength of the conclusion. First, Figure \ref{fig:hillebrandt-dcm} demonstrates that results are only reported for a subset of parameters, and posterior uncertainty is not directly quantified or interpreted; uncertainty is instead summarized indirectly through ANOVA-based confidence intervals on selected contrasts. Second, the use of a two-state DCM introduces substantial model complexity beyond the standard one-state formulation. Because parameter identifiability is already difficult to assess even in one-state models,\cite{daunizeau-etal-2011} the complexity of a two-state model further exacerbates the problem. Third, the regional time series used for model estimation are constructed using subject-specific voxel sets defined via local activation maxima, introducing a selection mechanism where each subject’s data is derived from peak activation voxels. Variability in region definition introduces an additional source of bias in estimated connectivity.

Together, these considerations raise an important question: to what extent are the reported connectivity patterns replicable across choices in model specification, parameter estimation, and ROI construction? We therefore investigate replicability of the connectivity results given by Hillebrandt et al\cite{hillebrandt-etal-2014} using the updated 100 Unrelated Subjects dataset within the HCP-YA 2025 release. Between the updated data preprocessing and limited subject overlap (24 percent) with the earlier Q2 release used in Hillebrandt et al,\cite{hillebrandt-etal-2014} we assess replicability of results under meaningful data perturbations.

In the current work, we revisit this motivating example using CDCM, a framework that permits sufficient conditions for parameter identifiability. Rather than relying on model selection alone, we evaluate posterior uncertainty for all connectivity parameters to determine which connections are reliably supported. In addition, we construct regional time series using prespecified, consistent voxel sets across subjects, thereby removing variability in subject-specific ROI selection. By jointly addressing identifiability, uncertainty quantification, and ROI construction, the analysis provides a more rigorous assessment of the stability and replicability of effective connectivity estimates in a large-scale task fMRI setting.

\section{Methods}\label{sec:methods} 

The standard DCM framework from Friston et al\cite{friston-etal-2003} is formulated as a bilinear one-state neural model, in which each brain region is represented by a single latent neural state whose dynamics are governed by intrinsic connectivity parameters and modulatory effects of experimental inputs. Inputs enter both additively and through bilinear interactions with the state (e.g., products of inputs and states), allowing (experimental) condition-specific modulations of connectivity. Then, the traditional observation model describes the hemodynamic transformation from neural activity to the observed blood-oxygen-level-dependent (BOLD) response in fMRI. The Balloon model\cite{buxton-etal-1998} is adapted for use in DCM to account for the hemodynamics.\cite{friston-etal-2000,friston-2002,stephan-etal-2007} These dynamics are represented by a four-dimensional system of nonlinear ODEs with several parameters to estimate; for more details, see Appendix 5 of Zeidman et al.\cite{zeidman-etal-2019}

Our approach, Canonical DCM, differs from the traditional bilinear DCM \cite{friston-etal-2003} in two respects: 1) CDCM utilizes a simpler observation model, and 2) parameter estimation is accomplished using an MCMC-based approach. Instead of the four-dimensional nonlinear Balloon observation model, we propose convolving the latent neural signal with a fixed hemodynamic response function (HRF). There are several choices of HRF used in fMRI, but the most common, or canonical HRF, is the linear combination of two gamma functions,
\begin{equation}\label{eq:HRF}
        h(t;\,(\alpha_1,\alpha_2,\beta_1,\beta_2,c)) = \frac{\beta_1^{\alpha_1}}{\Gamma(\alpha_1)}t^{\alpha_1 -1}e^{-\beta_1t} - c\frac{\beta_2^{\alpha_2}}{\Gamma(\alpha_2)}t^{\alpha_2 -1}e^{-\beta_2t},
\end{equation}
with parameters $\alpha_1=6$, $\alpha_2=16$, $\beta_1=\beta_2 = 1$, and $c = 1/6$, as implemented in SPM (\url{https://www.fil.ion.ucl.ac.uk/spm/software/spm12/}), a popular neuroimaging software.\cite{henson-friston-2006,lindquist-etal-2009} This proposal offers two main advantages. With the simplification, the model retains only a single linear (affine) ODE, for which an analytical solution can be derived. The analytical solution provides the ability to explicitly analyze the identifiability of the ODE system and removes the reliance on numerical solvers, enabling faster MCMC-based computation as compared to MCMC for traditional DCM.\cite{sengupta-etal-2016}

There is an important relation between the two observation model approaches through the Volterra series expansion of the BOLD signal. The Volterra expansion is similar to Taylor expansion except that the Volterra series captures memory, making it applicable to the BOLD signal as it ``remembers'' dependencies in the signal.\cite{friston-etal-2000}  Friston\cite{friston-2002} shows that if the biophysical parameters in the Balloon model are treated as fixed, the state equations linearized, and the Volterra expansion is cut off after first-order terms, the linearized Balloon model is well-approximated by the convolution of neural activation with the canonical HRF. Therefore, our convolution approach can be considered a special case approximation of the Balloon model. The details of the relationship are given in Supplementary Material A.

The convolution of neural activation with the canonical HRF is a natural choice for simplifying DCM while maintaining its biophysical foundations. In a simulation study of practical identifiability for DCM, Arand et al\cite{arand-etal-2015} use the canonical HRF approach coupled with profile likelihood estimation to obtain reasonable parameter estimates for synthetic data generated using various scanning parameters (e.g., repetition time, scanning session duration, block duration). A profile likelihood scheme can be viewed as a ``worst case scenario'' approach for DCM since it does not incorporate any prior information \cite{raue-etal-2012} and helps motivate the use of the canonical HRF within a Bayesian framework.

\subsection{Model Specification}\label{sec:mdl}
We now specify the latent neural dynamics and observation model in CDCM. 
Let $\mathbf{z} = \{\mathbf{z}(t)\}_{t\ge 0}$ denote the latent neural signal, initialized at $\mathbf{z}(0)\in \mathbb{R}^d$, where $\mathbf{z}(t) \in \mathbb{R}^d$ is the vector of states corresponding to neural activation at time $t$ across $d$ ROIs. Following Friston et al\cite{friston-etal-2003},  we model the latent neural dynamic by the ODE
\begin{equation}\label{eq: z_dynamics}
     \dot{\mathbf{z}}(t) = (\mathbf{A} + \sum_{i=1}^m u_i(t) \mathbf{B}_i)\mathbf{z}(t) + \mathbf{Cu}(t),\,\,\,t \in (0,\infty),
\end{equation}
where the vector $\mathbf{u}(t)=[u_1(t),\dots,u_m(t)]^\top \in \mathbb{R}^m$ indicates which, if any, of the $m$ stimuli are present at time $t$. Specifically, $u_i(t)=1$ indicates that the $i$th stimulus is present at time $t$, and $u_i(t)=0$ otherwise, for $i=1,\dots,m$. 

In the neural dynamic equation \eqref{eq: z_dynamics},  $\mathbf{A} \in \mathbb{R}^{d \times d}$ parameterizes the connections between ROIs in the hypothesis (endogenous connectivity), $\mathbf{B}_i \in \mathbb{R}^{d \times d}$ parameterizes the modulation from stimulus $i$ to the (hypothesized) connections between ROIs for $i=1,\dots,m$, and $\mathbf{C} \in \mathbb{R}^{d \times m}$ parameterizes the driving connections from experimental stimuli directly to ROIs. 
Based on the research hypothesis, many values in $\mathbf{A}$, $\mathbf{B} = [\mathbf{B}_1,\hdots,\mathbf{B}_m]$, and $\mathbf{C}$ are hypothesized to be zero, and are not considered during estimation. We define $\text{vec}(\mathbf{A})_{\neq 0}$ as the vectorization of the nonzero elements of $\mathbf{A}$. The parameters of interest in $\mathbf{A}$, $\mathbf{B}$, and $\mathbf{C}$ are $\nu_A = \text{vec}(\mathbf{A})_{\neq 0}$, $\nu_B = \text{vec}(\mathbf{B})_{\neq 0}$, and $\nu_C = \text{vec}(\mathbf{C})_{\neq 0}$. We summarize the set of parameters for the neural signal model as $\theta_z = \{\nu_A,\nu_B,\nu_C\}$. From this point forward, we refer to $\theta_z \in \mathbb{R}^{p_{\theta_z}}$ as the neural parameters, or neural parameter vector. 

DCM is specified in terms of continuous time $t \in [0,\infty)$ since the true signals are continuous. However, in practice, we only observe the BOLD response, denoted as $\mathbf{y}(t) \in \mathbb{R}^d$ at time $t$, at equally spaced repetition times (TR). Let $r$ denote the repetition time, and define $t_j = rj$ for $j = 1, \dots, n$, where $n$ is the number of scans in the fMRI session. To account for the hemodynamics, we model the mean of the observed BOLD response at the observed time points as the sum of ROI-specific offsets and a convolution term based on the convolution of $\mathbf{z}$ with the canonical HRF $h(\cdot;\theta_y)$ defined in \eqref{eq:HRF}, with fixed parameter values $\theta_y = \{\alpha_1=6, \alpha_2=16, \beta_1= 1, \beta_2 = 1, c = 1/6\}$. For notational simplicity, we denote the HRF as $h(\cdot) = h(\cdot; (6,16,1,1,1/6))$. We define the convolution component of the mean BOLD response at the observed time points as
\begin{equation}\label{eqn:discrete-convolve}
    \mu(t_j) = \sum_{i=0}^j h(t_i)\mathbf{z}(t_{j-i}),
\end{equation}
where the sum truncates at $i=0$ since both the neural signal and the canonical HRF are zero prior to the start of the scan.

Combining the latent neural dynamics and the observation model, CDCM is specified by
\begin{align}\label{eqn:canonical-dcm}
     \dot{\mathbf{z}}(t) &= (\mathbf{A} + \sum_{i=1}^m u_i(t) \mathbf{B}_i)\mathbf{z}(t) + \mathbf{Cu}(t),\,\,\,t \in (0,\infty) \\ \label{eqn:canonical-dcm-obs}
    \mathbf{y}(t_j) &= \mu(t_j) +  \beta + \varepsilon(t_j),\,\,\, j \in \mathbb{N},
\end{align}
where $\beta \in \mathbb{R}^d$ is an intercept term accounting for the mean of the BOLD signal within each region, and the measurement errors $\varepsilon(t_j)$, for $j=1,\dots,n$, are iid multivariate Gaussian with variance matrix $\text{diag}(\sigma_1^2,\dots,\sigma_d^2)$. Let $\mathbf{s}^* = \mathbf{z}(0)$ denote the initial state. 
The unknown parameters are the neural parameters $\theta_z$, the initial system condition $\mathbf{s}^* \in \mathbb{R}^d$, the intercept vector $\beta \in \mathbb{R}^d$, and the measurement error variance $\{\sigma_1^2,\dots,\sigma_d^2\}$.

To constrain the dynamic system and prevent it from diverging, we impose strictly negative self-loop connections on the diagonal of $\mathbf{A}$, concentrating them around a baseline of $-0.5$ Hz, similar to other DCM implementations~(e.g., Zeidman et al).\cite{zeidman-etal-2019} We reparameterize the diagonal entries of $\mathbf{A}$ via an unconstrained parameter vector $\nu_A^D$, such that $\text{diag}(\mathbf{A}) = -0.5\exp(\nu_A^D)$, where the transformation enforces negativity and the prior controls deviation from the baseline. In contrast to SPM, which reparameterizes the diagonal entries of both $\mathbf{A}$ and $\mathbf{B}_i$ for $i = 1,\dots,m$, we reparameterize only the diagonal of $\mathbf{A}$ and impose tighter priors on the diagonals of $\mathbf{A}$ and $\mathbf{B}_i$. More details on the distinction between the two parameterizations and its implications are discussed in Supplementary Material B.

For prior specification, we separate the diagonal and off diagonal components of $\mathbf{A}$ and $\mathbf{B}$. Define $\nu_{(A,B)}^D = \mathrm{vec}(\nu_A^D, \text{diag}(\mathbf{B}))_{\neq 0}$ as the vectorization of all nonzero diagonal entries across $\mathbf{A}$ and $\mathbf{B}$, where $\nu_A^D$ denotes the unconstrained representation of $\text{diag}(\mathbf{A})$. Similarly, we define $\nu_{(A,B)}^O = \mathrm{vec}(\text{offdiag}(\mathbf{A}), \text{offdiag}(\mathbf{B}))_{\neq 0}$ as the vectorization of all nonzero off-diagonal entries across $\mathbf{A}$ and $\mathbf{B}$. We then specify priors on the neural parameters $\theta_z$ through $ \nu_{(A,B)}^D$, $\nu_{(A,B)}^O $, and $\nu_C$. Specifically, the priors are given by,
\begin{align*}
    \nu_{(A,B)}^D &\sim \text{N}_{dim(\nu_{(A,B)}^D)}(0,\sigma_{\nu^D}^2\mathbf{I}), \\
    \nu_{(A,B)}^O &\sim \text{N}_{dim(\nu_{(A,B)}^O)}(0,\sigma_{\nu^O}^2\mathbf{I}), \\
    \nu_C &\sim \text{N}_{dim(\nu_C)}(0,\sigma_{\nu^O}^2\mathbf{I}), \\
    \mathbf{s}^* &\sim \text{N}_d(0,\sigma_{s}^2\mathbf{I}), \\
    \beta &\sim \text{N}_d(0, \mathbf{I}), \\
    \varepsilon(t_j)|\{\sigma_1^2,\dots,\sigma_d^2\} &\sim \text{N}_d(0,\text{diag}(\sigma_1^2,\dots,\sigma_d^2)), \\
    \sigma_\ell &\sim \text{Exp}(0.5) \text{ for } \ell \in \{1,\dots,d\}, 
\end{align*} 
where $\sigma_{\nu^D} = 0.125$, $\sigma_{\nu^O} = 1$, and $\sigma_s = 0.3$. This leads to the following likelihood for $\mathbf{y}(t_j)$ for $j=1,\dots,n$:
\begin{equation*}
    \mathbf{y}(t_j) | \{\theta_z, \mathbf{s}^*, \beta,  \sigma_1^2,\dots,\sigma_d^2\} \sim \text{N}_d(\mu(t_j) + \beta , \text{diag}(\sigma_1^2,\dots,\sigma_d^2)).
\end{equation*}

\subsection{Piecewise Analytic Solution of the Neural Signal}\label{sec:neural-ode-sol}

The observation model simplification via the canonical HRF is the contribution that makes MCMC-based estimation more feasible, but adoption of CDCM also hinges on computational efficiency. While MCMC provides a more faithful characterization of posterior uncertainty for DCM, its computational demands hinder routine application. Here, we demonstrate that a simplified model formulation, combined with a piecewise analytic solution to the neural state equation, renders MCMC both tractable and efficient. Many fMRI experiments use a block design, where subjects experience a stimulus repeatedly over blocks of time \cite{tie-etal-2009}; it is the intuition behind the block design that enables a piecewise analytic solution for CDCM.

Each element of $\mathbf{u}(t)$ at time $t$ is in $\{0,1\}$ based on which stimuli are present during a particular block of time. We divide the experiment into $B+1$ distinct blocks, assigning a new block each time the set of stimuli changes. Within a block, the subject experiences the same stimulus repeatedly. Let $\mathcal{T}^{(b)}$ be the set of time points $t_j$ belonging to the $b$th block and let $t^{(b)}$ represent the minimum of $\mathcal{T}^{(b)}$, which is the starting time point of the block $b$ for $b=0,1,\dots,B$. Then, $\mathcal{T}^{(b)} = \{t^{(b)},\, t^{(b)}+r,\,t^{(b)}+2r,\,\dots,t^{(b+1)}-r\}$, where we recall that $r$ is the repetition time (TR). Also, we represent the experimental stimuli as $\Tilde{\mathbf{U}} := [\Tilde{\mathbf{u}}^{(0)\top},\hdots,\Tilde{\mathbf{u}}^{(B)\top}]^\top \in \mathbb{R}^{(B+1) \times m}$, where $\Tilde{\mathbf{u}}^{(b)} = [\Tilde{u}_1^{(b)},\hdots,\Tilde{u}_m^{(b)}]^\top \in \mathbb{R}^m$ is a stimulus indicator vector for all observations within block $b$ for $b \in \{0,1,\hdots,B\}$. 

Note that for any $t_j \in \mathcal{T}^{(b)}$, $\mathbf{u}(t_j) = \Tilde{\mathbf{u}}^{(b)}$. The trajectory of the ODE system (\ref{eqn:canonical-dcm}) either follows linear or affine ODE dynamics \textit{within each block}, as $\sum_{i=1}^m u_i(t) \mathbf{B}_i$ and $\mathbf{C}\mathbf{u}(t)$ are constant for $t \in \mathcal{T}^{(b)}$. This observation underlies the closed-form solution to the neural ODE and enables us to leverage existing results on identifiability for linear or affine ODE systems.\cite{stanhope-etal-2014, duan-etal-2020, wang-etal-2024} 
More concretely, for $t \in \mathcal{T}^{(b)}$, the dynamics in \eqref{eqn:canonical-dcm} become
\begin{align}\label{eq:z_dynamics_within_b}
    \dot{\mathbf{z}}(t) = \tilde{\mathbf{A}}^{(b)}\mathbf{z}(t) + \tilde{\mathbf{c}}^{(b)}
\end{align}
where $\Tilde{\mathbf{A}}^{(b)} := (\mathbf{A} + \sum_{i=1}^m \Tilde{u}_i^{(b)}\mathbf{B}_i) \in \mathbb{R}^{d \times d}$ and $\Tilde{\mathbf{c}}^{(b)}:=\mathbf{C}\Tilde{\mathbf{u}}^{(b)} \in \mathbb{R}^d$ are both constants within block $b$. Note that the equation \eqref{eq:z_dynamics_within_b} is linear if $\tilde{\mathbf{c}}^{(b)}=0$ and affine otherwise. It is well known that for an affine ODE \begin{equation}\label{eq: affine_ODE}
    \dot{\mathbf{v}}(t) = \mathbf{A}\mathbf{v}(t) + \mathbf{c},
\end{equation}
with the initial condition $\mathbf{v}(0)$, the closed-form solution is given by,
\begin{equation}\label{eq: affine_ODE_sol}
    \mathbf{v}(t) = e^{\mathbf{A}t}(\mathbf{v}(0)+\mathbf{A}^{-1}\mathbf{c}) - \mathbf{A}^{-1}\mathbf{c},
\end{equation}
where $e^{\mathbf{A}t}$ denotes the matrix exponential. 
Proposition \ref{prop:dcm_sol_trajectory} formalizes the observation, and shows that the solution of the dynamics in \eqref{eqn:canonical-dcm} is piecewise identical to a solution of the form \eqref{eq: affine_ODE_sol}.
\begin{prop}
    \label{prop:dcm_sol_trajectory}

Let $\mathbf{v}(t;\, \mathbf{v}(0), (\mathbf{A},\mathbf{c}))$ denote the solution at time $t\in \mathbb{R}_+$ of the affine system \eqref{eq: affine_ODE}.
Also let $\mathbf{z}(t; \mathbf{s}^*,\theta^*) $ denote the solution of the DCM dynamics in \eqref{eqn:canonical-dcm}, and $\theta^* = \{\mathbf{A},\mathbf{B}_1,\hdots,\mathbf{B}_m,\mathbf{C}\}$. For $b \in \{0,\dots,B\}$, with $\Tilde{\mathbf{A}}^{(b)}$ and $\Tilde{\mathbf{c}}^{(b)}$ defined in \eqref{eq:z_dynamics_within_b},
\begin{equation*}
    \mathbf{z}(t; \,\mathbf{s}^*,\theta^*) = \mathbf{v}(t-t^{(b)};\,\mathbf{z}(t^{(b)}),(\Tilde{\mathbf{A}}^{(b)},\Tilde{\mathbf{c}}^{(b)})),\,\, t \in [t^{(b)}, t^{(b+1)}).
\end{equation*}
In particular, the closed-form solution of $\mathbf{z}(t; \mathbf{s}^*,\theta^*) $ is given by
\begin{equation*}
    \mathbf{z}(t; \mathbf{s}^*,\theta^*) = e^{\Tilde{\mathbf{A}}^{(b)}(t-t^{(b)})}\{\mathbf{z}(t^{(b)})+(\Tilde{\mathbf{A}}^{(b)})^{-1}\Tilde{\mathbf{c}}^{(b)}\} - (\Tilde{\mathbf{A}}^{(b)})^{-1}\Tilde{\mathbf{c}}^{(b)},\,\,t \in [t^{(b)}, t^{(b+1)}).
\end{equation*}
\end{prop}
The proof of Proposition \ref{prop:dcm_sol_trajectory} is given in Supplementary Material C. To evaluate the likelihood at parameters $\mathbf{s}$, $\theta$, we compute $\mathbf{z}(t_j; \mathbf{s},\theta)$ for $j=1,\dots,n$. By leveraging Proposition \ref{prop:dcm_sol_trajectory}, we apply the closed-form solution to efficiently compute the ODE trajectory, therefore bypassing the need for numerical ODE integration. A schematic illustration of Proposition~\ref{prop:dcm_sol_trajectory} is provided in Figure~\ref{fig:prop1-dynamics}. The figure highlights how, within each block $b \in \{0,1,2,3\}$, the stimulus $\tilde{\mathbf{u}}^{(b)}$ partitions time into blocks over which the neural dynamics reduce to affine ODEs, and how the trajectory $\mathbf{z}(t)$ is obtained by solving each block in closed form and propagating through the boundary states $\mathbf{z}(t^{(b)})$.

\begin{figure}[t!]
    \centering
    \includegraphics[width=\linewidth]{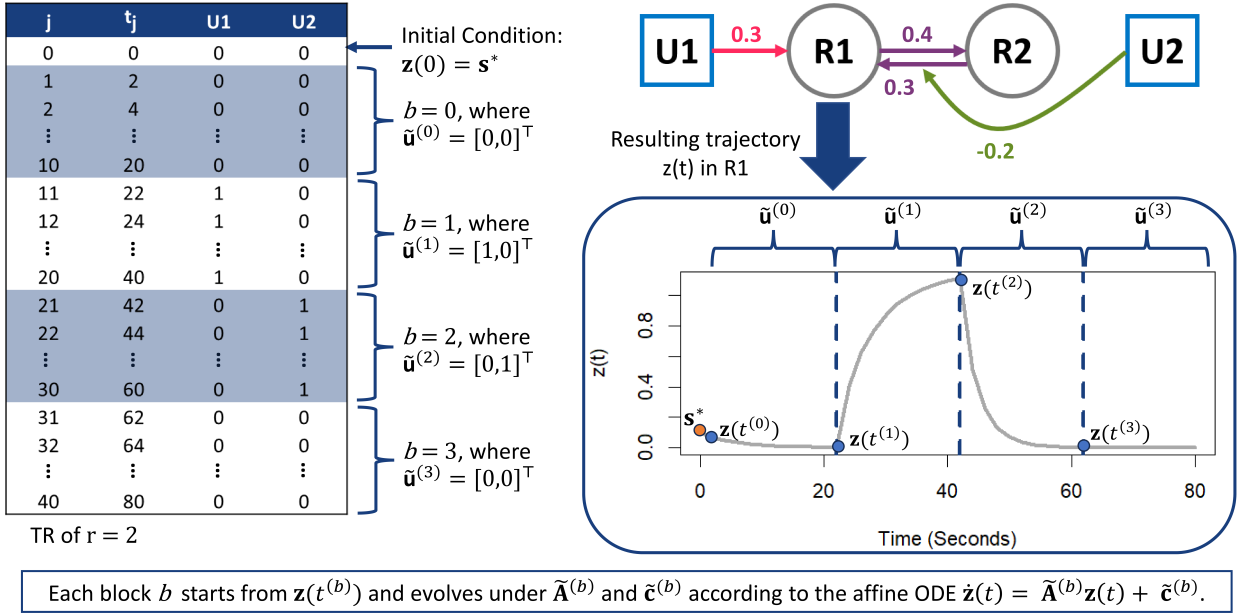}
    \caption{Schematic representation of Proposition~\ref{prop:dcm_sol_trajectory}. Constant stimulus blocks induce affine dynamics, and the trajectory $\mathbf{z}(t)$ is constructed by solving each block and propagating via $\mathbf{z}(t^{(b)})$.}
    \label{fig:prop1-dynamics}
\end{figure}

Proposition \ref{prop:dcm_sol_trajectory} is central to the CDCM framework because it makes MCMC-based estimation both tractable and efficient, and also, places the model in a form amenable to theoretical analysis. Notably, the piecewise analytic solution to the neural ODE enables the identifiability results we develop in Section \ref{sec:identification}, linking the block design structure directly to conditions under which parameters can be uniquely recovered.

\subsection{Individual and Group-Level Parametric Estimation Procedure}\label{sec:par-est}

Through the simplification of the observation model, CDCM involves fewer parameters and a simpler structure, enabling an exact parametric estimation procedure that is made more computationally efficient with the use of the piecewise analytic solution to the neural ODE. We use the No-U-Turn Sampler (NUTS) \cite{hoffman-gelman-2014} for posterior distribution sampling of network parameters in Stan \cite{stan-2024} with the \texttt{cmdstanr} interface.\cite{gabry-etal-2024} Initial values for the sampler are obtained via the Pathfinder algorithm.\cite{zhang-etal-2022} 

Convergence of the MCMC chains is assessed using the multivariate effective sample size (ESS) from Vats et al,\cite{vats-etal-2019} as implemented in the \texttt{mcmcse} package.\cite{mcmcse} The multivariate ESS is defined as,
\begin{equation*}
\widehat{\text{multiESS}}
= N \cdot \left( \frac{|\widehat{\Lambda}|}{|\widehat{\Sigma}|} \right)^{1/P},
\end{equation*}
where $N$ is the number of sampling draws, $P$ is the total number of parameters, $\widehat{\Lambda}$ is the sample covariance of the draws and $\widehat{\Sigma}$ is a consistent estimator of the asymptotic covariance matrix $\Sigma$. We note $P = p_{\theta_z} + 3d$, where $p_{\theta_z}$ is the number of neural parameters, and the remaining parameters are $\beta,\mathbf{s}^* \in \mathbb{R}^d$, and $\sigma_\ell^2$ for $\ell = 1,\dots,d$, resulting in an additional $3d$ parameters not of primary interest. We estimate the long-run covariance matrix $\Sigma$ using the multivariate moment least-squares (MLS) estimator from Song and Berg,\cite{song-berg-2025} implemented in the \texttt{momentLS} package \cite{momentLS} which estimates the auto- and cross-covariance sequence via a moment representation based on the reversibility of the Markov chain, and uses these to construct an asymptotic variance estimate.


Chains are considered converged once the estimated multivariate ESS satisfies the multivariate stopping criterion \cite{vats-etal-2019}:
\begin{equation*}
\widehat{\text{multiESS}} \ge W(P, \alpha_{ESS}, \varepsilon_{ESS}),
\end{equation*}
where $W(P, \alpha_{ESS}, \varepsilon_{ESS})$ depends on the dimension $P$, level $\alpha_{ESS}$, and relative precision $\varepsilon_{ESS}$. In practice, the stopping criterion is most sensitive to the choice of $\varepsilon_{ESS}$, with smaller values requiring substantially larger ESS to ensure small Monte Carlo standard errors. The effects of $P$ and $\alpha_{ESS}$ are comparatively modest, although increasing $P$ or decreasing $\alpha_{ESS}$ also requires longer chains to reach sampler convergence.

CDCM is a method for estimating a single-subject's directed connectivity, but most often, an fMRI study has many subjects, and the broader goal of the study involves examining connectivity patterns across the entire group. If we treat the individual subject-level parameters of interest as an independent ``study'', we can then use a standard multivariate normal–normal Bayesian hierarchical model to estimate the group-level DCM across all subjects.\cite{gelman-etal-2013,smith-spiegelhalter-1995}

For each subject $k = 1,\dots,K$, let $\hat{\theta}_{z,k} \in \mathbb{R}^{p_{\theta_z}}$ denote the estimated posterior means for subject $k$, and $\mathbf{S}_k \in \mathbb{R}^{p_{\theta_z} \times p_{\theta_z}}$ denote the within-subject posterior covariance matrix of $\hat{\theta}_{z,k}$, both estimated for every subject individually using CDCM. At the group-level, we control for differences between subjects by including subject-level covariates such as age, sex, fMRI task performance measures (if applicable), and other relevant sample characteristics. Let $\mathbf{b}_k \in \mathbb{R}^{q_s}$ denote the corresponding subject-level covariates for subject $k$. 

We model the true subject neural parameters $\theta_{z,k}$ hierarchically as
\begin{equation*}
    \theta_{z,k} \sim \text{N}_{p_{\theta_z}}(\boldsymbol{\alpha} + \Theta \mathbf{b}_k, \mathbf{T}),
\quad
\hat{\theta}_{z,k} \mid \theta_{z,k} \sim \text{N}_{p_{\theta_z}}(\theta_{z,k}, \mathbf{S}_k),
\end{equation*}
where $\boldsymbol{\alpha} \in \mathbb{R}^{p_{\theta_z}}$ are group-level intercepts indicating baseline group-level connectivity, $\Theta \ \in \mathbb{R}^{{p_{\theta_z}} \times q_s}$ are regression coefficients corresponding to the covariates, and $\mathbf{T} \in \mathbb{R}^{{p_{\theta_z}} \times {p_{\theta_z}}}$ is the between-subject covariance matrix. Because both levels are normal, the latent $\theta_{z,k}$ can be integrated out, yielding the marginalized likelihood,
\begin{equation*}
    \hat{\theta}_{z,k} \sim \text{N}_{p_{\theta_z}}(\boldsymbol{\alpha} + \Theta \mathbf{b}_k, \mathbf{T} + \mathbf{S}_k),
\end{equation*}
which is used for efficient estimation of group-level effects. Full details on the priors, covariance factorization, and sampling procedure are provided in Supplementary Material D.

\section{Identifiability of Model Parameters}\label{sec:identification}

We now turn to the theoretical contribution of this work, namely the development of conditions for parameter identifiability of CDCM. Loosely, identifiability for an ODE system refers to the ability to uniquely identify parameters from a whole trajectory for a given input and output, where a whole trajectory refers to the solution curve of the ODE considering all possible states.\cite{miao-etal-2011,wang-etal-2024} In many situations, we do not have access to the whole trajectory, we only have access to one, or a trajectory from one initial condition. Furthermore, even a single trajectory implies completely observing the system path in continuous time. This also is unlikely in practice and we only observe a set of discrete observations instead. In fMRI, we have a set of equally-spaced discrete observations sampled from a single trajectory. 

We discuss identifiability of CDCM in terms of error-free observations. Specifically, we consider if model parameters can be uniquely identified from inputs $\mathbf{U} = [\mathbf{u}(t_1)^\top,\hdots,\mathbf{u}(t_n)^\top]^\top \in \mathbb{R}^{n \times m}$ and observations $\mathbf{Y} = [\mathbf{y}(t_1)^\top,\hdots,\mathbf{y}(t_n)^\top]^\top \in \mathbb{R}^{n \times d}$, in the absence of measurement error. If two distinct sets of model parameters produce the same trajectory, consistent estimation of the model parameters is impossible, as these parameters are indistinguishable based on the observed data. We first establish a set of sufficient conditions for the identifiability of model parameters in the ODE system (\ref{eqn:canonical-dcm}) from the latent discrete $\mathbf{z}(t_j)$, as well as conditions for an injective mapping from $\mathbf{z}(t_j)$ to the mean of the observed data, $\mathbb{E}[\mathbf{y}(t_j)] = \mu(t_j) + \beta$. The injective mapping ensures that the parameters uniquely identified from $\mathbf{z}(t_j)$ are uniquely translated to the observable mean trajectory. 

The network parameters are $\theta^* = \{\mathbf{A},\mathbf{B}_1,\hdots,\mathbf{B}_m,\mathbf{C}\}$, along with initial system condition $\mathbf{s}^*$. We use $\theta^*$ here instead of $\theta_z$ because identifiability is defined at the level of the full dynamic system matrices; since $\theta_z$ corresponds to the nonzero elements of these matrices, identification of $\theta^*$ implies identification of $\theta_z$. For future reference, let $\mathbf{z}(t; \mathbf{s},\theta)$ denote the solution of the DCM dynamics in \eqref{eqn:canonical-dcm} at time $t$, given the initial state $\mathbf{s}$ and parameter $\theta$.  When $\mathbf{s}$ and $\theta$ are taken at the true values $\mathbf{s}^*$ and $\theta^*$, we suppress the explicit dependence on $\mathbf{s}^*,\theta^*$ and simply write  $\mathbf{z}(t) = \mathbf{z}(t; \mathbf{s}^*,\theta^*)$. We first define the identifiability that our theoretical result refers to, which is an extension of $(\mathbf{s}^*,\mathbf{A})$-identifiability in Wang et al.\cite{wang-etal-2024}
\begin{definition}[$(\mathbf{s}^*,\theta^*)$-identifiability]\label{def:identifiable}
    For $\mathbf{s}^* \in \mathbb{R}^d$ and $\theta^* =\{\mathbf{A},\mathbf{B}_1,\hdots,\mathbf{B}_m,\mathbf{C}\}$, we say the ODE system in (\ref{eqn:canonical-dcm}) is $(\mathbf{s}^*,\theta^*)$-identifiable from $\mathbf{z}(t_1),\hdots,\mathbf{z}(t_n)$ if for all $\mathbf{s}^\prime \in \mathbb{R}^d$ and $\theta^\prime =\{\mathbf{A}^\prime,\mathbf{B}_1^\prime,\hdots,\mathbf{B}_m^\prime,\mathbf{C}^\prime\}$ with $(\mathbf{s}^\prime,\theta^\prime) \neq (\mathbf{s}^*,\theta^*)$, there exists $j$ for $j \in \{1,\hdots,n\}$, such that $\mathbf{z}(t_j;\mathbf{s}^\prime,\theta^\prime) \neq \mathbf{z}(t_j;\mathbf{s}^*,\theta^*)$.
\end{definition}
Similar to Wang et al,\cite{wang-etal-2024} the definition of identifiability specifically refers to discrete error-free observations sampled from a single trajectory $\mathbf{z}(t;\mathbf{s}^*,\theta^*)$, meaning $\mathbf{s}^*$ and $\theta^*$ are fixed and $\mathbf{s}^\prime$ and $\theta^\prime$ are arbitrary (within the parameter space). That is, Definition \ref{def:identifiable} describes a property of a single system rather than a collective set of systems.

Now we discuss the $(\mathbf{s}^*,\theta^*)$-identifiability of CDCM. We make two distinct types of assumptions: 1) assumptions on the experimental design, and 2) conditions on the parameters and latent discrete observations. These assumptions are sufficient but not necessary for $(\mathbf{s}^*,\theta^*)$-identifiability of \eqref{eqn:canonical-dcm}.

The following assumptions concern the \textit{experimental design}:  
\begin{enumerate}[label=(A\arabic*), ref=(A\arabic*)]
    \setlength{\itemsep}{0em}
    \item\label{DCM-A1} The fMRI experiment utilizes a block design, and each block contains at least $d+1$ equally spaced observations following its initial within-block state $\mathbf{z}(t^{(b)})$. 
    \item\label{DCM-A2} The experimental design contains at least $m+1$ blocks
$b_1,\dots,b_{m+1}$ satisfying Assumption \ref{DCM-A1}, where $\mathcal{B}^* := \{b_1,\dots,b_{m+1}\} \subseteq \{0,\dots,B\}$. Moreover, the corresponding
stimulus combinations are unconfounded in the sense that
$$
\Tilde{\mathbf U}^* =
\begin{bmatrix}
1 & \widetilde{\mathbf u}^{(b_1)\top}\\
1 & \widetilde{\mathbf u}^{(b_2)\top}\\
\vdots & \vdots\\
1 & \widetilde{\mathbf u}^{(b_{m+1})\top}
\end{bmatrix}
\in \mathbb R^{(m+1)\times(m+1)}
$$
is invertible.
\end{enumerate}

Essentially, Assumptions \ref{DCM-A1}--\ref{DCM-A2} require that the fMRI experiment contains sufficiently long blocks to enable within-block identification of the latent dynamics, and that the experimental design includes sufficiently varied combinations of stimuli across blocks. In particular, Assumption \ref{DCM-A2} ensures that at least $m+1$ blocks are available whose stimulus configurations are linearly independent (after augmentation with an intercept), which prevents confounding between the baseline connectivity $\mathbf A$ and the modulatory effects $\mathbf B_1,\dots,\mathbf B_m$.

Regarding Assumption \ref{DCM-A1}, an fMRI block design is the natural choice, as the subject is repeatedly experiencing the same stimulus for a block of observations at a time, and each block is typically several TRs long, so \ref{DCM-A1} is easily satisfied. For DCM, it is recommended to keep the number of ROIs in the hypothesis fewer than 10, \cite{smith-etal-2013} meaning at worst, the block should be at least 10 TRs long, or about 20--30 s to satisfy Assumption \ref{DCM-A1}. In Assumption \ref{DCM-A2}, we require sufficiently varied combinations of experimental stimuli to prevent confounding effects. As a practical guideline, the total number of observations should be large enough to include at least $m+1$ blocks satisfying Assumption \ref{DCM-A1} (e.g., on the order of $(m+1)(d+1)$ observations). For example, even with a moderately complex model with $d=9$ and $m=3$, only a modest number of blocks and observations are needed, and both are easily achievable in typical fMRI experiments.

The remaining assumptions concern the \textit{parameters and latent discrete observations}:
\begin{enumerate}[label=(A\arabic*), ref=(A\arabic*)]
    \setcounter{enumi}{2}
    \setlength{\itemsep}{0em}
    \item\label{DCM-A3} For each $b \in \mathcal{B}^*$, $\Tilde{\mathbf{A}}^{(b)} = (\mathbf{A} + \sum_{i=1}^m \Tilde{u}_i^{(b)}\mathbf{B}_i) \in \mathbb{R}^{d \times d}$ has $d$ distinct, real eigenvalues.
    \item\label{DCM-A4} For each $b \in \mathcal{B}^*$, the data matrix $\mathbf{X}_0^{(b)} = [\mathbf{x}_0^{(b)},\dots,\mathbf{x}_{d}^{(b)}] \in \mathbb{R}^{(d+1) \times (d+1)}$ is invertible, where $\mathbf{x}_j^{(b)} := [\mathbf{z}(t^{(b)} + t_j),1]^\top=[\mathbf{v}^{(b)}_j,1]^\top$ for $j \in \{0,\dots,d\}$ are the first $d+1$ discrete observations from block $b$.
\end{enumerate}

These assumptions are made to ensure the identifiability of a linear (affine) ODE system within each block. The identifiability of such a system is equivalently characterized by whether the system's trajectory is constrained to a proper (affine) subspace space of $\mathbb{R}^d$.\cite{stanhope-etal-2014, duan-etal-2020} Conditions \ref{DCM-A3}--\ref{DCM-A4} guarantee that no such confinement occurs for the continuous trajectory, and discrete observations are made in a non-degenerative way, so that the the trajectory evaluated at discrete time points continues to span the full space. Similar assumptions are made in Duan et al, Stanhope et al, and Wang et al.\cite{stanhope-etal-2014, duan-etal-2020,wang-etal-2024}

\begin{theorem}
   \label{thm:canonical-dcm} 
Assume \ref{DCM-A1}--\ref{DCM-A4}. Then the ODE system in (\ref{eqn:canonical-dcm}) is $(\mathbf{s}^*,\theta^*)$-identifiable.
\end{theorem}

The detailed proofs for the theoretical results are given in Supplementary Material C. We leverage Proposition \ref{prop:dcm_sol_trajectory}, together with results from Duan et al.\cite{duan-etal-2020}, to prove Theorem \ref{thm:canonical-dcm}. In summary, the proof proceeds in three steps. First, for each block $b \in \mathcal B^*$, we show that the block-specific parameters $\Tilde{\mathbf A}^{(b)}$ and $\Tilde{\mathbf c}^{(b)}$ are uniquely determined by a finite number of within-block observations. Second, using Assumption \ref{DCM-A2}, we select $m+1$ blocks whose stimulus configurations are linearly independent, and express the parameters $\mathbf A$, $\mathbf B_1,\dots,\mathbf B_m$, and $\mathbf C$ as solutions to a linear system involving the block-specific parameters. Finally, the initial condition $\mathbf{s}^*$ is recovered using the closed-form affine solution over $[0,r)$ under the initial stimulus setting. For all empirical analyses, we assume that no external stimuli are present during the interval $[0,r)$ prior to the first observation, but the initial system condition is identifiable regardless of if stimuli are present in $[0,r)$. Together, these steps show that all parameters $(\mathbf{s}^*,\theta^*)$ are uniquely determined by the discrete trajectory, establishing identifiability according to Definition \ref{def:identifiable}.

Theorem \ref{thm:canonical-dcm} provides a set of sufficient conditions for the $(\mathbf{s}^*,\theta^*)$-identifiability of CDCM parameters relating to latent neural trajectory $\mathbf{z}(t;\mathbf{s}^*,\theta^*)$. We also present the following theorem for the injective mapping from discrete $\mathbf{z}(t_j)$ to the observable mean $\mathbb{E}[\mathbf{y}(t_j)] =\mu(t_j) + \beta$, which gives the important translation from $(\mathbf{s}^*,\theta^*)$-identifiability in $\mathbf{z}(t_j)$ to identifiability from the observed BOLD signal $\mathbf{y}(t_j)$ for $j \in \{1,\hdots,n\}$. 

\begin{theorem}
    \label{thm:one-to-one}
If $h(r) \neq 0$, that is, the canonical HRF evaluated at the TR $r$ is nonzero, then the mapping
$$
\{\mathbf{z}(t_j); j=0,\dots,n-1\} \;\mapsto\; \{\mathbb{E}[\mathbf{y}(t_j)]; j \in \{1,\hdots,n\}\}
$$
with $\mathbb{E}[\mathbf{y}(t_j)] = \mu(t_j) + \beta$ is injective. In particular, the system defined by \eqref{eqn:canonical-dcm} and \eqref{eqn:canonical-dcm-obs} is $(\mathbf{s}^*,\theta^*)$-identifiable from $\{\mathbb{E}[\mathbf{y}(t_1)],\dots,\mathbb{E}[\mathbf{y}(t_n)]\}$ in the sense that for $\mathbf{s}^\prime \in \mathbb{R}^d$ and $\theta^\prime = \{\mathbf{A}^\prime,\mathbf{B}_1^\prime,\hdots,\mathbf{B}_m^\prime,\mathbf{C}^\prime\}$ with $(\mathbf{s}^\prime,\theta^\prime) \neq (\mathbf{s}^*,\theta^*)$, there exists $j \in \{1,\hdots,n\}$ such that
$$
\mathbb{E}[\mathbf{y}(t_j;\mathbf{s}^\prime,\theta^\prime)] \neq \mathbb{E}[\mathbf{y}(t_j;\mathbf{s}^*,\theta^*)].
$$
As $\beta = \mathbb{E}[\mathbf{y}(t_j)] - \mu(t_j)$ for any $j$, the intercept $\beta$ is uniquely determined once $\mu(t_j)$ is identified.
\end{theorem}

We prove Theorem \ref{thm:one-to-one} by using the definition of discrete time convolution \eqref{eqn:discrete-convolve} and explicitly writing out the terms of $\mu(t_j)$. The discrete convolution can be represented in matrix form, and the matrix form of the terms involving the canonical HRF reduces to a lower triangular matrix and thus, is invertible. Matrix invertibility provides a one-to-one mapping from discrete dynamics $\{\mathbf{z}(t_j);j=0,\dots,n-1\}$ to $\{\mu(t_j); j \in \{1,\hdots,n\}\}$. Since $\mathbb{E}[\mathbf{y}(t_j)] = \mu(t_j) + \beta$ with $\beta$ constant across time, the addition of $\beta$ preserves injectivity and is itself identifiable from the mean trajectory. Therefore, the mapping from $\{\mathbf{z}(t_j);j=0,\dots,n-1\}$ to $\{\mathbb{E}[\mathbf{y}(t_j)]; j \in \{1,\hdots,n\}\}$ remains injective. For a standard TR (2--3 s), the nonzero condition for the HRF evaluated at the TR for Theorem \ref{thm:one-to-one} is satisfied. Therefore, if Assumptions \ref{DCM-A1}--\ref{DCM-A4} are met, along with a standard choice of TR, the parameters of CDCM are $(\mathbf{s}^*,\theta^*)$-identifiable from $\{\mathbb{E}[\mathbf{y}(t_j;\mathbf{s}^*,\theta^*)];j=1,\dots,n\}$.

\section{Simulation Studies}\label{sec:simulation}

We conduct two simulation studies to evaluate the practical utility of CDCM, with the aims of assessing computational efficiency as model complexity increases, and evaluating robustness to model misspecification. The first simulation in Section~\ref{sec:comp-sim} examines parameter recovery and computational efficiency as the number of ROIs increases, with particular emphasis on quantifying the computational gains from the piecewise analytic solution relative to numerical integration. The second simulation in Section~\ref{sec:balloon-sim} investigates performance under three variations of model misspecification, including when the hemodynamic data-generating process uses the Balloon model versus the canonical HRF used in CDCM. 

\subsection{Computational Efficiency}\label{sec:comp-sim}

In this subsection, we empirically investigate how CDCM scales with the number of ROIs, $d$, and quantify the computational gains of utilizing the piecewise analytic solution of the neural ODE (Proposition \ref{prop:dcm_sol_trajectory}) by comparing sampling time where the ODE is solved analytically versus numerically. Among the available solvers,\cite{stan-2024} the Cash-Karp Runge-Kutta (CKRK) method is the fastest for our setting and is used as the baseline for comparison.

\begin{figure}[t!]
    \centering
    \includegraphics[scale=0.45]{Figures/Fales_Fig3.png}
    \caption{\protect\input{Figure-Captions/sim-model}}
    \label{fig:sim-model}
\end{figure}

We consider two experimental settings: a simple and a complex model, both under the same block design. The experiment alternates between two stimuli (U1, U2) and rest in 10-TR blocks, with a TR of 2 s and a total of 150 scans. The design satisfies Assumptions \ref{DCM-A1}--\ref{DCM-A2}, with the remaining assumptions verified after specifying the model hypotheses. The corresponding state-space representations and true parameter values (aside from ROI self-loops) are shown in Figure \ref{fig:sim-model}; the complex model consists of three copies of the simple model chained together. The values in Figure \ref{fig:sim-model} define $\mathbf{A}$, $\mathbf{B}$, and $\mathbf{C}$, along with the reparameterized $\text{diag}(\mathbf{A}) = (-0.55) \cdot \mathbf{1}_d$. The initial condition is set to $\mathbf{s}^* = 0.1 \cdot \mathbf{1}_d$, reflecting nonzero baseline activity. Given these specifications, we verify that Assumptions \ref{DCM-A3}--\ref{DCM-A4} are satisfied.

Neural signals are simulated from the models, convolved with the canonical HRF, and Gaussian noise is added to match the estimated signal-to-noise ratio (SNR) of the attention to visual motion dataset (more discussion on this dataset in Section \ref{sec:motion}).\cite{buchel-friston-1997} To estimate the SNR from the attention to motion data, we fit a cubic smoothing spline to each ROI time series and treat the fitted values as the underlying signal. The regional SNR is estimated as the ratio of the standard deviation of the fitted signal to that of the residuals. Averaging SNR across regions, we obtain an SNR of $1.68$, which we use to set the noise variance $\sigma_\ell^2$ in each simulated region $\ell=1,\dots,d$ as
\begin{equation*}
    \sigma_\ell^2 =
   \frac{s^2_{\text{Truth},\ell}}{\text{SNR}^2} =
    \frac{s^2_{\text{Truth},\ell}}{1.68^2}, 
\end{equation*}
where $s^2_{\text{Truth},\ell}$ is the sample variance of $\mathbf{y}_{\text{Truth},\ell}  = (\mu_\ell(t_j))_{j=1}^n \in \mathbb{R}^n$, since we set $\beta = \mathbf{0}_d$ in the simulations.

We use CDCM to estimate network parameters for 5,000 warm-up iterations and 3,000 sampling iterations for each setting. For the simulation, the stopping criterion is the fixed number of iterations, not the multivariate ESS. The results for each scenario are averaged across 50 repeated simulations shown in Table \ref{tab:comp-sim}, with computation time in minutes. 

In assessing computational efficiency of the piecewise analytic solution compared to numerical integration, we find that the average computation time for the simple model using the piecewise analytic solution is 1.62 minutes, approximately twenty-two times faster than the CKRK solver, taking 35.46 minutes on average. For the complex model, the analytic solution finishes sampling in 4.4 hours on average, while remaining approximately four times faster than numerical integration despite the increased model complexity. In contrast, one of the simulations for the CKRK solver hit the 48-hour cluster wall-time at 80 percent completion, and the results in Table \ref{tab:comp-sim} are averaged across the remaining 49 simulations. Overall, computation time increases with the number of ROIs, $d$, for both approaches, although the piecewise analytic solution consistently provides substantial computational gains relative to numerical integration. The main computational burden comes from evaluating the likelihood, which involves solving the neural ODE at each iteration. Solving the ODE requires computing a matrix exponential within each time block, which scales at $O(d^3)$.\cite{moler-vanloan-2003}

In practice, DCM is typically applied to small, hypothesis-driven networks, often involving only a few regions,\cite{zeidman-etal-2019} and effective connectivity analyses are rarely considered for networks beyond eight to ten regions due to computational limitations.\cite{smith-etal-2013} Moreover, overly complex models are discouraged, as they can be difficult to estimate and interpret reliably.\cite{stephan-etal-2010} Therefore, the upper limit of $d=6$ in the current simulation is aligned with common practice, and while CDCM can in principle be applied to somewhat larger networks, the associated computational cost increases substantially and such models may be difficult to justify scientifically.

\begin{table}[t!]
\centering
\caption{Mean computation time (in minutes) across 50 simulations comparing the piecewise analytic solution to the ODE to the CKRK numeric solver for the simple and complex models (5,000 warm-up and 3,000 sampling iterations). Also included is the number of total parameters and neural parameters $(\theta_z)$ being estimated for each case, neural parameter coverage (the proportion of neural parameters captured within the HPD intervals on average), and average HPD interval length for a nominal coverage of 0.95. Note that for the complex model, one simulation replicate for the CKRK solver hit the wall-time of 48 hours prior to finishing sampling, and results shown are averaged across the 49 completed simulations.}
\resizebox{\textwidth}{!}{%
\begin{tabular}{cccccccc}
\hline
\textbf{\begin{tabular}[c]{@{}c@{}}ODE \\ Solution/Solver\end{tabular}} & \textbf{\begin{tabular}[c]{@{}c@{}}Total\\ Param.\end{tabular}} & \textbf{\begin{tabular}[c]{@{}c@{}}Neural\\ Param.\end{tabular}} & \textbf{\begin{tabular}[c]{@{}c@{}}Warm-up\\ (minutes)\end{tabular}} & \textbf{\begin{tabular}[c]{@{}c@{}}Sampling\\ (minutes)\end{tabular}} & \textbf{\begin{tabular}[c]{@{}c@{}}Total\\ (minutes)\end{tabular}} & \textbf{\begin{tabular}[c]{@{}c@{}}Neural\\ Parameter\\ Coverage\end{tabular}} & \textbf{\begin{tabular}[c]{@{}c@{}}Neural\\ Avg. Interval\\ Length\end{tabular}} \\ \hline
\multicolumn{8}{c}{\textbf{Simple Model}} \\ \hline
Piecewise Analytic & 12 & 6 & 1.10 & 0.52 & 1.62 & 0.967 & 0.353 \\
CKRK Numeric & 12 & 6 & 24.37 & 11.09 & 35.46 & 0.973 & 0.352 \\ \hline
\multicolumn{8}{c}{\textbf{Complex Model}} \\ \hline
Piecewise Analytic & 40 & 22 & 190.92 & 70.92 & 261.85 & 0.977 & 0.558 \\
CKRK Numeric & 40 & 22 & 819.21 & 339.69 & 1158.90 & 0.977 & 0.558 \\ \hline
\end{tabular}}
\label{tab:comp-sim}
\end{table}

Table \ref{tab:comp-sim} also shows the average neural parameter coverage, or the proportion of true neural parameters captured within the 95\% highest posterior density (HPD) intervals on average, and the average HPD interval length across parameters and chains. These results validate CDCM's ability to recover true parameters from noisy observation under the correctly specified model. 
For both models, coverage is close to the nominal 0.95, with wider intervals for the complex model. 
We note that the HPD intervals in Table~\ref{tab:comp-sim} differ slightly between the analytic and numeric solutions. Although the two solutions to the neural ODE agree up to a small numerical tolerance for any given parameter, these small differences in the likelihood, through repeated evaluations, lead to slightly different proposal trajectories and thus minor differences in the resulting posterior distributions.

\subsection{Robustness to Misspecification}\label{sec:balloon-sim}

In Section \ref{sec:comp-sim}, we show the performance of CDCM under a correctly specified model. In this subsection, we further examine parametric estimation performance under various types of model misspecification. 
In particular, we assess CDCM's robustness to misspecification in comparison to the one-state, bilinear DCM of Friston et al,\cite{friston-etal-2003} implemented in SPM. The simulation varies three types of misspecification: 1) observation model, 2) diagonal parameterization, and 3) hemodynamic delay.

Of the three types of misspecification we consider here, observation model misspecification is the most important, as the true hemodynamic data-generating process in fMRI is unknown. To study the impact of observation model misspecification, we generate data from two mechanisms: the nonlinear Balloon model used in SPM and the simplified observation model used in CDCM. This setting allows us to assess how CDCM compares to a more complex representation of the hemodynamic model. For the Balloon model, SPM internally represents the hemodynamic parameters on the log scale to enforce positivity, and we use the values $\theta_y = \{\log(\kappa) = -0.2, \log(\tau_h) = [-0.2, -0.3]^\top, \log(\varepsilon_h) = 0.15\}$. These correspond to the natural-scale parameters $\kappa$, the decay rate of the vasoactive signal, $\tau_h \in \mathbb{R}^d$ ($d=2$, as we simulate from a two-ROI system), the mean transit time (s), representing the average time blood spends in the venous compartment, and $\varepsilon_h$, the ratio of intra- to extravascular signal contributions. For more information on the Balloon model, we refer the reader to Appendix 5 of Zeidman et al.\cite{zeidman-etal-2019} Just as in CDCM, SPM also estimates the intercept vector $\beta \in \mathbb{R}^d$, and for the simulation, we set $\beta = \mathbf{0}_d$ for both observation models. For CDCM's observation model, as mentioned in Section \ref{sec:mdl}, all hemodynamic parameters in the canonical HRF are fixed values and there are no additional parameters to specify. 

Beyond observation model misspecification, we consider two other forms of misspecification, the first of which being diagonal parameterization. As mentioned in Section \ref{sec:mdl} (and elaborated on in Supplementary Material B), we adopt a different diagonal parameterization strategy in CDCM than SPM, where CDCM's parameterization concerns only the diagonal entries of $\mathbf{A}$ with a tight prior on the diagonal of $\mathbf{B}$, but SPM's parameterizes the diagonal of $\mathbf{A} + \sum_{i=1}^m u_i(t)\mathbf{B}_i$ directly. Both approaches are similar in that they are designed to regulate system dynamics from diverging, but they do imply a difference in neural ODE model structure when modulation of self-loops is part of the state-space hypothesis. Therefore, it is worth investigating the impact of these two parameterization choices on parametric estimation performance. We consider two neural parameter settings reflecting differences between CDCM and SPM. In the first, only $\text{diag}(\mathbf{A})$ is parameterized, meaning the neural models are equivalent; in the second setting, we additionally include a modulatory effect in $\text{diag}(\mathbf{B})$, implying misspecification under the opposing estimation procedure. Both settings use the same experimental design as in Section \ref{sec:comp-sim}. For the baseline case, the (untransformed) neural parameters are
\begin{equation*}
\mathbf{A} =
\begin{bmatrix}
-0.1 & 0.3 \\
0.4 & 0.15
\end{bmatrix}, \quad
\mathbf{B}_1 =
\begin{bmatrix}
0 & 0 \\
0 & 0
\end{bmatrix}, \quad
\mathbf{B}_2 =
\begin{bmatrix}
0 & -0.2 \\
0 & 0
\end{bmatrix}, \quad
\mathbf{C} =
\begin{bmatrix}
0.7 & 0 \\
0 & 0
\end{bmatrix}.
\end{equation*}
In the second setting, we additionally set $(\mathbf{B}_2)_{22} = 0.05$.  

The third type of misspecification we consider is the type of differential equation used in the model. SPM's DCM formulation uses delay differential equations to account for the time required for signals to travel between states, and defaults to fixed delays of $\text{TR}/2$ seconds.\cite{friston-etal-2003} CDCM does not use delay differential equations, so the two different differential equation approaches allow us to assess the impact of misspecified hemodynamic delays on parameter estimation performance. All together, the types of misspecification yield six combinations across observation model, diagonal parameterization, and delay inclusion. 

For each combination, we simulate neural signals, generate BOLD responses under the specified observation model, and add Gaussian noise targeting an SNR of 1.68, as in the first simulation. Due to differences in observation models, signals generated under the Balloon model have larger amplitudes than those from CDCM, and adding noise at the target variance can produce simulated BOLD series with a range exceeding four units. Since SPM internally rescales inputs exceeding this threshold,\cite{zeidman-etal-2019} we avoid the additional rescaling by scaling the additive noise when necessary. Specifically, after drawing Gaussian noise at the target variance, we check whether the resulting simulated BOLD series exceeds a range of four units. If so, we multiply the sampled noise by the largest constant that keeps the simulated series within the acceptable range. Chain length for CDCM is determined by multivariate ESS convergence for $\alpha_{ESS} = \varepsilon_{ESS} = 0.05$, and for every setting, we conduct 50 repeated simulations.

\begin{table}[t!]
\centering
\caption{Simulation results between CDCM and SPM for each type of misspecification, where the first six rows indicate model misspecification, and the next six rows are the correctly specified models. Results are averaged across 50 simulations. We generate data using both SPM's and CDCM's data-generating model (DGM), incorporating diagonal parameters in $\mathbf{A}$ only, or both in $\mathbf{A}$ and $\mathbf{B}$, and for SPM's DGM, conduct the simulation both assuming no stimulus delays between states, and 1 second delays. Results shown include the parameter coverage, or the proportion of true parameters captured within the 95\% HPD interval, the average HPD interval length across parameters, as well the MSE of the predicted BOLD signal vs. the true simulated signal for each condition. The predicted curve is generated using the posterior mean estimates.}
\resizebox{\textwidth}{!}{%
\begin{tabular}{ccccccc}
\hline
\textbf{\begin{tabular}[c]{@{}c@{}}Estimation\\ Procedure\end{tabular}} & \textbf{DGM} & \textbf{Diag.} & \multicolumn{1}{c|}{\textbf{Delays}} & \textbf{\begin{tabular}[c]{@{}c@{}}Parameter\\ Coverage\end{tabular}} & \textbf{\begin{tabular}[c]{@{}c@{}}Avg. Interval \\ Length\end{tabular}} & \textbf{\begin{tabular}[c]{@{}c@{}}MSE of Pred vs.\\ True Curve\end{tabular}} \\ \hline
\multicolumn{7}{c}{\textbf{Misspecified Models}} \\ \hline
CDCM & SPM & \textbf{A} & \multicolumn{1}{c|}{None} & 0.750 & 0.236 & 0.0360 \\
CDCM & SPM & \textbf{A} \& \textbf{B} & \multicolumn{1}{c|}{None} & 0.643 & 0.242 & 0.0358 \\
CDCM & SPM & \textbf{A} & \multicolumn{1}{c|}{1 s} & 0.750 & 0.236 & 0.0360 \\
CDCM & SPM & \textbf{A} \& \textbf{B} & \multicolumn{1}{c|}{1 s} & 0.643 & 0.242 & 0.0358 \\
SPM & CDCM & \textbf{A} & \multicolumn{1}{c|}{None} & 0.490 & 0.237 & 0.0021 \\
SPM & CDCM & \textbf{A} \& \textbf{B} & \multicolumn{1}{c|}{None} & 0.357 & 0.212 & 0.0021 \\ \hline
\multicolumn{7}{c}{\textbf{Correctly Specified Models}} \\ \hline
SPM & SPM & \textbf{A} & \multicolumn{1}{c|}{None} & 0.703 & 0.200 & 0.0019 \\
SPM & SPM & \textbf{A} \& \textbf{B} & \multicolumn{1}{c|}{None} & 0.720 & 0.182 & 0.0019 \\
SPM & SPM & \textbf{A} & \multicolumn{1}{c|}{1 s} & 0.703 & 0.200 & 0.0019 \\
SPM & SPM & \textbf{A} \& \textbf{B} & \multicolumn{1}{c|}{1 s} & 0.720 & 0.182 & 0.0019 \\
CDCM & CDCM & \textbf{A} & \multicolumn{1}{c|}{None} & 0.960 & 0.311 & 0.0061 \\
CDCM & CDCM & \textbf{A} \& \textbf{B} & \multicolumn{1}{c|}{None} & 0.966 & 0.313 & 0.0058 \\ \hline
\end{tabular}}
\label{tab:Balloon-sim}
\end{table}

We present the results, averaged across the 50 simulations, in Table \ref{tab:Balloon-sim}. In terms of parameter coverage, or the proportion of true neural parameters captured within the 95\% HPD intervals, there is a drop in coverage for both CDCM and SPM under observation model misspecification. However, CDCM's performance declines less than that of SPM's, and even under correct model specification, SPM's intervals undercover compared to the nominal coverage of 0.95. The average interval length for SPM is shorter than CDCM in every case, indicating that posterior uncertainty may be underestimated. When CDCM is estimating a correctly specified model, average interval length is wider and overcovering parameters; the average length is larger than those of the misspecified observation model, but a likely explanation is the SNR difference between the two sets of simulated data. Also, we do not observe an impact of signal delay specification on the simulation results.

We also use the estimated posterior means to construct the predicted BOLD signal curves and compute the mean-squared error (MSE) between the predicted and true signals in Table \ref{tab:Balloon-sim}. In all cases, the MSE is small, indicating the models are performing well, even when parameter coverage is lower for the misspecified variations. We observe that SPM has a smaller MSE on average across all combinations compared to CDCM. This is not surprising, given that SPM has four free parameters in the hemodynamic portion of the observation model compared to zero in CDCM.

\begin{figure}[t!]
    \centering
    \includegraphics[width=0.8\linewidth]{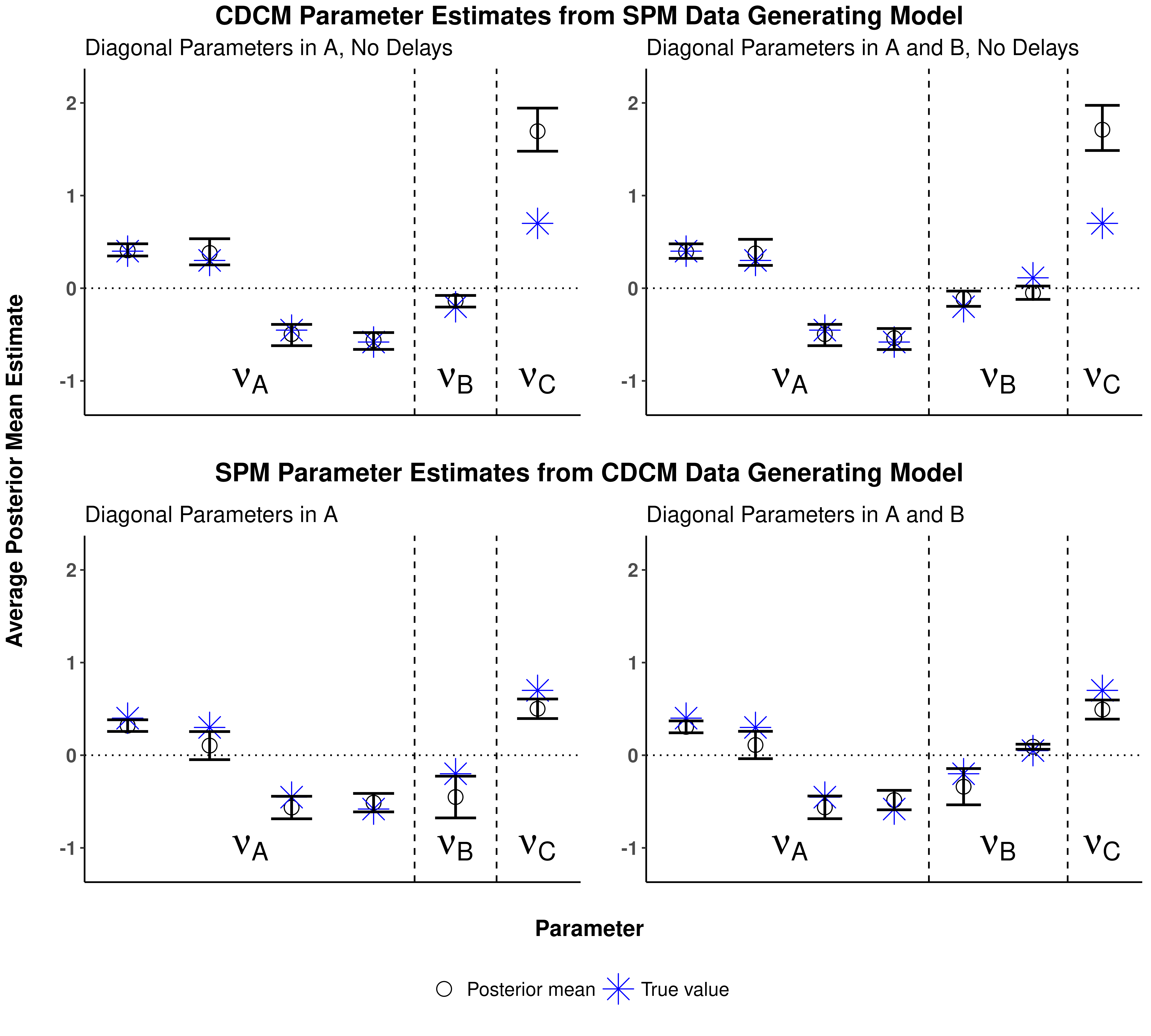}
    \caption{\protect\input{Figure-Captions/misspec-result}}
    \label{fig:misspec-result}
\end{figure}

To better understand if there are patterns across the parameters for which coverage drops, we perform parameter-wise comparisons for the misspecified models. In Figure \ref{fig:misspec-result}, we plot the true values, posterior mean estimates, and 95\% HPD intervals for the neural parameters. Across conditions for the correctly specified model results, as well as the 1 s delay variations, we observe qualitatively similar results, therefore, we defer the other comparison figures to Supplementary Material F. 

Since CDCM and SPM both reparameterize the diagonal elements of $\mathbf{A}$, and SPM additionally reparameterizes $\mathbf{B}$, they are transformed back to their original scale for direct comparison. For CDCM, we transform posterior draws $\hat\nu_{A,\ell,s}^D$ for the $\ell$th diagonal of $\mathbf{A}$ via
$$
a_{\ell,s} = -0.5 \exp(\hat\nu_{A,\ell,s}^D),
$$
where $\hat\nu_{A,\ell,s}^D$ denotes posterior draw $s \in \{1,\dots,N\}$ on the unconstrained scale, and $N$ is the number of sampling draws.
For SPM, variational Laplace returns a Gaussian approximation to the posterior of the reparameterized diagonal elements of $\mathbf{A}$ and $\mathbf{B}$.
Since the transformation from the reparameterized diagonal elements to the original diagonal elements of $\mathbf{A}$ and $\mathbf{B}_i$ is nonlinear, the posterior mean and variance on the original scale cannot be obtained by simply transforming the Gaussian posterior mean and variance. Therefore, we additionally generate $N = 20{,}000$ draws from the estimated posterior distributions. We denote the corresponding unconstrained draws by $\bar{a}_{\ell,s}$ and $\bar{b}_{\ell,s}^{(i)}$ for the diagonal elements of $\mathbf{A}$ and $\mathbf{B}_i$, $i = 1,\dots,m$. The transformed diagonal elements are computed as
$$
a_{\ell,s} = -0.5 \exp(\bar{a}_{\ell,s}), \qquad
b_{\ell,s}^{(i)} = -0.5\, a_{\ell,s} \left[\exp\big(\bar{b}_{\ell,s}^{(i)}\big) - 1\right].
$$
After transformation, posterior means and variances are summarized in the original scale of the diagonal elements of $\mathbf{A}$ and $\mathbf{B}$. Additional details on SPM's parameterization are provided in Supplementary Material B.

The top panels show CDCM’s performance when estimating data generated under SPM’s observation model. Across all conditions (Table \ref{tab:Balloon-sim}), CDCM accurately recovers parameters in $\mathbf{A}$, while recovery for $\mathbf{B}$ is more variable: the off-diagonal effect is covered, on average, whereas the diagonal effect (right panels) is not recovered, with intervals including zero. For $\mathbf{C}$, the posterior mean consistently has the correct sign but substantially overestimates the true magnitude. In contrast, the bottom panels show that SPM recovers the correct sign for all parameters, including the diagonal effect in $\mathbf{B}$, but exhibits poor interval coverage, with estimates often narrowly missing the true values. As with CDCM, the largest discrepancies occur in $\mathbf{C}$; however, SPM systematically underestimates its magnitude. The difference is likely driven by the underlying hemodynamic models, as Balloon-based signals exhibit larger amplitudes than those generated under CDCM, implying that scaling primarily impacts estimates in $\mathbf{C}$.

\section{Real Data Applications}\label{sec:application}

We consider two fMRI data applications to evaluate CDCM in relation to existing DCM implementations. Our goals are threefold: 1) compare the agreements and disagreements between CDCM and standard DCM implementations, 2) determine which connections are credibly different from zero based on posterior uncertainty across methods, and 3) examine the replicability of estimated connectivity patterns across subjects and modeling choices. The first application analyzes data from a single subject, providing a controlled setting to directly compare parameter estimates and uncertainty quantification between methods. The second revisits the motivating large-scale application introduced in Section \ref{sec:motivation}, allowing assessment of the stability and replicability of connectivity patterns across subjects, sessions, and experimental conditions. Jointly, the connectivity analyses of these data connect method comparison, uncertainty quantification, and replicability.

\subsection{Attention to Visual Motion Task}\label{sec:motion}

The attention to visual motion study from B\"{u}chel and Friston,\cite{buchel-friston-1997} with data provided by the Wellcome Centre for Human Neuroimaging, provides a single subject's fMRI data for a task involving three experimental stimuli. The hypothesis from the study is shown in Figure \ref{fig:motion-model}. There are three ROIs: V1, V5, and the superior parietal cortex (SPC). These ROIs respond to visual stimulation, with V1 sensitive to general input, V5 selective for motion, and SPC engaged during attention to motion. The Photic condition corresponds to viewing static dots, representing the most general level of visual stimulation. The Motion condition extends this by introducing moving dots, and the Attention condition further builds on Motion by requiring focused engagement. Thus, the conditions are hierarchically related, progressing from general visual stimulation (Photic), to motion processing (Motion), to attentional engagement (Attention). Assumptions \ref{DCM-A1}--\ref{DCM-A2} are satisfied for the design, but given these are real data, we cannot verify \ref{DCM-A3}--\ref{DCM-A4} directly. The TR is 3.22 seconds and there are 360 scans in the session. 

\begin{figure}[tb!]
    \centering
    \includegraphics[scale=0.65]{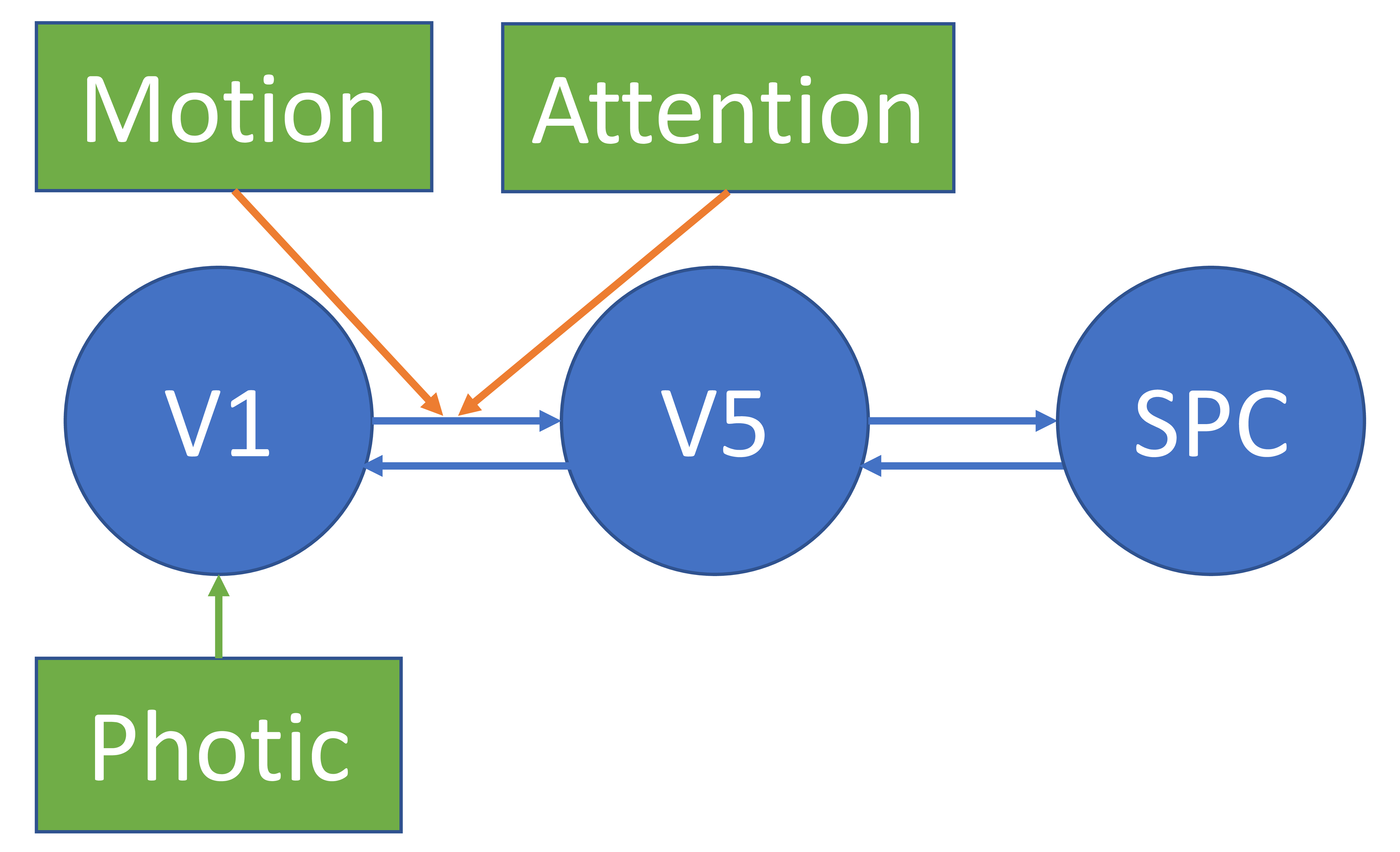}
    \caption{\protect\input{Figure-Captions/motion-model}}
    \label{fig:motion-model}
\end{figure}

Time series extraction is performed in SPM following the DCM tutorial procedures (SPM12 Manual, Chapter 36).\cite{spm12manual} For the single-subject application, voxels are selected based on peak task activation within 8 mm spheres around each ROI, defined as $p < 0.001$ (uncorrected) from a voxelwise task-activation general linear model (GLM). The first principal component of these voxels is used as the representative time course, and the resulting signals are scaled to ensure $\text{range}(\mathbf{Y}) \leq 4$ for numerical stability.

We compare CDCM against two existing DCM implementations, SPM and VBA,\cite{daunizeau-etal-2014} on the attention to visual motion data. Although the SPM12 software\cite{spm12manual} includes multiple DCM variations, when we refer to SPM, we are referring to the traditional, one-state, bilinear DCM from Friston et al.\cite{friston-etal-2003} VBA is a Matlab toolbox for variational Bayesian inversion of nonlinear state-space models, including DCM-style formulations. Although it uses the same bilinear neural model structure as SPM, it differs in prior specification, ODE numerical integration, and variational optimization scheme, and defaults to using a linearized variant of the Balloon observation model. 
For SPM and VBA under the default settings, both optimizations converge in approximately one minute. The CDCM results are from 8,000 MCMC iterations (after 5,000 warm-up iterations), taking 49.72 minutes total, and achieving a multivariate ESS of 9,313.23, corresponding to a precision of $\varepsilon_{ESS} = 0.048$ for 19 total parameters and $\alpha_{ESS} = 0.05$.

\begin{figure}[t!]
    \centering
    \includegraphics[width =.8\linewidth]{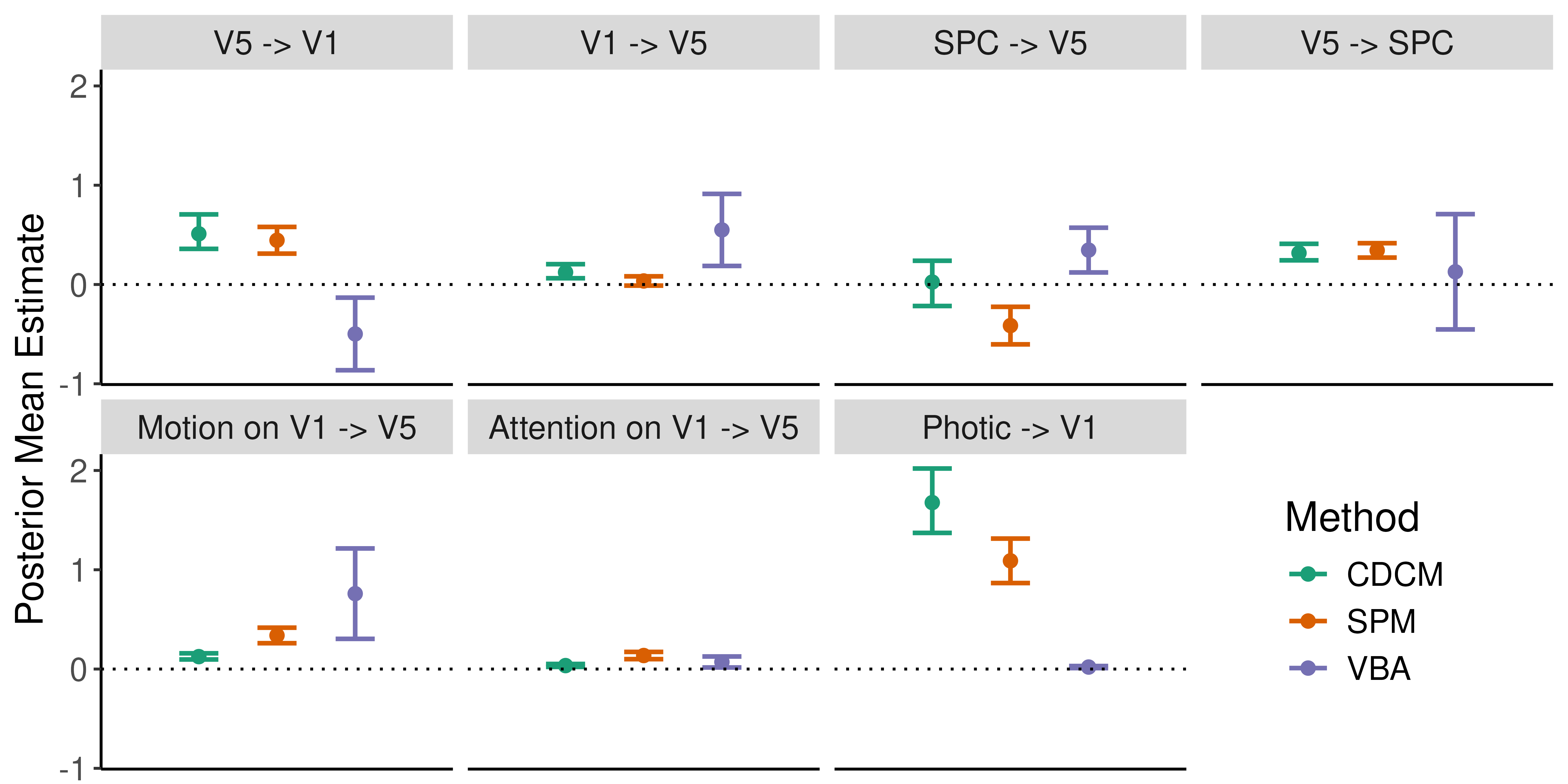}
    \caption{\protect\input{Figure-Captions/motion-results-comp}}
    \label{fig:motion-results-comp}
\end{figure}

Figure \ref{fig:motion-results-comp} shows posterior means and 95\% HPD intervals for seven of the ten neural parameters; the three diagonal elements of $\mathbf{A}$ are omitted as they are not of primary interest. For most parameters, the results from CDCM and SPM are in agreement. The only difference between CDCM and SPM is for the connection from SPC to V5, where SPM's interval is entirely negative, whereas CDCM's interval covers zero. Also, both intervals are strongly positive for the Photic stimulus to V1, and given the scaling differences between CDCM and SPM regarding the $\mathbf{C}$ matrix (as described in Section \ref{sec:balloon-sim}), it is unsurprising that CDCM estimates a larger value than SPM. 

In contrast, while VBA produces similar estimates for several connections, it differs more noticeably for others. The most prominent difference is for the direct photic input to V1, as it is estimated to be near zero under VBA, in contrast to the positive effects estimated by both CDCM and SPM. The discrepancy is unlikely to reflect an absence of stimulus-driven response, but rather a difference in how the VBA implementation allocates stimulus-related variance across model components. Although the input specification $\mathbf{U}$ is standard and identifiable, the direct input effect may be distributed differently across the neural and observation components under alternative parameterizations. Using the nonlinear hemodynamic option in VBA produces a similar estimate, suggesting that this behavior is not attributed to the choice of hemodynamic model. Finally, assuming the posterior means represent the ground truth, simulation using CDCM estimates indicates that Assumptions \ref{DCM-A3}--\ref{DCM-A4} are satisfied.

\begin{figure}[t!]
    \centering
    \includegraphics[width=0.8\linewidth]{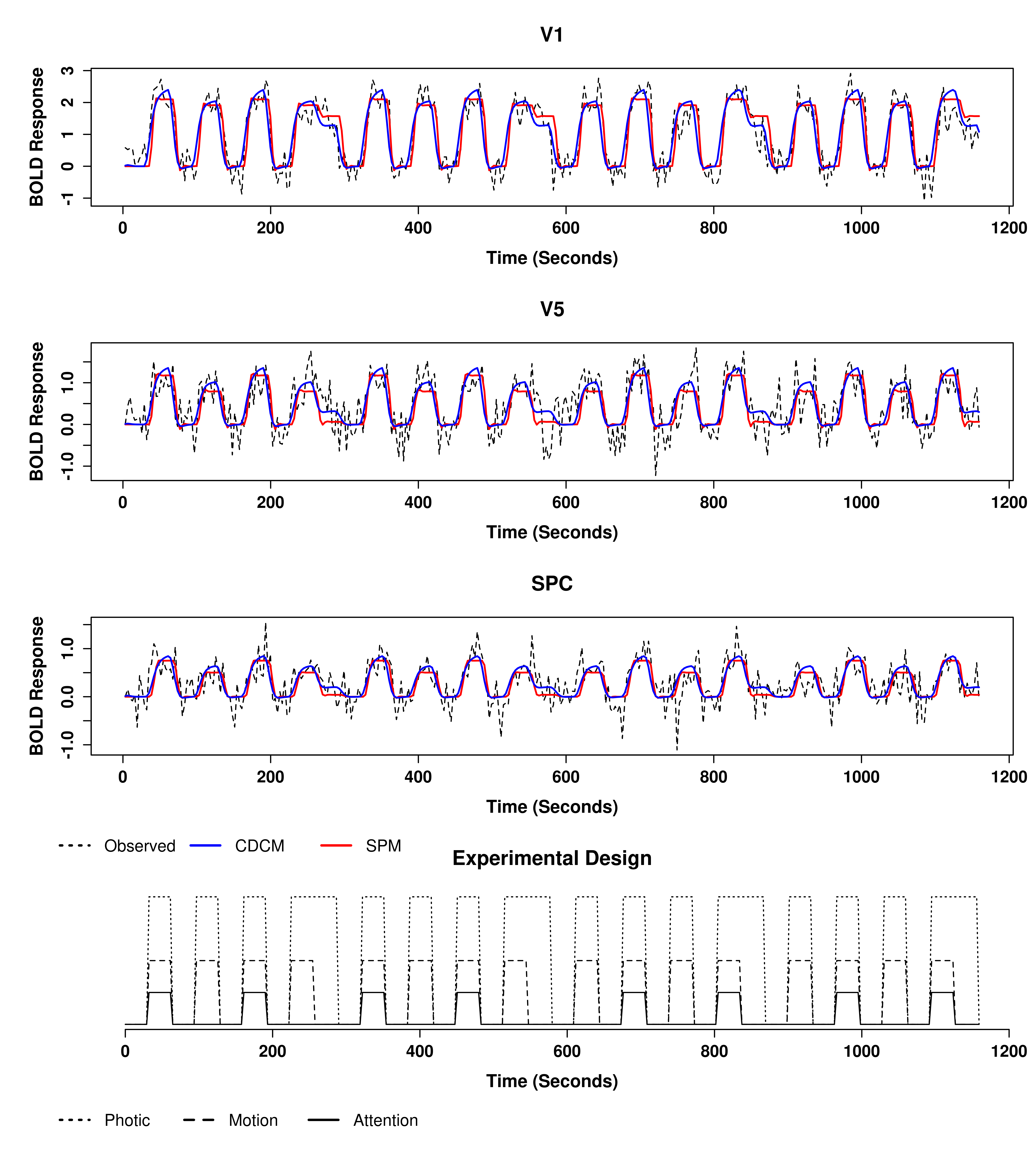}
    \caption{\protect\input{Figure-Captions/est-vs-obs}}
    \label{fig:est-vs-obs}
\end{figure}

Next, we compare the predicted versus the observed BOLD signal in each of the ROIs between CDCM and SPM only, as SPM is the more commonly used DCM implementation. The predicted curves for each method are generated using the posterior means shown in Figure \ref{fig:motion-results-comp}. The results are shown in Figure \ref{fig:est-vs-obs}, where the dotted curve shows the observed data, the blue line is the prediction from CDCM, and the red line is from SPM; both models fit the data well. 

\begin{table}[t!]
\centering
\caption{Mean squared error (MSE), block bootstrap estimated standard error (SE), and 95\% block bootstrap confidence intervals (CI) between the predicted and the observed BOLD response for CDCM and SPM in each ROI. The block bootstrap accounts for the temporal dependence within fMRI data (10 TR block length, 10,000 bootstrap replicates). For each ROI, the intervals between the two methods overlap, indicating that neither method is superior in estimating the data.} 
\begin{tabular}{l|cc|cc}
\hline
\multicolumn{1}{c|}{\textbf{ROI}} &
\multicolumn{1}{c}{\textbf{CDCM MSE (SE)}} &
\multicolumn{1}{c|}{\textbf{CDCM 95\% CI}} &
\multicolumn{1}{c}{\textbf{SPM MSE (SE)}} &
\multicolumn{1}{c}{\textbf{SPM 95\% CI}} \\ \hline
V1  & 0.232 (0.019) & (0.198, 0.272) & 0.199 (0.027) & (0.148, 0.253) \\
V5  & 0.177 (0.015) & (0.151, 0.210) & 0.178 (0.014) & (0.152, 0.207) \\
SPC & 0.092 (0.010) & (0.076, 0.113) & 0.090 (0.009) & (0.075, 0.108) \\ \hline
\end{tabular}
\label{tab:motion-mse}
\end{table}

In terms of MSE between the data and each of the predicted curves in each ROI, Table \ref{tab:motion-mse} provides the point estimates, as well as block bootstrap estimated standard errors and 95\% confidence intervals to provide a sense of error variation from one point on the curve to the next. For each ROI, the intervals overlap, indicating that neither method is superior in fitting the data. Overall, these results support CDCM as a reliable alternative to existing DCM approaches, achieving comparable estimates and predictive fit while enabling clearer assessment of uncertainty and stability.

\subsection{Human Connectome Project (HCP) Social Cognition Task}\label{sec:hcp}

For the second application, we revisit the motivating example from Section \ref{sec:motivation}, with the goals of assessing the stability and replicability of effective connectivity patterns across methods, subjects, sessions, and hemispheres for the social cognition task using the updated HCP-YA 2025 data. Using previous Q2 HCP data, Hillebrandt et al\cite{hillebrandt-etal-2014} find that the full model with all possible connections has the highest evidence (Figure \ref{fig:hillebrandt-dcm}), but only report posterior mean estimates for the $\mathbf{B}$ and $\mathbf{C}$ matrix connections and do not directly report posterior SD for any estimates; we use this full model as our state-space hypothesis and confirm, using the posterior uncertainty associated with the connectivity estimates, which connections are reliably supported, and assess if results replicate across sessions and hemispheres.

\subsubsection{Data and Analysis Setup}

In the HCP task implementation, participants view 20 s video clips of geometric shapes that either move randomly or interact in an animated manner. Following each clip, participants indicate whether the shapes exhibit an animated (mental) interaction, no interaction, or are unsure. Each session consists of alternating blocks of ``mental interaction'' and ``random interaction'' conditions, interleaved with 15 s fixation blocks.\cite{barch-etal-2013} Each participant has data from two task sessions (TR = 0.72 s, 2 mm isotropic resolution, 274 scans per session) one with right-to-left (RL) phase encoding, and the other with left-to-right (LR). We model all connections given in Figure \ref{fig:hillebrandt-dcm} and define the input matrix $\mathbf{U}$ using a hierarchical stimulus structure, similar to the structure we use in Section \ref{sec:motion}. The All Motion stimulus indicates when any video is active during the scan, and the Animate stimulus indicates the subset of animate motion videos.  Under the design, Assumptions \ref{DCM-A1}--\ref{DCM-A2} are satisfied.

Instead of selecting voxels for each ROI based on subject-specific task-activation thresholds, we extract the regional time courses using the group-level activation MNI coordinates as given by Hillebrandt et al,\cite{hillebrandt-etal-2014} and define ROI cubes of voxels around these coordinates. For V5, we use MNI $(44,-64,4)$ mm (right middle temporal gyrus) and $(-44,-74,4)$ mm (left middle occipital gyrus) for the right and left hemispheres, respectively. For the pSTS, we use MNI $(54,-50,16)$ mm (right middle temporal gyrus) and $(-56,-52,10)$ mm (left middle temporal gyrus). We define $14 \times 14 \times 14 \text{ mm} = 2744 \text{ mm}^3$ cubes, equivalent to 343 voxels. The regional time courses are extracted by taking the first principal component of all voxels in the ROI, flipping the sign as needed such that the average voxel loading is positive. The range of the representative BOLD time courses exceeds four units, so the data for each subject are scaled accordingly, just as in Section \ref{sec:motion}.

Effective connectivity is estimated using CDCM and SPM (one-state bilinear DCM) on all session $\times$ hemisphere combinations for every subject. The chains in CDCM use a warm-up of 5,000 iterations with a target NUTS acceptance probability of $0.9$, and multivariate ESS for  $\alpha_{ESS}  = \varepsilon_{ESS} = 0.05$ determines sampler convergence. All 400 chains are run until convergence, a maximum of 100,000 iterations, or a wall-time of 48 hours, whichever comes first. Of the 400 chains, 11 encounter either divergence or maximum tree-depth warnings and require additional sampler tuning, consisting of adjustments to the target acceptance probability and, in some cases, a longer warm-up period. After tuning, 393 chains achieve the multivariate ESS required for a precision of $\varepsilon_{ESS} = 0.05$, and the remaining seven achieve a slightly lower but still acceptable precision between 0.05 and 0.1. Thus, all chains are retained for analysis, and a table with the additional tuning requirements is provided in Supplementary Material F. As for SPM, 15 of the 400 VL optimizations do not converge under the default maximum number of iterations of 128, but all converge within 512 iterations.

Four separate group-level models, one for each session $\times$ hemisphere combination, are estimated for both methods using five parallel chains, each with 1,000 warm-up and 5,000 sampling iterations. We include participant age category (participants' exact age is restricted), gender, and the Penn Progressive Matrices (PMAT24\_A\_CR) score as covariates.\cite{vanessen-etal-2013} The PMAT score is the number of correct responses on the 24-item Penn Progressive Matrices behavioral assessment \cite{bilker-etal-2012} abbreviated by Van Essen et al\cite{vanessen-etal-2013} for the HCP. The PMAT assessment provides a measure of subjects' fluid cognitive intelligence. Age and gender are effect-coded and the PMAT score is standardized, making the interpretation of the group-level intercept coefficient estimates $\hat{\boldsymbol{\alpha}} \in \mathbb{R}^{p_{\theta_z}}$ in terms of the ``average subject''.

\subsubsection{Results}

\begin{figure}[t!]
    \centering
    \includegraphics[width=0.9\linewidth]{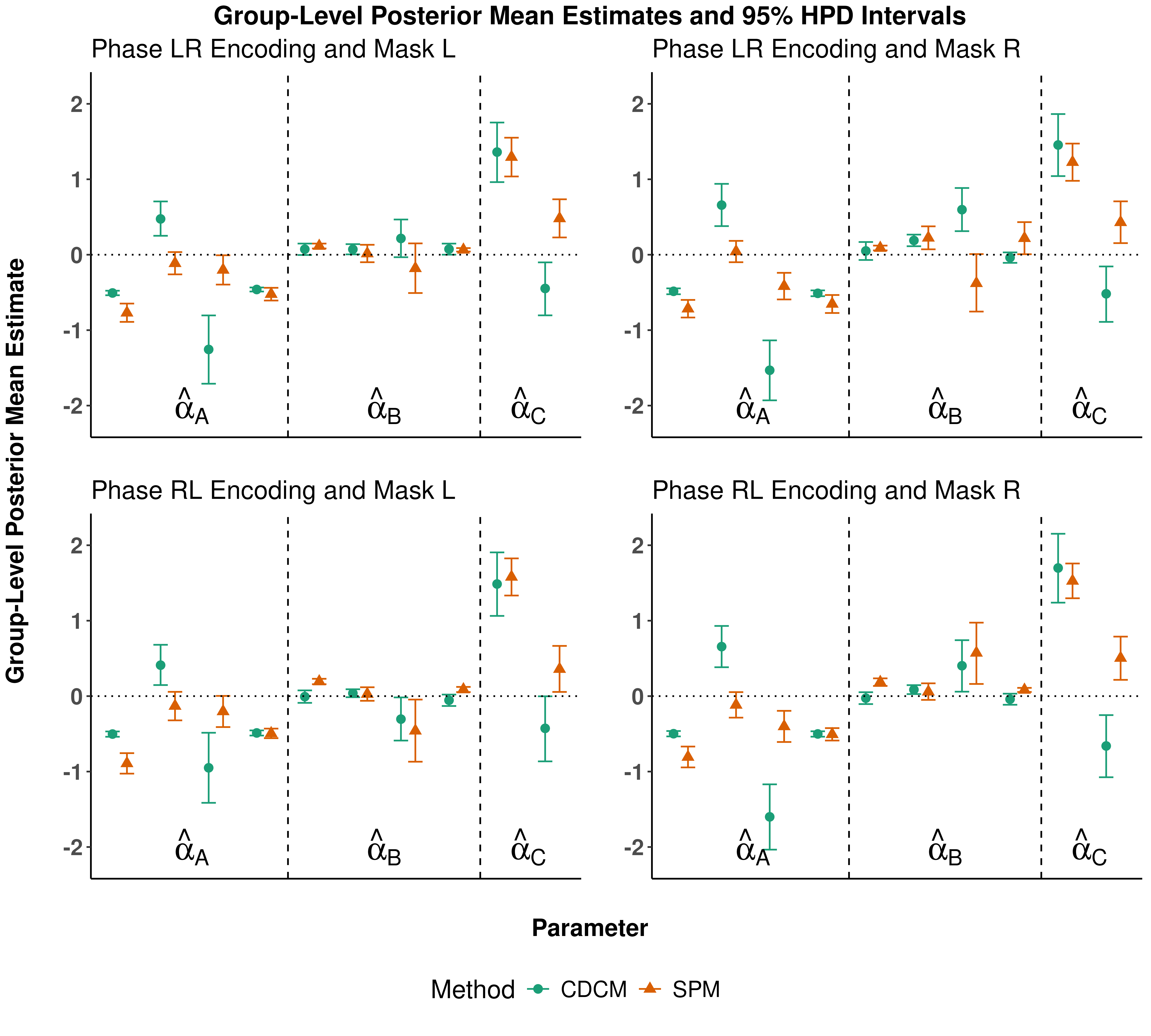}
    \caption{\protect\input{Figure-Captions/group-alpha-comparison}}
    \label{fig:group-alpha-comparison}
\end{figure}

We present group-level posterior means and 95\% HPD intervals for both methods in Figure \ref{fig:group-alpha-comparison}, with panels corresponding to session (LR, RL phase encoding) and hemisphere (mask L, R). Recall, the first goal of the empirical study is to compare agreements and disagreements between CDCM and standard DCM implementations, where greater agreement indicates more replicable connectivity patterns. Overall, CDCM and SPM’s one-state bilinear DCM show substantial agreement, particularly in connectivity direction, with only modest differences in magnitude (e.g., the third position of $\hat{\boldsymbol{\alpha}}_A$).

Of the ten neural parameters, three show disagreement or are inconclusive due to uncertainty. In the second position of $\hat{\boldsymbol{\alpha}}_A$, CDCM estimates a positive effect, whereas SPM’s interval includes zero. In the third position of $\hat{\boldsymbol{\alpha}}_B$, both methods are uncertain for LR phase encoding with mask L; for mask R, CDCM’s interval is entirely positive while SPM’s is mostly negative, though still inconclusive. Most notably, the methods disagree on the direction of the driving input from All Motion to pSTS (second position of $\hat{\boldsymbol{\alpha}}_C$), with CDCM estimating a negative effect and SPM a positive one. This is the only group-level sign disagreement for which both intervals exclude zero.

Given the posterior uncertainty in all other potential disagreements, the results regarding the first goal demonstrate one key disagreement between the two one-state methods. We compare CDCM and SPM further by investigating the observed discrepancies through a targeted simulation study using the group-level posterior means from each method as the ground truth. In summary, we find that both methods recover the true parameter signs reliably and produce estimates close to the truth across all simulation settings, although both approaches exhibit overconfident intervals with poor coverage. Performance is driven more by whether the assumed parameter values align with the estimation method than by whether the data-generating model is correctly specified. For the key parameter of interest, the driving input from All Motion to pSTS, CDCM more consistently recovers the correct sign, while SPM shows lower reliability. Overall, the results suggest that both methods are robust to misspecification, but CDCM may provide more stable directional inference. Full simulation details are provided in Supplementary Material E.

Before assessing consistency between the one-state models presented here and the two-state DCM in Hillebrandt et al,\cite{hillebrandt-etal-2014} we first examine results for goals two and three, as these establish whether averaging effective connectivity across experimental conditions, following their approach, is justified and thus whether such pooled estimates provide a reliable basis for comparison across methods. Given the overall agreement between the two one-state methods, we focus on these goals using CDCM. As noted above, the second objective is to determine which connections are credibly different from zero based on posterior uncertainty. 

Regarding this goal, we find that all parameters in the $\mathbf{A}$ and $\mathbf{C}$ matrices are credibly nonzero, whereas those in $\mathbf{B}$ are less certain. The diagonal elements of $\mathbf{B}$ (first and fourth positions of $\hat{\boldsymbol{\alpha}}_B$; Figure \ref{fig:group-alpha-comparison}) represent modulatory self-loops and are strongly regularized toward zero by the prior (Section \ref{sec:mdl}), resulting in posterior intervals that often include zero and thus inconclusive directionality. The third position of $\hat{\boldsymbol{\alpha}}_B$, corresponding to animate motion modulation of the backward connection from pSTS to V5, shows notable variability across conditions, with some intervals excluding zero and others not. The intervals in Figure \ref{fig:group-alpha-comparison}, however, are not adjusted for multiple testing and therefore likely underestimate posterior uncertainty.

Excluding the self-loops and the pSTS-to-V5 modulation, we find the parameters in the selected model from Hillebrandt et al\cite{hillebrandt-etal-2014} are credibly nonzero, with uncertainty in the self-loop estimates largely driven by prior regularization. For the third goal of examining the replicability of estimated connectivity patterns across subjects and modeling choices, the results in Figure \ref{fig:group-alpha-comparison} broadly support the conclusion of Hillebrandt et al\cite{hillebrandt-etal-2014} that connectivity patterns replicate across sessions and hemispheres. The only potential deviation is the third element of $\hat{\boldsymbol{\alpha}}_B$, which exhibits sign changes across conditions; however, given that uncertainty is likely underestimated and several intervals include or are near zero, there is no strong evidence of meaningful differences across sessions or hemispheres.

\begin{figure}[t!]
    \centering
    \includegraphics[width=\linewidth]{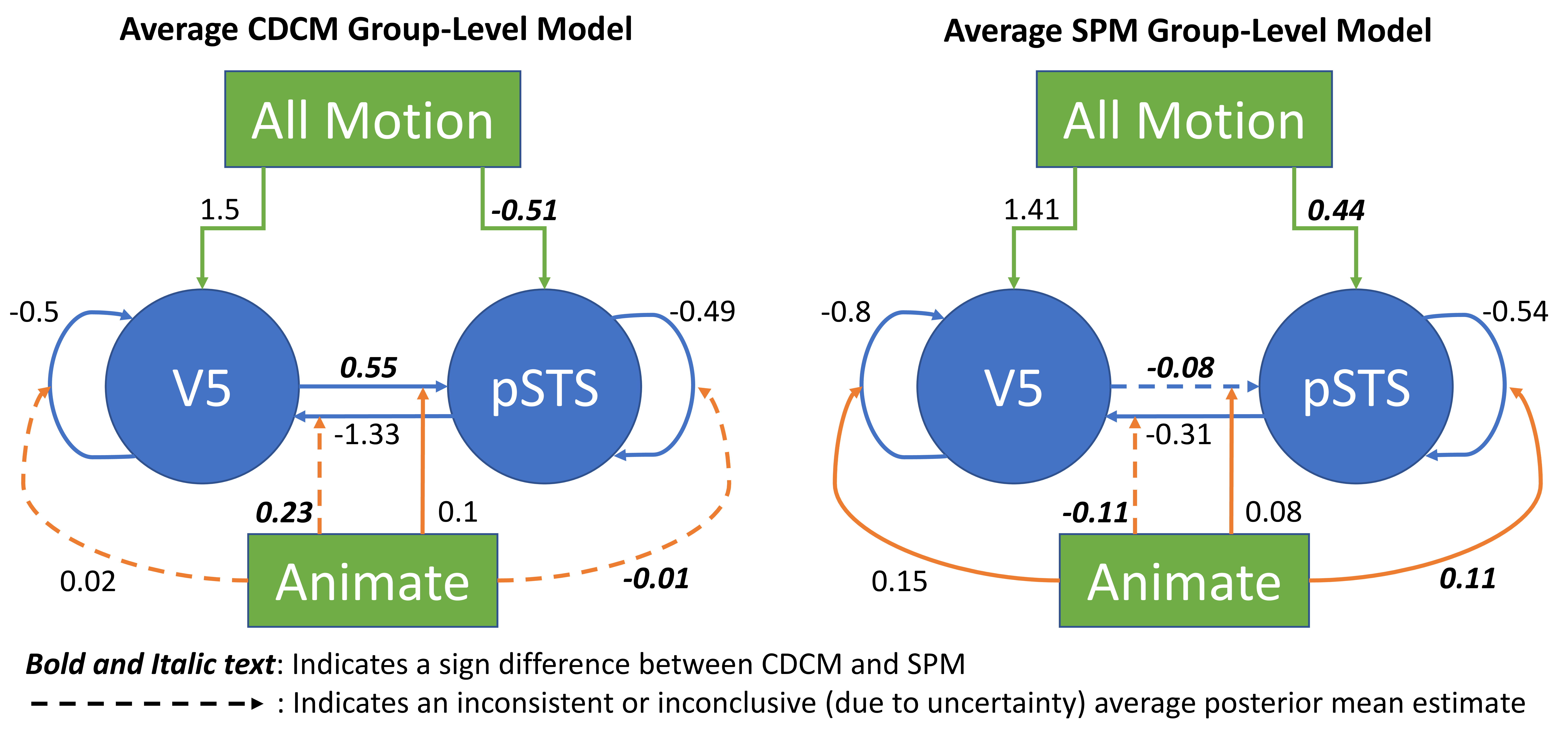}
    \caption{\protect\input{Figure-Captions/avg-group}}
    \label{fig:avg-group}
\end{figure}

Consequently, we return to the first goal and compare the results from Figure \ref{fig:hillebrandt-dcm} in Hillebrandt et al\cite{hillebrandt-etal-2014} and the averaged group-level results in Figure \ref{fig:avg-group}. Of the ten connectivity parameters, four differ in sign between CDCM and SPM, three of which are inconclusive due to uncertainty (dashed lines in Figure \ref{fig:avg-group}). The primary discrepancy between the one-state models concerns the driving input from All Motion to pSTS. While direct comparison with the two-state model is limited by differences in neural dynamics, partial comparisons remain informative. The parameterization of $\mathbf{C}$ is unchanged between one- and two-state models, and differences in $\mathbf{A}$ and $\mathbf{B}$ largely reflect within-region dynamics rather than between-region interactions.\cite{marreiros-etal-2008} Accordingly, we focus on agreement in the signs of $\mathbf{C}$ and the off-diagonal elements of $\mathbf{B}$, as off-diagonal entries of $\mathbf{A}$ are not reported.\cite{hillebrandt-etal-2014}

Comparing Figures \ref{fig:avg-group} and \ref{fig:hillebrandt-dcm}, all three models agree on the sign of the driving input to V5. For the input to pSTS, the estimated effect is negative, consistent with CDCM. Although the estimate in Hillebrandt et al\cite{hillebrandt-etal-2014} is closer to zero, their Bonferroni-corrected interval remains entirely negative $(-0.091, -0.030)$. For the modulatory parameters, both off-diagonal effects in Figure \ref{fig:hillebrandt-dcm} are positive, with a stronger effect on the forward connection from V5 to pSTS (mean $=1.22$) than on the backward connection (mean $=0.16$). The sign of the forward connection is consistent across all three models, whereas the backward connection is less stable, showing the greatest variation within $\hat{\boldsymbol{\alpha}}_B$ for both one-state approaches.

Using the 100 Unrelated Subjects dataset from the HCP-YA 2025 release, we provide a thorough assessment of the replicability of effective connectivity estimates across modeling approaches, subjects, and experimental conditions. Generally, we find consistency across these dimensions and with earlier work, but do so under the CDCM framework, a method built on a foundation of theoretical identifiability. Across small- and large-scale fMRI applications, CDCM achieves consistent and replicable connectivity findings; in controlled simulation settings, CDCM is more robust to misspecification than its complex predecessor.

\section{Discussion}\label{sec:discussion}

We introduced Canonical DCM (CDCM), a simplified framework for modeling effective connectivity in fMRI that enables both theoretical analysis and efficient estimation. While motivated by DCM, the proposed approach reflects a broader strategy for state-space ODE models: simplifying the system structure to obtain a representation that is amenable to theoretical analysis. In contrast to existing DCM formulations, which rely primarily on empirical assessment of parameter recovery,\cite{friston-etal-2003,marreiros-etal-2008,friston-etal-2014} CDCM admits a set of sufficient conditions for parameter identifiability. The use of a piecewise analytic solution to the neural ODE further enables efficient implementation using gradient-based Hamiltonian Monte Carlo and facilitates accurate uncertainty quantification. Ljung\cite{ljung-2010} states that a model should be formed taking two aspects into account: 1) it fits the data well, and 2) is not overly complex. Across simulation studies and real data applications, CDCM produces estimates comparable to those obtained from more complex approaches while offering improved MCMC-based computational efficiency and clearer theoretical guarantees within a simpler framework. 

A central motivation of this work is the role of identifiability in supporting replicability. Because DCM is inherently confirmatory, the reliability of scientific conclusions depends on the ability to consistently recover model parameters across datasets, sessions, and experimental conditions. By establishing explicit conditions under which parameters are uniquely identifiable, CDCM provides a foundation for stable and interpretable inference. This connection is reflected in both the simulation studies and the HCP application, where we observe general agreement with existing methods alongside consistency across repeated measurements, highlighting the importance of identifiability for reproducible connectivity estimates. More generally, the relationship between identifiability and replicability extends to state-space ODE models beyond neuroimaging, where unidentifiable parameters can similarly undermine the stability of inference.

We identify two key directions for future work. First, to meet Assumptions \ref{DCM-A1}--\ref{DCM-A4}, CDCM is currently restricted to fMRI block designs. A block design is a simple way to ensure that there are at least $d+1$ consecutive, equally-spaced observations in each block to identify system parameters from. Depending on the research questions, it may make more sense to have a task fMRI design more similar to that of the event-related design. An area of future work is to investigate the potential relaxation of the Assumptions \ref{DCM-A1}--\ref{DCM-A4} wherever possible. Wang et al\cite{wang-etal-2024} provide some conditions for the identifiability of linear ODE systems from degraded observations, specifically aggregated and time-scaled observations. With event-related or hybrid designs, observations are not necessarily degraded, but they are not consecutive and equally-spaced either. Also, even in a block design, if the model hypothesis includes $d$ ROIs and the experiment only has individual blocks with fewer than $d+1$ observations, can we still provide a set of sufficient conditions for $(\mathbf{s}^*,\theta^*)$-identifiability? Future research examining possible relaxations of Assumptions \ref{DCM-A1}--\ref{DCM-A4} could extend the applicability of CDCM to additional fMRI study designs without being overly restrictive as to how the study and hypotheses should be constructed.

Second, to our knowledge, there are currently no extensions of the DCM framework to incorporate longitudinal hypotheses in a joint estimation procedure. Many existing studies have a longitudinal component (such as pre- and post-treatment/intervention), and currently, the standard procedure is to estimate the DCM separately at each time point and then use a hierarchical parametric empirical Bayes approach to capture longitudinal effects.\cite{almgren-etal-2018,eo-etal-2023,huang-etal-2022,park-etal-2017} We believe there is valuable information worth capturing between the two scans and a joint estimation approach across multiple scanning sessions would provide useful insights. Given the complexity of existing DCM implementations, a longitudinal extension may not be realistic. Therefore, CDCM being simpler provides strong motivation for exploring potential longitudinal extensions to the methodology.

Overall, CDCM provides a conceptually transparent and computationally efficient approach to effective connectivity modeling, together with explicit theoretical guarantees that are often lacking in existing methods. By linking identifiability, estimation, and replicability within a unified framework, this work contributes to a more principled understanding of state-space ODE models and their application to complex biological data.

\subsection*{Data and Code Availability Statement}

The Attention to Visual Motion fMRI dataset is publicly accessible from the Wellcome Centre for Human Neuroimaging here: \url{https://www.fil.ion.ucl.ac.uk/spm/data/attention/}.

Data are provided, in part, by the Human Connectome Project, WU‑Minn Consortium (Principal Investigators: David C. Van Essen and Kamil Ugurbil; 1U54MH091657) funded by the 16 NIH Institutes and Centers that support the NIH Blueprint for Neuroscience Research, and by the McDonnell Center for Systems Neuroscience at Washington University. We acknowledge the contributions of the HCP consortium in making these data publicly available. 

We have developed an \texttt{R} package for simulating from and implementing CDCM available on GitHub: \url{https://github.com/kaitlyn-fales/cdcm}. The \texttt{R} code for reproducing all analyses in the manuscript and supplementary materials is available via GitHub at \url{https://github.com/kaitlyn-fales/Canonical-DCM-Method}.

\begin{singlespace}
\bibliographystyle{vancouver-custom}
\bibliography{references}
\end{singlespace}

\renewcommand{\thesection}{\Alph{section}}
\renewcommand{\thelemma}{S\arabic{lemma}}
\renewcommand{\thecor}{S\arabic{cor}}

\setcounter{section}{0}
\setcounter{table}{0}
\setcounter{figure}{0}
\setcounter{theorem}{0}

\counterwithin{table}{section}
\counterwithin{figure}{section}
\numberwithin{equation}{section}

\clearpage

{\centering \LARGE Supplementary Materials \par}

\section{Canonical HRF as a Special Case of the Balloon Model}\label{appx:special-case}

As we identify in the main text, we can consider the canonical HRF as a type of special case of the Balloon model. We show here why that is the case, by synthesizing the concepts and results of prior work\cite{buxton-etal-1998,glover-1999,friston-etal-2000,friston-2002,friston-etal-2003,zeidman-etal-2019}; no new results are introduced. For simplicity, we present the derivation in the univariate case, where $z(t)$ and the hemodynamic states $s(t), f(t), v(t), q(t)$ are scalar quantities rather than vectors. Following the standard DCM formulation,\cite{friston-etal-2000,friston-etal-2003,buxton-etal-1998} the hemodynamic Balloon model couples neural activity $z(t)$ to a vasoactive signal $s(t)$, blood inflow $f(t)\equiv f_{\mathrm{in}}(t)$, venous volume $v(t)$, and deoxyhemoglobin content $q(t)$:
\begin{align*}
\dot{s}(t) &= z(t) - \kappa s(t) - \gamma\,[f(t)-1], \\
\dot{f}(t) &= s(t), \\
\tau_h \dot{v}(t) &= f(t) - f_{\mathrm{out}}(v,t), \qquad f_{\mathrm{out}}(v,t)=v(t)^{1/\zeta}, \\
\tau_h \dot{q}(t) &= f(t) E(f(t)) - f_{\mathrm{out}}(v,t) \frac{q(t)}{v(t)}.
\end{align*}
Following the implementation used in modern SPM, as described in Zeidman et al,\cite{zeidman-etal-2019} the oxygen extraction function $E(f(t))$ is written in the normalized form,
$$
E(f)=\frac{1-(1-E_0)^{1/f}}{E_0},
$$
where $E_0$ is the resting oxygen extraction fraction. The corresponding fractional change in BOLD signal $y(t) = \Delta S(t)/S_0$ is,
\begin{equation*}
y(t) = \frac{\Delta S(t)}{S_0} \approx V_0 \Big[ k_1 (1-q(t)) + k_2 \Big( 1 - \frac{q(t)}{v(t)} \Big) + k_3 (1-v(t)) \Big],
\end{equation*}
with parameters $\kappa,\gamma,\tau_h,\zeta,E_0,V_0$, and $k_i$ for $i = 1,2,3$ as defined in Zeidman et al,\cite{zeidman-etal-2019} except we use $\zeta$ to denote Grubb's exponent.\cite{grubb-etal-1974}

To obtain a linear time-invariant (LTI) approximation, we linearize the system about its resting point:
$$
f=v=q=1, \qquad s=z=0.
$$
We define small deviations from rest:
$$
x_s=s, \quad x_f=f-1, \quad x_v=v-1, \quad x_q=q-1.
$$
Substituting these into the state equations and keeping only first-order terms from the Taylor expansion gives
\begin{align*}
\dot{x}_s &= z - \kappa x_s - \gamma x_f, \\
\dot{x}_f &= x_s, \\
\tau_h \dot{x}_v &= x_f - \frac{1}{\zeta} x_v, \\
\tau_h \dot{x}_q &= \omega_{E_0} x_f - \left( \frac{1}{\zeta} -1 \right) x_v - x_q,
\end{align*}
where $\omega_{E_0} = 1+ \frac{(1-E_0)\log(1-E_0)}{E_0}$ is a constant, as $E_0$ is a preset value in SPM.\cite{zeidman-etal-2019}

\textit{Note.} In the DCM implementation used by Zeidman et al\cite{zeidman-etal-2019} and in the modern version of SPM, the hemodynamic states $f,v,q$ are expressed in log form ($\log f, \log v, \log q$) for numerical stability and to ensure positivity. Here we use the ratio form (normalized to baseline), which is equivalent to first order since $\log(1+\delta)\approx \delta$ for small deviations. At rest, this corresponds to $f=v=q=1$ in ratio form (or $0$ in log form).

Collecting these equations into vector form with $\mathbf{w}(t) = [x_s, x_f, x_v, x_q]^\top$ gives
$$
\dot{\mathbf{w}}(t) = \mathbf{M}\mathbf{w}(t) + \boldsymbol{a} z(t), \qquad y(t) = \mathbf{d} \mathbf{w}(t),
$$
where
\begin{align*}
    \mathbf{M} &=
\begin{bmatrix}
-\kappa & -\gamma & 0 & 0\\
1 & 0 & 0 & 0\\
0 & 1/\tau_h & -1/(\zeta \tau_h) & 0\\
0 & \omega_{E_0}/\tau_h & -(1/\tau_h)(1/\zeta-1) & -(1/\tau_h)
\end{bmatrix}, \quad
\boldsymbol{a} =
\begin{bmatrix}
1\\0\\0\\0
\end{bmatrix}, \text{ and } \\[10pt]
\mathbf{d} &= V_0
\begin{bmatrix}
0 & 0 & (k_2-k_3) & -(k_1+k_2)
\end{bmatrix}.
\end{align*}

Applying the Laplace transform converts the coupled differential equations into algebraic equations in the Laplace domain, with Laplace-domain variable $\rho$.\cite{oppenheim-etal-1996} Using the property $\mathcal{L}\{\dot{\mathbf{w}}(t)\}=\rho \mathbf{W}(\rho)-\mathbf{w}(0)$, we obtain the algebraic system,
$$
(\rho\mathbf{I}-\mathbf{M})\mathbf{W}(\rho) = \boldsymbol{a} Z(\rho),
$$
where $\mathbf{W}(\rho)$ and $Z(\rho)$ are the Laplace transforms of $\mathbf{w}(t)$ and $z(t)$, and $\mathbf{I}$ denotes the identity matrix. Assuming zero initial conditions ($\mathbf{w}(0)=0$),
$$
\mathbf{W}(\rho) = (\rho\mathbf{I}-\mathbf{M})^{-1} \boldsymbol{a} Z(\rho), \qquad Y(\rho) = \mathbf{d}\mathbf{W}(\rho),
$$
so that the transfer function of the system is:
$$
H(\rho) \equiv \frac{Y(\rho)}{Z(\rho)} = \mathbf{d}(\rho\mathbf{I}-\mathbf{M})^{-1} \boldsymbol{a}.
$$

Intuitively, the Laplace transform replaces differentiation with multiplication by $\rho$, reducing the coupled differential equations to algebraic form. The matrix inverse, $(\rho\mathbf{I}-\mathbf{M})^{-1}$, encodes the system’s internal dynamics. Its poles (the roots of $\det(\rho\mathbf{I}-\mathbf{M})=0$) determine the system's characteristic time constants, controlling the rise, decay, and undershoot of the impulse response. 

The impulse response of the linearized hemodynamic system is defined as $h(t)=\mathcal{L}^{-1}\{H(\rho)\}$. For an arbitrary input $z(t)$, with $Z(\rho)=\mathcal{L}\{z(t)\}$ and $Y(\rho)=\mathcal{L}\{y(t)\}$,
$$
Y(\rho)=H(\rho)Z(\rho)
\quad \Longleftrightarrow \quad
y(t) = (h * z)(t),
$$
so the BOLD response is given by the convolution of the impulse response with the neural input. In the Volterra framework, this convolution corresponds to the first-order term in the Volterra expansion, and thus $h(t)$ is identified with the first-order Volterra kernel of the nonlinear Balloon model.\cite{friston-etal-2000}

For physiologically plausible parameters, $h(t)$ exhibits a biphasic form, with a positive peak followed by a smaller undershoot, corresponding to the dominant poles of the linearized system.\cite{buxton-etal-1998} The canonical HRF \eqref{eq:HRF} used in SPM is a convenient parametric approximation to this impulse response, given by:
\begin{equation*}
        h(t) = \frac{\beta_1^{\alpha_1}}{\Gamma(\alpha_1)}t^{\alpha_1 -1}e^{-\beta_1t} - c\frac{\beta_2^{\alpha_2}}{\Gamma(\alpha_2)}t^{\alpha_2 -1}e^{-\beta_2t},
\end{equation*}
with parameters $\alpha_1=6$, $\alpha_2=16$, $\beta_1=\beta_2 = 1$, and $c = 1/6$. Thus, the canonical HRF can be viewed as a parametric approximation to the impulse response of the linearized Balloon model, or equivalently, to the first-order Volterra kernel of the nonlinear hemodynamic system.\cite{glover-1999,friston-2002} For a gentle introduction to the Laplace transform and transfer-function representation, see Oppenheim et al.\cite{oppenheim-etal-1996}

\section{Diagonal Parameterization}\label{appx:diagonal-par}

Here, we provide further discussion of the diagonal parameterization used in SPM's DCM. We show how the standard SPM exponential mapping for diagonal elements affects the relationship between the original neural equation and the reparameterized form. Additionally, we discuss our alternative reparameterization strategy used in CDCM that preserves the additive structure of the neural model while allowing more flexible prior specification for diagonal and off-diagonal elements.

Consider the latent neural state $\mathbf{z}(t) \in \mathbb{R}^d$, with $t \in (0, \infty)$. Recall, the neural signal model in the bilinear, one-state DCM is,
\begin{equation}\label{eqn:neural-dcm}
\dot{\mathbf{z}}(t) = (\mathbf{A}  + \sum_{i=1}^{m} u_i(t) \mathbf{B}_i )\mathbf{z}(t) + \mathbf{C} \mathbf{u}(t).
\end{equation}
We can split the $\mathbf{A}$ and $\mathbf{B} = [\mathbf{B}_1,\dots,\mathbf{B}_m]$ connectivity matrices into diagonal and off-diagonal parts:
\begin{equation*}
\dot{\mathbf{z}}(t) = (\mathbf{A}^D + \sum_{i=1}^{m} u_i(t) \mathbf{B}_i^D) \mathbf{z}(t) + (\mathbf{A}^O + \sum_{i=1}^{m} u_i(t) \mathbf{B}_i^O) \mathbf{z}(t) + \mathbf{C} \mathbf{u}(t),
\end{equation*}
where $\mathbf{A}^D = \mathrm{diag}(a_1, \ldots, a_d)$, and
$\mathbf{B}_i^D = \mathrm{diag}(b_1^{(i)}, \ldots, b_d^{(i)})$, for $i=1,\dots,m$. Thus, the diagonal contribution for element $\ell$ is,
\begin{equation*}
a_\ell + \sum_{i=1}^{m} u_i(t) b_\ell^{(i)}, \quad \ell=1\dots,d.
\end{equation*}
A sufficient condition for numerical stability of the dynamic system is the baseline $a_\ell \le 0$ and  modulation $b_\ell^{(i)} \le -a_\ell$ for all $\ell$ and $i$. Essentially, each region should experience self-inhibition $(a_\ell \le 0)$ at baseline to constrain system dynamics, and any modulation to that baseline should not be too strong $(b_\ell^{(i)} \le -a_\ell)$ such that each region is always experiencing some level of self-inhibition. 

To achieve this, SPM uses a default intrinsic (self-inhibition) strength of $-0.5$ Hz, and applies an exponential mapping to the diagonal \cite{zeidman-etal-2019}:
\begin{equation*}
\mathbf{A}^D + \sum_{i=1}^{m} u_i(t) \mathbf{B}_i^D \quad \mapsto \quad 
-0.5 \, \exp\Big(\bar{\mathbf{A}}^{D} + \sum_{i=1}^{m} u_i(t) \bar{\mathbf{B}}_i^{D}\Big),
\end{equation*}
in which each diagonal element then becomes,
\begin{equation*}
-0.5\exp\Big(\bar{a}_\ell + \sum_{i=1}^{m} u_i(t) \bar{b}_\ell^{(i)}\Big).
\end{equation*}
Recovery of the original neural model \eqref{eqn:neural-dcm}, aside from the multiplicative constant of $-0.5$, requires,
\begin{equation*}
a_\ell + \sum_{i=1}^{m} u_i(t) b_\ell^{(i)} = \exp\Big(\bar{a}_\ell + \sum_{i=1}^{m} u_i(t) \bar{b}_\ell^{(i)}\Big), \quad \forall \ell \in \{1,\dots,d\}.
\end{equation*}
If all $u_i(t) = 0$, for $i = 1,\dots,m$, then,
\begin{equation*}
\bar{a}_\ell = \log(a_\ell).
\end{equation*}
If a single input $u_i(t) = 1$, then,
\begin{equation*}
\bar{b}_\ell^{(i)} = \log\Big(1 + \frac{b_\ell^{(i)}}{a_\ell}\Big).
\end{equation*}
When multiple inputs are active, cross-terms appear. For $u_i(t) = u_{i^\prime}(t) = 1$, $i \neq i^\prime$:
\begin{align*}
a_\ell + b_\ell^{(i)} + b_\ell^{(i^\prime)} &= a_\ell \exp\big(\bar{b}_\ell^{(i)} + \bar{b}_\ell^{(i^\prime)}\big) \\
&= a_\ell \Big(1 + \frac{b_\ell^{(i)}}{a_\ell}\Big) \Big(1 + \frac{b_\ell^{(i^\prime)}}{a_\ell}\Big) \\
&= a_\ell + b_\ell^{(i)} + b_\ell^{(i^\prime)} + \frac{b_\ell^{(i)} b_\ell^{(i^\prime)}}{a_\ell}.
\end{align*}
Thus, exact equivalence requires
\begin{equation*}
b_\ell^{(i)} b_\ell^{(i^\prime)} = 0 \quad \forall i \ne i^\prime,
\end{equation*}
which implies that multiple stimuli cannot simultaneously modulate an intrinsic connection. In practice, this is not an entirely unreasonable assumption to make, as it rarely would be of interest to researchers. However, it is important to note because this diagonal parameterization does not automatically guarantee equivalence of the neural model specification given in \eqref{eqn:neural-dcm}.

To view the reparameterization differently, consider small values of 
$\bar{a}_\ell$ and $\bar{b}_\ell^{(i)}$ (e.g., 
$|\bar{a}_\ell|, |\bar{b}_\ell^{(i)}| \le 0.2$). Using a second-order 
Taylor expansion,
\begin{equation*}
e^x \approx 1 + x + \frac{1}{2}x^2,
\end{equation*}
we obtain, when a single input $u_i(t) = 1$,
\begin{align*}
-0.5 \, e^{\bar{a}_\ell + \bar{b}_\ell^{(i)}} 
&\approx -0.5 \Bigl[
1 + \bar{a}_\ell + \bar{b}_\ell^{(i)}
+ \tfrac{1}{2}\bar{a}_\ell^2
+ \bar{a}_\ell \bar{b}_\ell^{(i)}
+ \tfrac{1}{2}\{\bar{b}_\ell^{(i)}\}^2
\Bigr].
\end{align*}
Retaining only first-order terms yields the simpler approximation
\begin{equation*}
-0.5 \, e^{\bar{a}_\ell + \bar{b}_\ell^{(i)}}
\approx -0.5\bigl(1 + \bar{a}_\ell + \bar{b}_\ell^{(i)}\bigr).
\end{equation*}
Thus, although the parametrization is multiplicative on the constrained 
scale, it behaves approximately additively in $\bar{a}_\ell$ and 
$\bar{b}_\ell^{(i)}$ for small values, inducing an approximate effect 
$a_\ell + b_\ell^{(i)}$ on the original scale.

In this view, the reparameterization only introduces a small error (approximately 2\%) and a constant offset term. Additionally, in the case where $u_i(t) = u_{i^\prime}(t) = 1$, with small $a_\ell, b_\ell^{(i)}$ and $b_\ell^{(j)}$ the cross-terms are near zero. However, SPM's diagonal priors do not necessarily guarantee small enough modulation elements $\bar{b}_\ell^{(i)}$ by default, making the Taylor approximation approach impractical. The standard SPM approach applies a restrictive prior (variance $1/64$) to all elements of $\mathbf{A}$, including off-diagonal, and uses a wider prior (variance $1$) for all elements of $\mathbf{B}_i$.\cite{zeidman-etal-2019} The rationale for this comes from the fact that the modulation elements are effectively scaled by the baseline intrinsic connections, as seen in the earlier argument above. Therefore, SPM's approach implies either the imposition of an additional assumption on intrinsic modulatory parameters, or a more restrictive prior on them in order to recover \eqref{eqn:neural-dcm}.

Alternatively, we can reparameterize only the diagonal of $\mathbf{A}$ while keeping the $\mathbf{B}_i$ matrices unchanged, outside of a more restrictive prior, which is the approach we take for CDCM (discussed in Section \ref{sec:mdl}). In this approach, the diagonal elements of both $\mathbf{A}$ and $\mathbf{B}_i$ are assigned restrictive priors with variance $1/64$, while the off-diagonal elements of both matrices receive wider priors with variance $1$. Reparameterizing only the diagonal of $\mathbf{A}$ preserves the additive structure of \eqref{eqn:neural-dcm}, controls dynamic system stability via the diagonal priors, and avoids the restrictive assumption that multiple diagonal modulatory terms in $\mathbf{B}_i$ must be mutually exclusive.

\section{Detailed Proofs}\label{appx:proofs}

\setcounter{theorem}{0}
\setcounter{prop}{0}

In this section, we present the detailed proofs of our proposition and theorems. 

\subsection{Proof of Proposition \ref{prop:dcm_sol_trajectory}}
\begin{prop}

Let $\mathbf{v}(t;\, \mathbf{v}(0), (\mathbf{A},\mathbf{c}))$ denote the solution at time $t\in \mathbb{R}_+$ of the affine system \eqref{eq: affine_ODE}.
Also let $\mathbf{z}(t; \mathbf{s}^*,\theta^*) $ denote the solution of the DCM dynamics in \eqref{eqn:canonical-dcm}, and $\theta^* = \{\mathbf{A},\mathbf{B}_1,\hdots,\mathbf{B}_m,\mathbf{C}\}$. For $b \in \{0,\dots,B\}$, with $\Tilde{\mathbf{A}}^{(b)}$ and $\Tilde{\mathbf{c}}^{(b)}$ defined in \eqref{eq:z_dynamics_within_b},
\begin{equation*}
    \mathbf{z}(t; \,\mathbf{s}^*,\theta^*) = \mathbf{v}(t-t^{(b)};\,\mathbf{z}(t^{(b)}),(\Tilde{\mathbf{A}}^{(b)},\Tilde{\mathbf{c}}^{(b)})),\,\, t \in [t^{(b)}, t^{(b+1)}).
\end{equation*}
In particular, the closed-form solution of $\mathbf{z}(t; \mathbf{s}^*,\theta^*) $ is given by
\begin{equation*}
    \mathbf{z}(t; \mathbf{s}^*,\theta^*) = e^{\Tilde{\mathbf{A}}^{(b)}(t-t^{(b)})}\{\mathbf{z}(t^{(b)})+(\Tilde{\mathbf{A}}^{(b)})^{-1}\Tilde{\mathbf{c}}^{(b)}\} - (\Tilde{\mathbf{A}}^{(b)})^{-1}\Tilde{\mathbf{c}}^{(b)},\,\,t \in [t^{(b)}, t^{(b+1)}).
\end{equation*}
\end{prop}
\begin{proof}
Let $b=0,\dots,B$ be given. For $t \in [t^{(b)}, t^{(b+1)})$, 
\begin{align*}
\dot{\mathbf{z}}(t; \mathbf{s}^*,\theta^*) 
&= \left(\mathbf{A} + \sum_{i=1}^m \Tilde{u}_i^{(b)}\mathbf{B}_i\right) \mathbf{z}(t; \mathbf{s}^*,\theta^*) + \mathbf{C}\Tilde{\mathbf{u}}^{(b)} \\
&= \Tilde{\mathbf{A}}^{(b)} \mathbf{z}(t; \mathbf{s}^*,\theta^*) + \Tilde{\mathbf{c}}^{(b)}.
\end{align*}
Therefore, on $[t^{(b)}, t^{(b+1)})$, the trajectory $\mathbf{z}(t; \mathbf{s}^*,\theta^*)$ is identical to
$$
\mathbf{v}(t-t^{(b)};\mathbf{z}(t^{(b)};\mathbf{s}^*,\theta^*),(\Tilde{\mathbf{A}}^{(b)},\Tilde{\mathbf{c}}^{(b)})),
$$
establishing the first claim.

As the closed-form solution of \eqref{eq: affine_ODE} is given by \eqref{eq: affine_ODE_sol}, substituting $t-t^{(b)}$ for $t$, $\mathbf{z}(t^{(b)})$ for $\mathbf{v}(0)$, and $(\Tilde{\mathbf{A}}^{(b)},\Tilde{\mathbf{c}}^{(b)})$ for $(\mathbf{A},\mathbf{c})$ in \eqref{eq: affine_ODE_sol} gives the stated closed-form expression for $\mathbf{z}(t; \mathbf{s}^*,\theta^*)$.
\end{proof}

\subsection{Proof of Theorem \ref{thm:canonical-dcm}}

We first introduce some notations. For $b=0,1,\dots,B$, we define 
\begin{align*}
    \mathbf{v}^{(b)}(t) = \mathbf{v}(t;\,\mathbf{z}(t^{(b)}), (\Tilde{\mathbf{A}}^{(b)}, \Tilde{\mathbf{c}}^{(b)}))
\end{align*}
as the solution at time $t\in\mathbb{R}_{+}$ of the affine
system 
\begin{equation}
    \dot{\mathbf{v}}(t) = \Tilde{\mathbf{A}}^{(b)}\mathbf{v}(t) + \Tilde{\mathbf{c}}^{(b)},
\end{equation}
with the initial condition $\mathbf{z}(t^{(b)})$. By Proposition~\ref{prop:dcm_sol_trajectory}, this solution coincides with the corresponding segment of the original trajectory: \[\mathbf{v}^{(b)}(t) = \mathbf{z}(t^{(b)}+t),\quad t\in[t^{(b)}, t^{(b+1)}). \] We then write $\mathbf{v}_j^{(b)}$ for observations taken at equally spaced intervals of length $r$:
\begin{align}\label{def: v_j^b}
    \mathbf{v}_j^{(b)} = \mathbf{v}^{(b)}(jr) = \mathbf{z}(t^{(b)}+jr),\,\,j=0,1,2,\dots
\end{align}

We present results from Duan et al\cite{duan-etal-2020} that are used to establish identifiability of the block-specific parameters $(\Tilde{\mathbf A}^{(b)}, \Tilde{\mathbf c}^{(b)})$ in Lemma \ref{thm:block-dynamics}. We then use Lemma \ref{thm:block-dynamics} to prove Theorem \ref{thm:canonical-dcm}.


 \begin{lemma}[Duan et al,\cite{duan-etal-2020} Lemma 4.2]\label{thm:duan-lemma}
     Let $\mathbf{v}_0,\dots,\mathbf{v}_{d+1}$ be any $d+2$ vectors in $\mathbb{R}^d$. 
     Define the augmented vectors $\mathbf{x}_j = [\mathbf{v}_j,1]^\top$ for $j=0,\dots,d+1$. Define the following matrices:
     \begin{align*}
         \mathbf{Z}_i &:= [\mathbf{v}_i,\dots,\mathbf{v}_{i+d-1}] \in \mathbb{R}^{d\times d} \mbox{ for } i \in \{0,1,2\}\\
         \mathbf{X}_i &:=[\mathbf{x}_i,\dots,\mathbf{x}_{i+d}] \in \mathbb{R}^{(d+1)\times(d+1)} \mbox{ for } i \in \{0,1\}.
     \end{align*}
     Suppose there exists a matrix $\Psi\in\mathbb{R}^{(d+1)\times (d+1)}$ that maps the first block of augmented vectors to the second, such that:
     $$\Psi \mathbf{X}_0 = \mathbf{X}_1.$$
     If the matrix $\mathbf{X}_0$ is invertible, then $\Psi$ must take the block form:
     $$\Psi = \begin{bmatrix} \Phi & \mathbf{p} \\ 0 & 1 \end{bmatrix}$$
     where the top-left block $\Phi \in \mathbb{R}^{d \times d}$ and the top-right vector $\mathbf{p} \in \mathbb{R}^d$ are given by:
     \begin{align*}
         \Phi &= (\mathbf{Z}_2 - \mathbf{Z}_1)(\mathbf{Z}_1 - \mathbf{Z}_0)^{-1}\\
         \mathbf{p} &= \mathbf{v}_{i+1} - \Phi \mathbf{v}_i\quad \text{for any } i \in \{0, \dots, d\}.
     \end{align*}
\end{lemma}

\begin{lemma}[Duan et al,\cite{duan-etal-2020} Theorem 4.5]\label{thm:duan-id}
    Let $\Phi \in \mathbb{R}^{d \times d}$. The following are equivalent:
    \begin{description}\setlength{\itemsep}{-0.5em}
        \item[(a)] For any $\mathbf{p} \in \mathbb{R}^d$, there exists a real, unique solution pair $(\mathbf{A},\mathbf{c}) \in \mathbb{R}^{d \times d} \times \mathbb{R}^d$ to the system,
        \begin{align*}
            e^{\mathbf{A}} &= \Phi, \\
            w(\mathbf{A})\mathbf{c} &= \mathbf{p},
        \end{align*}
        where $w(\mathbf{A}) := \sum_{k=0}^\infty \frac{1}{(k+1)!}\mathbf{A}^k$.
        \item[(b)] There exists a unique real solution $\mathbf{A} \in \mathbb{R}^{d \times d}$ to $\Phi = e^{\mathbf{A}}$.
        \item[(c)] The eigenvalues of $\Phi$ are all positive real and each Jordan block of $\Phi$ only occurs once. 
    \end{description}
\end{lemma}

\begin{cor}[Duan et al,\cite{duan-etal-2020} Corollary 4.7]\label{thm:duan-cor}
    Let $\mathbf{A} \in \mathbb{R}^{d \times d}$, and $t\in \mathbb{R}_+$. If $t\mathbf{A}$ is the unique real logarithm of $e^{t\mathbf{A}}$, then $w(t;\mathbf{A})$ is invertible.
\end{cor}

We now present Lemma \ref{thm:block-dynamics}, which establishes identifiability of the block-specific parameters $(\Tilde{\mathbf A}^{(b)}, \Tilde{\mathbf c}^{(b)})$ as functions of the latent discrete observations $\{\mathbf v_0^{(b)},\mathbf v_1^{(b)},\dots,\mathbf v_{d+1}^{(b)}\}$ within block $b$.

\begin{lemma}\label{thm:block-dynamics}
Let a block $b\in\mathcal B^*$ be fixed, over which the stimulus vector is constant and equal to
$\Tilde{\mathbf u}^{(b)} \in \mathbb{R}^m$.
Consider the DCM latent trajectory $\mathbf z(t;\mathbf s^*,\theta^*)$ defined by \eqref{eqn:canonical-dcm}. Recall the definition of the block-specific matrix and affine term for block $b$:
\begin{equation*}
        \Tilde{\mathbf A}^{(b)}
    =
    \mathbf A + \sum_{i=1}^m \Tilde u_i^{(b)}\mathbf B_i \in \mathbb{R}^{d \times d},
    \qquad
    \Tilde{\mathbf c}^{(b)} 
    =
    \mathbf C\Tilde{\mathbf u}^{(b)} \in \mathbb{R}^{d}.
\end{equation*}
Under Assumptions \ref{DCM-A3} and \ref{DCM-A4}, $\Tilde{\mathbf A}^{(b)}$ and $\Tilde{\mathbf c}^{(b)}$ are uniquely determined as functions of $\{\mathbf v_0^{(b)},\mathbf v_1^{(b)},\dots,\mathbf v_{d+1}^{(b)}\}$.
\end{lemma}
\begin{proof}
As the stimulus vector is constant within block $b$, for $t \in [t^{(b)}, t^{(b+1)})$, the trajectory follows the affine dynamics governed by $(\Tilde{\mathbf{A}}^{(b)},\Tilde{\mathbf{c}}^{(b)})$. Leveraging results from Duan et al\cite{duan-etal-2020}, we show that  $\Tilde{\mathbf{A}}^{(b)}$ and $\Tilde{\mathbf{c}}^{(b)}$ can be expressed as functions of $\{\mathbf v_0^{(b)},\mathbf v_1^{(b)},\dots,\mathbf v_{d+1}^{(b)}\}$.

    \textbf{Step (i) --} Using Lemma \ref{thm:duan-lemma}, we establish that $\Phi^{(b)}$ and $\mathbf{p}^{(b)}$, defined respectively as the upper-left and upper-right blocks of
    \begin{equation*}
        \text{exp}
        \left\{ r
            \begin{bmatrix}
                \Tilde{\mathbf{A}}^{(b)} & \Tilde{\mathbf{c}}^{(b)} \\
                0 & 0
            \end{bmatrix}
        \right\},
    \end{equation*} are uniquely determined by the first $d+2$ discrete observations of block $b$, $\mathbf{v}^{(b)}_0,\dots,\mathbf{v}^{(b)}_{d+1}$.

Let $\mathbf{x}^{(b)}(t) := [\mathbf{v}^{(b)}(t),1]^\top \in \mathbb{R}^{d+1}$, which gives the linear system,
    \begin{align}\label{eqn:affine-system}
        \dot{\mathbf{x}}^{(b)}(t) &= 
        \begin{bmatrix}
            \Tilde{\mathbf{A}}^{(b)} & \Tilde{\mathbf{c}}^{(b)} \\
            0 & 0
        \end{bmatrix} \mathbf{x}^{(b)}(t) = 
        \begin{bmatrix}
            \Tilde{\mathbf{A}}^{(b)} & \Tilde{\mathbf{c}}^{(b)} \\
            0 & 0
        \end{bmatrix}
        \begin{bmatrix}
            \mathbf{v}^{(b)}(t) \\
            1
        \end{bmatrix} = 
        \begin{bmatrix}
            \Tilde{\mathbf{A}}^{(b)}\mathbf{v}^{(b)}(t) + \Tilde{\mathbf{c}}^{(b)} \\
            0
        \end{bmatrix}.
    \end{align}
    In particular, the solution to system (\ref{eqn:affine-system}) is 
    \begin{equation}\label{eq: lems3-eq2}
     \mathbf{x}^{(b)}(t) =(\Psi^{(b)})^r \mathbf{x}^{(b)}(t-r) =  \hdots= (\Psi^{(b)})^t \mathbf{x}^{(b)}(0)   
    \end{equation}
    where we define
    \begin{equation*}
        \Psi^{(b)} := \exp
        \left\{
            \begin{bmatrix}
                \Tilde{\mathbf{A}}^{(b)} & \Tilde{\mathbf{c}}^{(b)} \\
                0 & 0
            \end{bmatrix}
        \right\}.
    \end{equation*}
    
    Let $\mathbf{x}_j^{(b)} := \mathbf{x}^{(b)}(t_j)$, $j \in \{0,\dots,d+1\}$. 
    Define $\mathbf{X}_{0}^{(b)}$, $\mathbf{X}_{1}^{(b)} \in \mathbb{R}^{(d+1)\times(d+1)}$ as $\mathbf{X}_{0}^{(b)} = [\mathbf{x}_{0}^{(b)},\dots,\mathbf{x}_{d}^{(b)}]$, and $\mathbf{X}_{1}^{(b)} = [\mathbf{x}_1^{(b)},\dots,\mathbf{x}_{d+1}^{(b)}]$. By \eqref{eq: lems3-eq2}, there is a recursive relation between $\mathbf{x}_{j+1}^{(b)}$ and $\mathbf{x}_{j}^{(b)}$, where $\mathbf{x}_{j+1}^{(b)} = (\Psi^{(b)})^r\mathbf{x}_{j}^{(b)}$ for each $j \in \{0,\dots,d\}$ since $t_{j+1}-t_j = r$. Then, \[\mathbf{X}_{1}^{(b)} = (\Psi^{(b)})^r\mathbf{X}_{0}^{(b)}.\] Since $(\mathbf{X}_{0}^{(b)})^{-1}$ exists by Assumption \ref{DCM-A4}, $ (\Psi^{(b)})^r = \mathbf{X}_{1}^{(b)}(\mathbf{X}_{0}^{(b)})^{-1}$ and by Lemma \ref{thm:duan-lemma}, 
    \begin{align}\label{eq: Phib^r_form}
        (\Psi^{(b)})^r= \begin{bmatrix}
            \Phi^{(b)} &\mathbf{p}^{(b)}\\
            0 &1
        \end{bmatrix}
    \end{align}
    where
    $\Phi^{(b)}$ and $ \mathbf{p}^{(b)}$ are given by 
    \begin{equation}
    \begin{aligned}
        \Phi^{(b)}&= (\mathbf{Z}_2^{(b)} - \mathbf{Z}_1^{(b)})(\mathbf{Z}_1^{(b)} - \mathbf{Z}_0^{(b)})^{-1} \\
       \mathbf{p}^{(b)} &= \mathbf{v}^{(b)}_{j+1} - \Phi^{(b)}\mathbf{v}^{(b)}_j,\,\,\,j \in \{0,\dots,d\}
    \end{aligned}
    \end{equation}
   where $\mathbf{Z}^{(b)}_i := [\mathbf{v}^{(b)}_i,\dots,\mathbf{v}^{(b)}_{i+d-1}]$ for $i\in \{0,1,2\}$. This provides a mapping from $\{\mathbf{v}^{(b)}_0,\dots,\mathbf{v}^{(b)}_{d+1}\}$ to $\Phi^{(b)}$ and $\mathbf{p}^{(b)}$.
   
    \textbf{Step (ii) --} Show that $\Tilde{\mathbf{A}}^{(b)}$ and $\Tilde{\mathbf{c}}^{(b)}$ are uniquely identified from $\Phi^{(b)}$ and $\mathbf{p}^{(b)}$, which implies identification of $\Tilde{\mathbf{A}}^{(b)}$ and $\Tilde{\mathbf{c}}^{(b)}$ from $\{\mathbf{v}^{(b)}_0,\dots,\mathbf{v}^{(b)}_{d+1}\}$.

    Following the argument of Duan et al,\cite{duan-etal-2020} 
    \begin{equation}\label{eq: trans_mat_equiv}
        \text{exp}
        \left\{ r
            \begin{bmatrix}
                \Tilde{\mathbf{A}}^{(b)} & \Tilde{\mathbf{c}}^{(b)} \\
                0 & 0
            \end{bmatrix}
        \right\} = \begin{bmatrix}
            e^{r\Tilde{\mathbf{A}}^{(b)}} & w(r;\Tilde{\mathbf{A}}^{(b)})\Tilde{\mathbf{c}}^{(b)} \\
            0 & 1
        \end{bmatrix}
    \end{equation}
    where $w(t;\mathbf{A})$ denotes the antiderivative of $e^{t\mathbf{A}}$ given by
\begin{equation}\label{def: anti_deriv_eAt}
    w(t;\mathbf{A}) = \sum_{k=0}^\infty \frac{t^{k+1}}{(k+1)!}\mathbf{A}^k, 
    \end{equation}
    so that $\frac{d}{dt}w(t;\mathbf{A}) = e^{t\mathbf{A}}$.
    
Combining \eqref{eq: Phib^r_form} and \eqref{eq: trans_mat_equiv},
    \begin{align*}
         e^{r\Tilde{\mathbf{A}}^{(b)}} &= \Phi^{(b)} = (\mathbf{Z}_2^{(b)} - \mathbf{Z}_1^{(b)})(\mathbf{Z}_1^{(b)} - \mathbf{Z}_0^{(b)})^{-1}, \nonumber\\
         w(r\Tilde{\mathbf{A}}^{(b)})(r\Tilde{\mathbf{c}}^{(b)}) &= \mathbf{p}^{(b)} = \mathbf{v}^{(b)}_{j+1} - \Phi^{(b)}\mathbf{v}^{(b)}_j, \,\, \forall j \in \{0,\dots,d\}.
    \end{align*}
    where we also use
    \begin{equation*}
        w(r; \Tilde{\mathbf{A}}^{(b)}) = \sum_{k=0}^\infty \frac{r^{k+1}}{(k+1)!}  (\Tilde{\mathbf{A}}^{(b)})^k = r \sum_{k=0}^\infty \frac{1}{(k+1)!}  (r\Tilde{\mathbf{A}}^{(b)})^k = r w(1; r\Tilde{\mathbf{A}}^{(b)}).
    \end{equation*}

    By Assumption \ref{DCM-A3}, $\Tilde{\mathbf{A}}^{(b)}$ is diagonalizable, and therefore $\Tilde{\mathbf{A}}^{(b)} = \mathbf{P} \mathbf{D}\mathbf{P}^{-1}$ for an invertible $\mathbf{P}\in \mathbb{R}^{d\times d}$ and a diagonal matrix $\mathbf{D} = \text{diag}(\lambda_1,\dots,\lambda_d)$ such that $\lambda_i \ne \lambda_j$ for $i\ne j$. Then 
    \begin{equation*}
     e^{r\Tilde{\mathbf{A}}^{(b)}} = \sum_{k=0}^\infty \frac{(r\Tilde{\mathbf{A}}^{(b)})^k}{k!} =    \mathbf{P}\sum_{k=0}^\infty \frac{(r\mathbf{D})^k}{k!} \mathbf{P}^{-1} = \mathbf{P}e^{r\mathbf{D}} \mathbf{P}^{-1}
    \end{equation*}
    where $e^{r\mathbf{D}} = \text{diag}(e^{r\lambda_1},\dots,e^{r\lambda_d})$. In particular, the eigenvalues of $e^{r\Tilde{\mathbf{A}}^{(b)}}$ are all positive real and distinct, therefore condition \textbf{(c)} in Lemma~\ref{thm:duan-id} is satisfied for $\Phi^{(b)} = e^{r\Tilde{\mathbf{A}}^{(b)}}$. Then there exists a unique real logarithm of $r\Tilde{\mathbf{A}}^{(b)} = \log(e^{r\Tilde{\mathbf{A}}^{(b)}}) = \log((\mathbf{Z}_2^{(b)} - \mathbf{Z}_1^{(b)})(\mathbf{Z}_1^{(b)} - \mathbf{Z}_0^{(b)})^{-1})$. Moreover, by Corollary \ref{thm:duan-cor}, if $r\Tilde{\mathbf{A}}^{(b)}$ is the unique logarithm of $e^{r\Tilde{\mathbf{A}}^{(b)}}$, then $w(r\Tilde{\mathbf{A}}^{(b)})$ is invertible. Therefore, 
    \begin{align*}
        \Tilde{\mathbf{A}}^{(b)} &= r^{-1}\log((\mathbf{Z}_2^{(b)} - \mathbf{Z}_1^{(b)})(\mathbf{Z}_1^{(b)} - \mathbf{Z}_0^{(b)})^{-1})\\
        \Tilde{\mathbf{c}}^{(b)} &= r^{-1} w(r\Tilde{\mathbf{A}}^{(b)})^{-1}(\mathbf{v}^{(b)}_{1} - \Phi^{(b)}\mathbf{v}^{(b)}_0).
    \end{align*}
    where $\log(\mathbf{A})$ denotes a matrix logarithm of the matrix $\mathbf{A}$. Thus, $\Tilde{\mathbf A}^{(b)}$ and $\Tilde{\mathbf c}^{(b)}$ are functions of $\{\mathbf v_0^{(b)},\dots,\mathbf v_{d+1}^{(b)}\}$. 

\end{proof}

Now, we use Lemma \ref{thm:block-dynamics} to prove Theorem \ref{thm:canonical-dcm}.

\begin{theorem}
Assume \ref{DCM-A1}--\ref{DCM-A4}. Then the ODE system in (\ref{eqn:canonical-dcm}) is $(\mathbf{s}^*,\theta^*)$-identifiable.
\end{theorem}
\begin{proof}
Let $\mathbf z_j := \mathbf{z}(t_j;\mathbf{s}^*,\theta^*)$ for $j=1,\dots,n$.
We prove the contrapositive: if there exist $\mathbf{s}'$ and $\theta'$ such that
$\mathbf{z}(t_j;\mathbf{s}',\theta') = \mathbf{z}(t_j;\mathbf{s}^*,\theta^*)$ for all $j=1,\dots,n$, then $\mathbf{s}'=\mathbf{s}^*$ and $\theta'=\theta^*$. Thus, it suffices to show that $(\mathbf{s}^*,\theta^*)$ are uniquely determined by the discrete trajectory $\{\mathbf z_j\}_{j=1}^n$. We show that $\mathbf{s}^*$ and $\theta^* = \{\mathbf{A},\mathbf{B}_1,\dots,\mathbf{B}_m,\mathbf{C}\}$ can be expressed as functions of $\{\mathbf z_j\}_{j=1}^n$.

By Assumption \ref{DCM-A2}, there exist $m+1$ blocks
$\mathcal B^*=\{b_1,\dots,b_{m+1}\}$ with each block $b_k$ satisfying Assumption \ref{DCM-A1}, such that
$$
\Tilde{\mathbf U}^* =
\begin{bmatrix}
1 & \Tilde{\mathbf u}^{(b_1)\top}\\
\vdots & \vdots\\
1 & \Tilde{\mathbf u}^{(b_{m+1})\top}
\end{bmatrix}
$$
is invertible.

We let $\mathbf v_j^{(b)} := \mathbf z(t^{(b)} + jr)$ for $j=0,\dots,d+1$ denote the first $d+2$ latent states within block $b$, as defined in Lemma \ref{thm:block-dynamics}. Then, by Lemma \ref{thm:block-dynamics}, for each $b_k\in\mathcal B^*$,
$\Tilde{\mathbf A}^{(b_k)}$ and $\Tilde{\mathbf c}^{(b_k)}$ are functions of the first $d+2$ observations $\{\mathbf v_0^{(b_k)},\dots,\mathbf v_{d+1}^{(b_k)}\}$. 

For each selected block $b_k$,
$$
\Tilde{\mathbf A}^{(b_k)} = \mathbf A + \sum_{i=1}^m \Tilde u_i^{(b_k)}\mathbf B_i.
$$
Stacking these equations yields
\begin{align*}
    \begin{bmatrix}
\Tilde{\mathbf A}^{(b_1)}\\
\vdots\\
\Tilde{\mathbf A}^{(b_{m+1})}
\end{bmatrix}
=
\begin{bmatrix}
\mathbf I_d & \Tilde u_1^{(b_1)}\mathbf I_d & \cdots & \Tilde u_m^{(b_1)}\mathbf I_d\\
\vdots & \vdots & & \vdots\\
\mathbf I_d & \Tilde u_1^{(b_{m+1})}\mathbf I_d & \cdots & \Tilde u_m^{(b_{m+1})}\mathbf I_d
\end{bmatrix}
\begin{bmatrix}
\mathbf A\\
\mathbf B_1\\
\vdots\\
\mathbf B_m
\end{bmatrix} .
\end{align*}
Since $\Tilde{\mathbf U}^*$ and $\mathbf{I}_d$ are invertible, the block matrix 
\begin{equation*}
   \Tilde{\mathbf U}^* \otimes  \mathbf{I}_d=\begin{bmatrix}
\mathbf I_d & \Tilde u_1^{(b_1)}\mathbf I_d & \cdots & \Tilde u_m^{(b_1)}\mathbf I_d\\
\vdots & \vdots & & \vdots\\
\mathbf I_d & \Tilde u_1^{(b_{m+1})}\mathbf I_d & \cdots & \Tilde u_m^{(b_{m+1})}\mathbf I_d
\end{bmatrix}
\end{equation*}
is invertible. Therefore,
\begin{equation*}
    \begin{bmatrix}
\mathbf A\\
\mathbf B_1\\
\vdots\\
\mathbf B_m
\end{bmatrix}
=
\begin{bmatrix}
\mathbf I_d & \Tilde u_1^{(b_1)}\mathbf I_d & \cdots & \Tilde u_m^{(b_1)}\mathbf I_d\\
\vdots & \vdots & & \vdots\\
\mathbf I_d & \Tilde u_1^{(b_{m+1})}\mathbf I_d & \cdots & \Tilde u_m^{(b_{m+1})}\mathbf I_d
\end{bmatrix}^{-1}
\begin{bmatrix}
\Tilde{\mathbf A}^{(b_1)}\\
\Tilde{\mathbf A}^{(b_2)}\\
\vdots\\
\Tilde{\mathbf A}^{(b_{m+1})}
\end{bmatrix}.
\end{equation*}
In particular, $\mathbf{A}$ and $\mathbf{B}_1,\dots,\mathbf{B}_m$ are functions of observations in the discrete trajectory $\{\mathbf{z}(t^{(b)}+jr); \,j=0,\dots,d+1,\,b \in \mathcal{B}^*\}$.

Similarly, writing $\mathbf C=[\mathbf c_1,\dots,\mathbf c_m]$, where $\mathbf c_i\in\mathbb R^d$ is the $i$th column of $\mathbf C$, we have
$$
\Tilde{\mathbf c}^{(b_k)}
=
\sum_{i=1}^m \Tilde u_i^{(b_k)}\mathbf c_i,
\qquad \text{for } k=1,\dots,m+1.
$$
Stacking over the selected blocks gives
\begin{equation*}
    \begin{bmatrix}
\Tilde{\mathbf c}^{(b_1)}\\
\Tilde{\mathbf c}^{(b_2)}\\
\vdots\\
\Tilde{\mathbf c}^{(b_{m+1})}
\end{bmatrix}
=
\begin{bmatrix}
\mathbf I_d & \Tilde u_1^{(b_1)}\mathbf I_d & \cdots & \Tilde u_m^{(b_1)}\mathbf I_d\\
\mathbf I_d & \Tilde u_1^{(b_2)}\mathbf I_d & \cdots & \Tilde u_m^{(b_2)}\mathbf I_d\\
\vdots & \vdots & & \vdots\\
\mathbf I_d & \Tilde u_1^{(b_{m+1})}\mathbf I_d & \cdots & \Tilde u_m^{(b_{m+1})}\mathbf I_d
\end{bmatrix}
\begin{bmatrix}
\mathbf 0\\
\mathbf c_1\\
\vdots\\
\mathbf c_m
\end{bmatrix}.
\end{equation*}
Therefore,
\begin{equation*}
    \begin{bmatrix}
\mathbf 0\\
\mathbf c_1\\
\vdots\\
\mathbf c_m
\end{bmatrix}
=
\begin{bmatrix}
\mathbf I_d & \Tilde u_1^{(b_1)}\mathbf I_d & \cdots & \Tilde u_m^{(b_1)}\mathbf I_d\\
\mathbf I_d & \Tilde u_1^{(b_2)}\mathbf I_d & \cdots & \Tilde u_m^{(b_2)}\mathbf I_d\\
\vdots & \vdots & & \vdots\\
\mathbf I_d & \Tilde u_1^{(b_{m+1})}\mathbf I_d & \cdots & \Tilde u_m^{(b_{m+1})}\mathbf I_d
\end{bmatrix}^{-1}
\begin{bmatrix}
\Tilde{\mathbf c}^{(b_1)}\\
\Tilde{\mathbf c}^{(b_2)}\\
\vdots\\
\Tilde{\mathbf c}^{(b_{m+1})}
\end{bmatrix},
\end{equation*}
and $\mathbf{C}$ is uniquely determined by $\{\mathbf{z}(t^{(b)}+jr); \,j=0,\dots,d+1,\,b \in \mathcal{B}^*\}$.


Finally, we identify the initial state $\mathbf{s}^*$. Let $\Tilde{\mathbf{u}}_{\rm init}$ denote the constant stimulus vector during $[0,r)$. Typically $\Tilde{\mathbf{u}}_{\rm init} =0$, corresponding to no stimulus before the first TR, or $\Tilde{\mathbf{u}}_{\rm init} =\Tilde{\mathbf{u}}^{(0)}$, corresponding
to the stimulus during the first block. During this interval, the system evolves according to $\Tilde{\mathbf{A}}_{\rm init} = \mathbf{A} + \sum_{i=1}^m \Tilde{u}_{\textrm{init},i} \mathbf{B}_i$ and $\Tilde{\mathbf{c}}_{\rm init} = \mathbf{C}\Tilde{\mathbf{u}}_{\rm init}$. By the closed-form solution of the affine ODE, 
\begin{align*}
    \mathbf{z}(r) = e^{r \Tilde{\mathbf{A}}_{\rm init}}\mathbf{s}^*+w(r; \Tilde{\mathbf{A}}_{\rm init})\Tilde{\mathbf{c}}_{\rm init}.
\end{align*}
That is,
\begin{align*}
    \mathbf{s}^* = e^{-r \Tilde{\mathbf{A}}_{\rm init}}\{ \mathbf{z}(r)-w(r; \Tilde{\mathbf{A}}_{\rm init})\Tilde{\mathbf{c}}_{\rm init}\}.
\end{align*}
Since $\mathbf{A}, \mathbf{B}_1,\dots,\mathbf{B}_m$, and $\mathbf{C}$ are functions of observations $\{\mathbf{z}(t^{(b)}+jr); \,j=0,\dots,d+1,\,b \in \mathcal{B}^*\}$, so are $\Tilde{\mathbf{A}}_{\rm init}$ and $\Tilde{\mathbf{c}}_{\rm init}$. It follows that $\mathbf{s}^*$ is a function of $\{\mathbf{z}(r)\}\cup\{\mathbf{z}(t^{(b)}+jr); \,j=0,\dots,d+1,\,b \in \mathcal{B}^*\}$.

Therefore, $(\mathbf{s}^*,\theta^*)$ are uniquely determined by the discrete trajectory $\{\mathbf{z}(rj); \,\,j=1,\dots,n\}$, completing the proof.
\end{proof}

\subsection{Proof of Theorem \ref{thm:one-to-one}}
\begin{theorem}
If $h(r) \neq 0$, that is, the canonical HRF evaluated at the TR $r$ is nonzero, then the mapping
$$
\{\mathbf{z}(t_j); j=0,\dots,n-1\} \;\mapsto\; \{\mathbb{E}[\mathbf{y}(t_j)]; j \in \{1,\hdots,n\}\}
$$
with $\mathbb{E}[\mathbf{y}(t_j)] = \mu(t_j) + \beta$ is injective. In particular, the system defined by \eqref{eqn:canonical-dcm} and \eqref{eqn:canonical-dcm-obs} is $(\mathbf{s}^*,\theta^*)$-identifiable from $\{\mathbb{E}[\mathbf{y}(t_1)],\dots,\mathbb{E}[\mathbf{y}(t_n)]\}$ in the sense that for $\mathbf{s}^\prime \in \mathbb{R}^d$ and $\theta^\prime = \{\mathbf{A}^\prime,\mathbf{B}_1^\prime,\hdots,\mathbf{B}_m^\prime,\mathbf{C}^\prime\}$ with $(\mathbf{s}^\prime,\theta^\prime) \neq (\mathbf{s}^*,\theta^*)$, there exists $j \in \{1,\hdots,n\}$ such that
$$
\mathbb{E}[\mathbf{y}(t_j;\mathbf{s}^\prime,\theta^\prime)] \neq \mathbb{E}[\mathbf{y}(t_j;\mathbf{s}^*,\theta^*)].
$$
As $\beta = \mathbb{E}[\mathbf{y}(t_j)] - \mu(t_j)$ for any $j$, the intercept $\beta$ is uniquely determined once $\mu(t_j)$ is identified.
\end{theorem}
\begin{proof}

We begin by providing a more explicit definition of the discrete time convolution \eqref{eqn:discrete-convolve} in the main text. For a function $f:[0,\infty)\to \mathbb{R}$, we define $f[j] = f(t_j) = f(rj)$ as the sampled vector with sampling interval $r$ for $j\in \mathbb{N}$. We model the mean $\mu[j] = \mu(t_j)$ of the $j$th observation $\mathbf{y}[j] = \mathbf{y}(t_j)$ as 
\begin{equation}\label{eqn:discrete-convolve-full}
     \mu[j] = (\mathbf{z} * h)[j] = \sum_{i=-\infty}^\infty h[i]\mathbf{z}[j-i] = \sum_{i=0}^j h[i]\mathbf{z}[j-i],
 \end{equation}
since $h[i]=h(ri) = 0$ when $i<0$ and $\mathbf{z}[j-i]=0$ when $i>j$, as neither the neural signal nor the canonical HRF have nonzero values prior to the start of the scan.

From (\ref{eqn:discrete-convolve-full}), we have $\mu(t_j) = \mu[j] = \sum_{i=0}^j h[i]\mathbf{z}[j-i]$. Explicitly writing out the first few terms of $\mu[j]$ gives,
\begin{align*}
    \mu[1] &= h[0]\mathbf{z}[1] + h[1]\mathbf{z}[0], \\
    \mu[2] &= h[0]\mathbf{z}[2] + h[1]\mathbf{z}[1] + h[2]\mathbf{z}[0], \\
    \mu[3] &= h[0]\mathbf{z}[3] + h[1]\mathbf{z}[2] + h[2]\mathbf{z}[1] + h[3]\mathbf{z}[0], \\
    &\vdots \\
    \mu[n] &= h[0]\mathbf{z}[n] + h[1]\mathbf{z}[n-1] + \dots + h[n-1]\mathbf{z}[1] + h[n]\mathbf{z}[0].
\end{align*}
Note that $h[0] = h(t_0) = h(0) = 0$ from the definition of the canonical HRF. After eliminating this term, we can represent the system of equations in matrix form:
\begin{equation*}
    \begin{bmatrix}
        \mu[1] \\ \mu[2] \\ \mu[3] \\ \vdots \\ \mu[n]
    \end{bmatrix} =
    \begin{bmatrix}
        h[1]   & 0      & 0      & \cdots & 0   \\
        h[2]   & h[1]   & 0      & \cdots & 0   \\
        h[3]   & h[2]   & h[1]   & \ddots & 0   \\
        \vdots & \vdots & \vdots & \ddots & 0   \\
        h[n] & h[n-1] & h[n-2] & \dots & h[1]
    \end{bmatrix}
    \begin{bmatrix}
        \mathbf{z}[0] \\ \mathbf{z}[1] \\ \mathbf{z}[2] \\ \vdots \\ \mathbf{z}[n-1]
    \end{bmatrix},
\end{equation*}
or, alternatively, $\Tilde{\mu} = \mathbf{H}\Tilde{\mathbf{z}}$ where $\Tilde{\mu} := [\mu[1],\dots,\mu[n]]^\top$ and $\Tilde{\mathbf{z}}:=[\mathbf{z}[0],\dots,\mathbf{z}[n-1]]^\top$. The matrix $\mathbf{H} \in \mathbb{R}^{n \times n}$ is a lower triangular matrix, and $\det(\mathbf{H}) = \prod_{j=1}^n h[1] = (h[1])^n$. If $h[1] \neq 0$, then $\det(\mathbf{H}) \neq 0$, and $\mathbf{H}$ is invertible. Thus, $\Tilde{\mathbf{z}} = \mathbf{H}^{-1}\Tilde{\mu}$, so for a given $\Tilde{\mathbf{z}}$, there is a one-to-one mapping from $\Tilde{\mathbf{z}}$ to $\Tilde{\mu}$ when $\mathbf{H}$ is invertible.

Finally, since $\mathbb{E}[\mathbf{y}(t_j)] = \mu(t_j) + \beta$ with $\beta$ constant across $j$, the mapping from $\Tilde{\mu}_s$ to $\{\mathbb{E}[\mathbf{y}(t_j)]\}$ is a translation and therefore injective. Hence, the one-to-one correspondence between $\Tilde{\mathbf{z}}$ and $\Tilde{\mu}$ extends to a one-to-one mapping between $\Tilde{\mathbf{z}}$ and $\{\mathbb{E}[\mathbf{y}(t_j)]\}$. Moreover, $\beta$ is uniquely determined from $\mathbb{E}[\mathbf{y}(t_j)] - \mu(t_j)$.
\end{proof}

\section{Group-Level Bayesian Hierarchical Model Details}\label{appx:grp-mdl}

For each subject $k = 1,\dots,K$, let $\hat{\theta}_{z,k} \in \mathbb{R}^{p_{\theta_z}}$ denote the estimated subject-level neural parameter posterior means, $\mathbf{S}_k \in \mathbb{R}^{{p_{\theta_z}}\times {p_{\theta_z}}}$ the estimated within-subject posterior covariance, and $\mathbf{b}_k \in \mathbb{R}^{q_s}$ the standardized and/or effect-coded subject-level covariates. We define the mean structure:
\begin{equation*}
    \eta_k = \boldsymbol{\alpha} + \Theta\mathbf{b}_k,
\end{equation*}
where $\boldsymbol{\alpha} \in \mathbb{R}^{p_{\theta_z}}$ are group-level intercepts and $\Theta \in \mathbb{R}^{{p_{\theta_z}}\times q_s}$ are regression coefficients. We standardize continuous covariates and effect-code categorical covariates because the interpretation of the group-level intercepts $\boldsymbol{\alpha}$ is then in terms of the ``average subject'' and is not specific to particular reference level(s). The hierarchical model can be written as,
\begin{align*}
\theta_{z,k} \mid \boldsymbol{\alpha}, \Theta, \mathbf{T} 
    &\sim \text{N}_{p_{\theta_z}}\big(\eta_k,\; \mathbf{T}\big), \\
\hat{\theta}_{z,k} \mid \theta_{z,k}
    &\sim \text{N}_{p_{\theta_z}}\big(\theta_{z,k},\; \mathbf{S}_k\big),
\end{align*}
where the latent vectors $\theta_{z,k}$ represent the true (unobserved) subject-level neural parameters, and $\mathbf{T} \in \mathbb{R}^{{p_{\theta_z}}\times {p_{\theta_z}}}$ is the between-subject covariance matrix describing heterogeneity across subjects. 

Because both levels are multivariate normal, the latent $\theta_{z,k}$ can be integrated out analytically:
\begin{align*}
p(\hat{\theta}_{z,k} \mid \boldsymbol{\alpha}, \Theta, \mathbf{T})
    &= \int p(\hat{\theta}_{z,k} \mid \theta_{z,k})\, p(\theta_{z,k} \mid \boldsymbol{\alpha}, \Theta, \mathbf{T})\, d\theta_{z,k} \\
    &= \int 
        \text{N}_{p_{\theta_z}}(\hat{\theta}_{z,k} \mid \theta_{z,k}, \mathbf{S}_k)\;
        \text{N}_{p_{\theta_z}}(\theta_{z,k} \mid \eta_k, \mathbf{T})\;
        d\theta_{z,k} \\
    &= \text{N}_{p_{\theta_z}}\big(\hat{\theta}_{z,k} \mid \eta_k,\; \mathbf{T} + \mathbf{S}_k\big).
\end{align*}
Thus, the marginalized likelihood used for estimation is,
\begin{equation*}
    \hat{\theta}_{z,k} \sim \text{N}_{p_{\theta_z}}\big(\eta_k,\; \mathbf{T} + \mathbf{S}_k\big),
\end{equation*}
where the total covariance is the sum of the between-subject variability $\mathbf{T}$ and the within-subject uncertainty $\mathbf{S}_k$.
This form is algebraically equivalent to the hierarchical model above, but is computationally more efficient because the $\theta_{z,k}$ are integrated out. The between-subject covariance $\mathbf{T}$ is expressed through a Cholesky factorization with an LKJ prior \cite{lewandowski-etal-2009}:
$$
\mathbf{T} = \mathbf{L}_\tau \mathbf{L}_\tau^\top, 
\quad 
\mathbf{L}_\tau = \mathrm{diag}(\tau) \, \mathbf{L}_{\mathrm{corr}},
$$
where $\tau \in \mathbb{R}_{>0}^{p_{\theta_z}}$ are between-subject standard deviations and $\mathbf{L}_{\mathrm{corr}}$ is the Cholesky factor of a correlation matrix capturing correlations among parameters. Priors are specified as,
\begin{align*}
\boldsymbol{\alpha} &\sim \text{N}_{p_{\theta_z}}(\mathbf{0}, \mathbf{I}), \\
\mathrm{vec}(\Theta) &\sim \text{N}_{{p_{\theta_z}}q_s}(\mathbf{0}, \mathbf{I}),\\
\tau &\sim \text{N}^+_{p_{\theta_z}}(\mathbf{0}, \mathbf{I}),\\
\mathbf{L}_{\mathrm{corr}} &\sim \mathrm{LKJ}(2),
\end{align*}
where $\text{N}^+$ denotes a half-normal distribution. Similar to CDCM, we use NUTS \cite{hoffman-gelman-2014} for posterior distribution sampling of group-level parameters with initial values from the output of the Pathfinder algorithm.\cite{zhang-etal-2022}

\section{HCP Social Cognition Task Simulation Study}\label{appx:hcp-sim}

Comparisons of group-level results from the HCP social cognition task indicate agreement between CDCM and SPM, particularly in the signs of estimated parameters, with one notable exception: the driving input from the All Motion stimulus to the pSTS. For this parameter, CDCM consistently estimates a negative effect, whereas SPM estimates a positive effect. As the ground truth is unknown in real data, we conduct a follow-up simulation study to assess which method more accurately recovers the underlying parameters under the specified model.

In this simulation study, we create four simulated datasets of the same size as the HCP data (100 subjects $\times$ 2 sessions $\times$ 2 hemispheres), and estimate the single subject DCMs for all datasets using both CDCM and SPM. The results are summarized at the group-level for each of the datasets in the same way as they are for the real data. From the real data, we can treat the group-level posterior means as the true set of neural parameters, and simulate data from those values. The four simulated datasets arise from the $2 \times 2$ combinations of true values (CDCM group-level as the truth, and SPM group-level as the truth) and data-generating models (CDCM DGM and SPM DGM). Within each dataset, the true values we use correspond to the posterior means for that session and hemisphere combination. As an example, for the dataset using CDCM as the truth and SPM as the DGM, then for LR phase encoding and the left hemisphere, the true values used are the CDCM group-level posterior means from the LR phase encoding, mask L panel of Figure \ref{fig:group-alpha-comparison}, and the simulated data are generated from those values using SPM. 

We also retain the set of subject IDs such that, in each dataset, we use the $\mathbf{U}$ matrix for that subject and session, just as in the real data. For the two datasets using SPM as the data-generating model, we also use the individual estimated hemodynamic parameters from the Balloon model for that subject, session, and hemisphere in the real data. Since there are no hemodynamic parameters for CDCM, this step is not necessary for the other two datasets. After simulating the true signals, we add Gaussian noise with a target SNR of 1.68, just as in our earlier simulations. Similar to Section \ref{sec:balloon-sim}, the simulated signals for all datasets have a range close to four. To ensure that both methods are fit to identical data on a scale comparable to the real data, the additive Gaussian noise is reduced when necessary to prevent the data range from exceeding four and triggering SPM’s internal rescaling. As a result, the noise is more constrained in V5 than in pSTS, leading to a higher average SNR in V5.

The convergence criterion for each estimation procedure is identical as described for the real data. For CDCM, there are five chains total across the four datasets that do not converge initially and require an increase in the target acceptance probability from $0.9$ to either 0.95 (3 chains) or 0.99 (2 chains) to converge according to a multivariate ESS satisfying $\alpha_{ESS} = \varepsilon_{ESS} = 0.05$. These five chains all come from the two datasets using SPM's estimates as the truth. SPM's VL optimization converges in all cases within 512 iterations. 

The simulation results are shown in Table \ref{tab:sensitivity}, with the left three columns defining the session and hemisphere condition, the true set of values, and the DGM, followed by the results from each estimation procedure. When summarized at the group-level, both estimation procedures tend to be overconfident, with narrow HPD intervals that often miss the true value. We still observe this in cases where the truth, DGM, and estimation procedure all match. 

\begin{table}[t!]
\centering
\caption{Simulated HCP dataset group-level results. Simulated subjects' effective connectivity networks are estimated using both estimation procedures, and then we estimate the group-level parameters controlling for subject-level covariates. From the group-level results, we report the proportion of correct sign $(+/-)$ matches between the truth and the posterior mean estimates, as well as the mean distance to the true values. Both methods have a large proportion of correct sign matches, and similar average distances to the true values.}
\begin{tabular}{ccc|cc|cc}
\hline
\textbf{Condition} & \textbf{Truth} & \textbf{DGM} & \multicolumn{2}{c|}{\textbf{Correct Sign}} & \multicolumn{2}{c}{\textbf{Mean Dist. to Truth}} \\
 &  &  & \textbf{CDCM} & \textbf{SPM} & \textbf{CDCM} & \textbf{SPM} \\ \hline
LR-L & CDCM & CDCM & 1.0 & 0.8 & 0.113 & 0.286 \\
LR-L & CDCM & SPM & 1.0 & 0.9 & 0.123 & 0.167 \\
LR-L & SPM & CDCM & 0.7 & 0.9 & 0.183 & 0.048 \\
LR-L & SPM & SPM & 0.8 & 1.0 & 0.245 & 0.046 \\ \hline
LR-R & CDCM & CDCM & 1.0 & 0.7 & 0.073 & 0.386 \\
LR-R & CDCM & SPM & 1.0 & 0.8 & 0.106 & 0.292 \\
LR-R & SPM & CDCM & 1.0 & 0.7 & 0.129 & 0.131 \\
LR-R & SPM & SPM & 0.8 & 0.9 & 0.228 & 0.105 \\ \hline
RL-L & CDCM & CDCM & 0.9 & 0.6 & 0.117 & 0.196 \\
RL-L & CDCM & SPM & 0.9 & 0.8 & 0.113 & 0.113 \\
RL-L & SPM & CDCM & 0.6 & 0.7 & 0.271 & 0.131 \\
RL-L & SPM & SPM & 0.6 & 0.9 & 0.377 & 0.043 \\ \hline
RL-R & CDCM & CDCM & 0.8 & 0.5 & 0.124 & 0.434 \\
RL-R & CDCM & SPM & 0.8 & 0.7 & 0.198 & 0.263 \\
RL-R & SPM & CDCM & 0.8 & 0.9 & 0.396 & 0.238 \\
RL-R & SPM & SPM & 0.8 & 1.0 & 0.261 & 0.093 \\ \hline
\end{tabular}
\label{tab:sensitivity}
\end{table}

This apparent lack of coverage is likely driven in part by the structure of the simulation design. Within each dataset, the 100 subjects are simulated as repeated noisy realizations of the same underlying parameter vector. This provides a controlled setting that reflects the structure of the real data while limiting between-subject variability, allowing the group-level model to efficiently pool information across observations. At the same time, the group-level analysis retains subject-level covariates as in the real data, even though the simulated parameter values do not vary systematically with those covariates. Together, these features lead to relatively narrow group-level intervals, so coverage is not the most informative metric for this comparison. Instead, we focus on the proportion of correct sign matches to the truth, as well as the mean distance to the truth, given that the sign difference between CDCM and SPM for the second estimate in $\hat{\boldsymbol{\alpha}}_C$ is what prompted this simulation study.

Even though both methods exhibit poor coverage, they both estimate the correct sign of the parameter the majority of the time, and do so with similar magnitude, as the mean distances between the true values and estimated values are small. This is the case even when considering the typical scale of the neural parameters in DCM. The results in Table \ref{tab:sensitivity} also show that the data-generating mechanisms matter less for simulation performance than the true parameter set used to simulate the data, as there is a larger drop in parameter recovery when the true values do not match the estimation procedure than when the DGM does not match the estimation procedure. For a more visual representation of the results shown in Table \ref{tab:sensitivity}, we refer the reader to Supplementary Material \ref{appx:fig-tab}, where results from each simulated dataset are plotted alongside the true values in the same fashion as in Figure \ref{fig:group-alpha-comparison}.

Table \ref{tab:sensitivity} shows the overall simulation results, which demonstrate that both methods, even under misspecification and a different set of parameters than originally estimated for the real data, do a reasonable job of estimating values close to the truth, even if the intervals do not capture the truth. Regarding the original motivation for the simulation, the difference in direction for the driving input of All Motion into the pSTS, CDCM estimates the correct sign for this parameter every time, whereas SPM only estimates the correct sign 68.75\% of the time. However, CDCM’s proportion is nearly 75\%, as the posterior mean estimates for this input are only marginally positive for the dataset generated using SPM and DGM, resulting in intervals that narrowly exclude zero. Overall, the simulation results indicate that both models provide reasonable estimates of the network parameters across a range of conditions. In the real data analysis, however, the direction of the driving input from All Motion to the pSTS cannot be determined conclusively; accordingly, we do not attempt to interpret the biological significance of this parameter.

\section{Supplementary Figures and Tables}\label{appx:fig-tab}

Here, we present several supplementary figures and tables pertaining to the simulation in Section \ref{sec:balloon-sim}, and the HCP real data application from Section \ref{sec:hcp}.

\begin{figure}[ht!]
    \centering
    \includegraphics[width=0.9\linewidth]{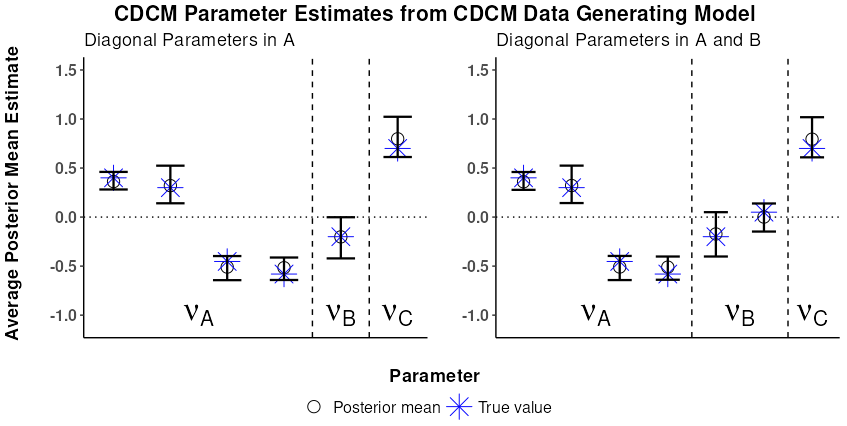}
    \caption{Average posterior mean estimates and 95\% HPD intervals across 50 simulations using CDCM where data are also generated using CDCM. The left panel includes diagonal parameters in $\mathbf{A}$ only (positions three and four), whereas the right includes diagonal parameters in $\mathbf{A}$ (positions three and four) as well as one diagonal parameter in $\mathbf{B}$ (position six). Intervals and true values for the diagonal parameters in $\mathbf{A}$ have been reparameterized. Here, like in the simulation from Section \ref{sec:comp-sim}, CDCM recovers all parameters in both scenarios. As mentioned in the main text, the intervals are slightly overcovering compared to that of SPM.}
\end{figure}

\begin{figure}[t!]
    \centering
    \includegraphics[width=0.9\linewidth]{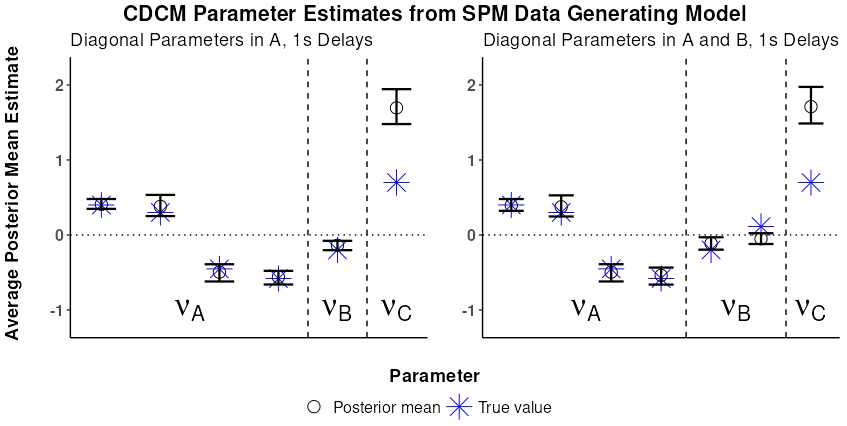}
    \caption{Average posterior mean estimates and 95\% HPD intervals across 50 simulations using CDCM where data are generated using SPM with 1 s signal delays. The left panel includes diagonal parameters in $\mathbf{A}$ only (positions three and four), whereas the right includes diagonal parameters in $\mathbf{A}$ (positions three and four) as well as one diagonal parameter in $\mathbf{B}$ (position six). Intervals and true values for the diagonal parameters in $\mathbf{A}$ have been reparameterized. The results with delays do not differ much from the results shown for CDCM in Figure \ref{fig:misspec-result}.}
\end{figure}

\begin{table}[b!]
\centering
\caption{Summary of NUTS sampler hyperparameter tuning changes for the HCP social cognition task analysis. For these cases, the CDCM default settings of a target acceptance probability of 0.9 and 5,000 warm-up iterations resulted in either divergence or maximum tree depth warnings.}
\resizebox{\textwidth}{!}{
\begin{tabular}{ccccc}
\hline
\textbf{Subject ID} & \textbf{\begin{tabular}[c]{@{}c@{}}Session \\ (Phase Encoding)\end{tabular}} & \textbf{Hemisphere} & \textbf{\begin{tabular}[c]{@{}c@{}}Target Acceptance\\ Probability\end{tabular}} & \textbf{\begin{tabular}[c]{@{}c@{}}Warm-up \\ Iterations\end{tabular}} \\ \hline
101107 & LR & L & 0.80 &  \\
118730 & RL & R & 0.99 &  \\
131217 & LR & L & 0.80 & 8,000 \\
149337 & RL & L & 0.95 &  \\
149539 & LR & R & 0.99 & 8,000 \\
154734 & LR & L & 0.80 & 8,000 \\
162733 & LR & L & 0.95 & 8,000 \\
163129 & RL & L & 0.80 & 8,000 \\
189450 & RL & R & 0.95 &  \\
245333 & RL & L & 0.99 &  \\
414229 & RL & R & 0.99 &  \\ \hline
\end{tabular}}
\end{table}

\begin{figure}[ht!]
    \centering
    \includegraphics[width=0.9\linewidth]{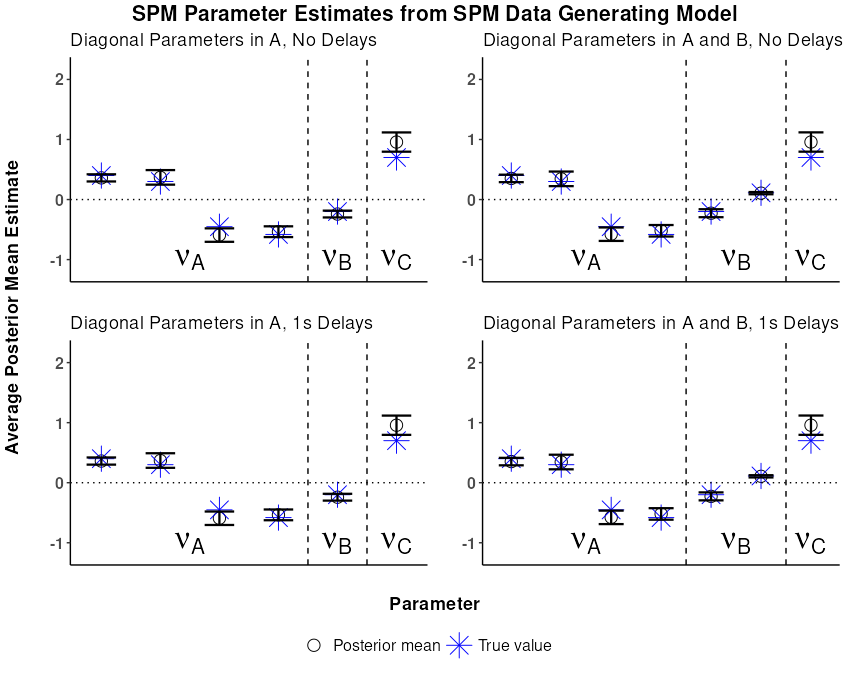}
    \caption{Average posterior mean estimates and 95\% HPD intervals across 50 simulations using SPM where data are also generated using SPM. The top two panels do not incorporate signal delays between regions and stimuli, and the bottom two incorporate 1s delays. The left two panels include diagonal parameters in $\mathbf{A}$ only (positions three and four), whereas the right two include diagonal parameters in $\mathbf{A}$ (positions three and four) as well as one diagonal parameter in $\mathbf{B}$ (position six). Intervals and true values for the diagonal parameters in $\mathbf{A}$ and $\mathbf{B}$ have been reparameterized. Most notably, SPM appears to be overly confident in the parameter estimates, and does not cover all parameters. This is especially true for the value in $\mathbf{C}$, interestingly, in which SPM consistently overestimates the parameter.}
\end{figure}

\begin{figure}[ht!]
    \centering
    \includegraphics[width=\linewidth]{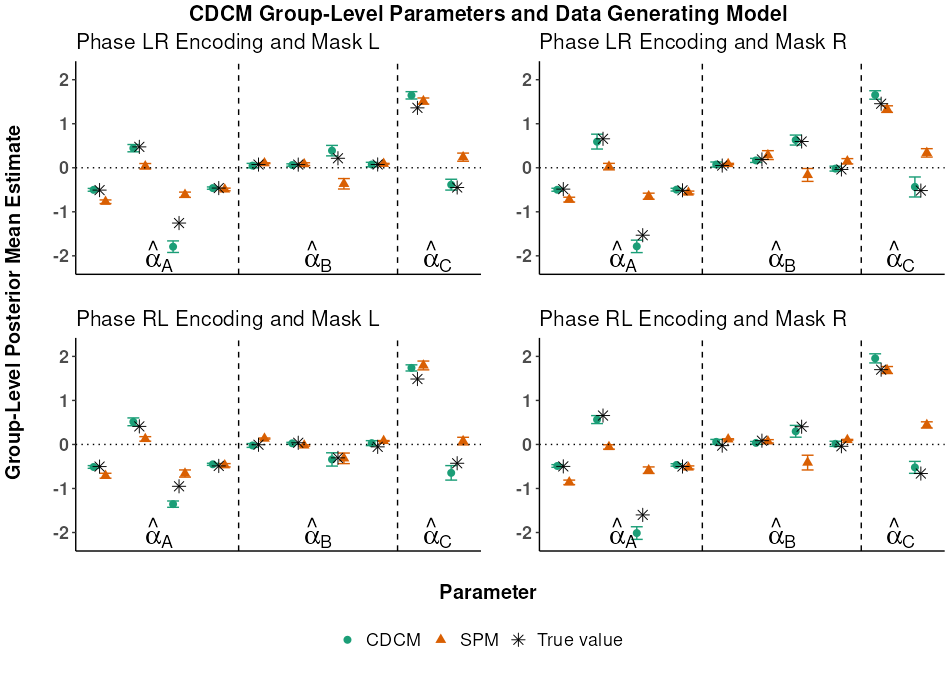}
    \caption{Group-level posterior mean estimates and 95\% HPD interval comparison between CDCM and SPM across all phase encoding and hemisphere combinations for the HCP follow-up simulation study. The data are simulated using CDCM's group-level posterior means shown in Figure \ref{fig:group-alpha-comparison}, and CDCM's data generating model; the true values are denoted by the stars. Posterior means and intervals for the diagonal parameters have been reparameterized according to the respective modeling definition to facilitate direct comparison. Diagonal parameters are in positions one, four, five, and eight within each panel. Unsurprisingly, CDCM outperforms SPM in terms of recovering the true parameters, but both methods are consistent in many cases. For the key parameter of interest, the second value in $\mathbf{C}$, SPM continues to estimate a positive value, whereas CDCM correctly estimates a negative value.}
\end{figure}

\begin{figure}[ht!]
    \centering
    \includegraphics[width=\linewidth]{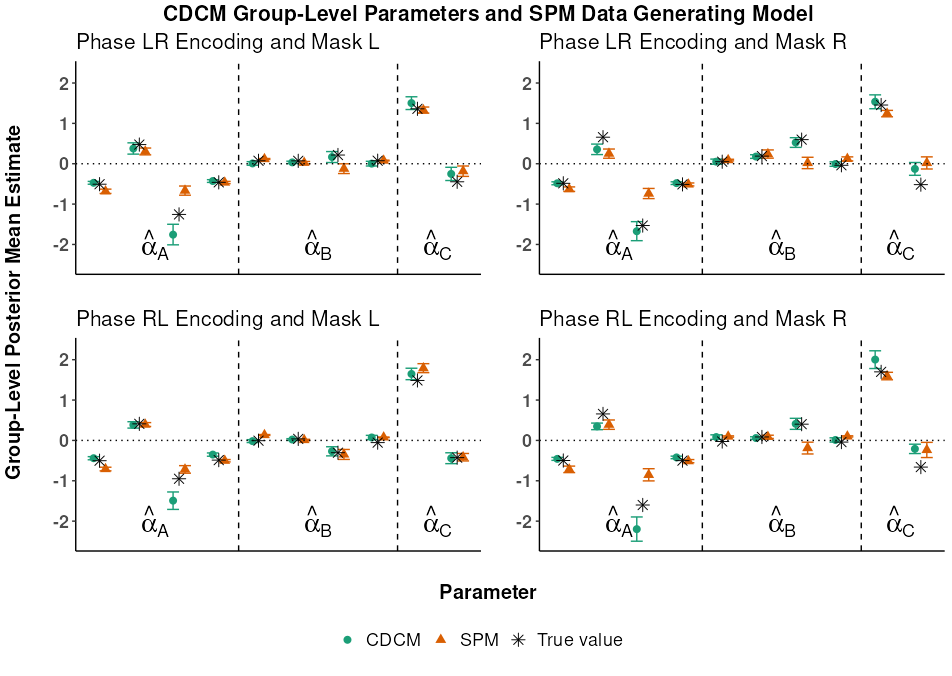}
    \caption{Group-level posterior mean estimates and 95\% HPD interval comparison between CDCM and SPM across all phase encoding and hemisphere combinations for the HCP follow-up simulation study. The data are simulated using CDCM's group-level posterior means shown in Figure \ref{fig:group-alpha-comparison}, and SPM's data generating model; the true values are denoted by the stars. Posterior means and intervals for the diagonal parameters have been reparameterized according to the respective modeling definition to facilitate direct comparison. Diagonal parameters are in positions one, four, five, and eight within each panel. Even with the mismatch between true values and DGM from the real data, both models estimate values very close to the truth in nearly all cases. Concerning the key parameter of interest in $\mathbf{C}$, both estimation procedures estimate a correct negative value, or at the very least, a value whose HPD interval includes zero.}
\end{figure}

\begin{figure}[ht!]
    \centering
    \includegraphics[width=\linewidth]{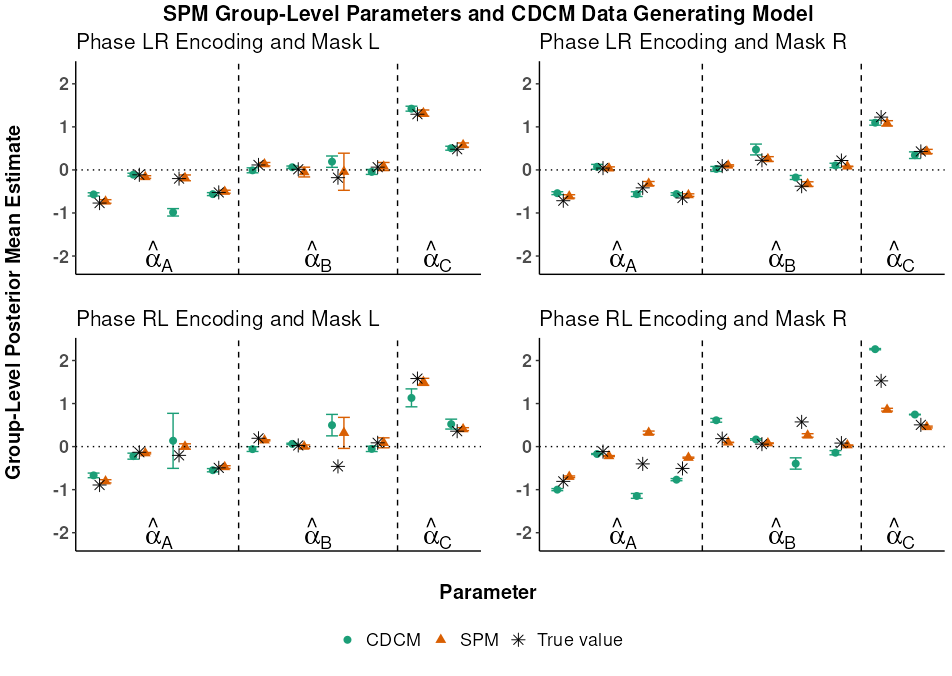}
    \caption{Group-level posterior mean estimates and 95\% HPD interval comparison between CDCM and SPM across all phase encoding and hemisphere combinations for the HCP follow-up simulation study. The data are simulated using SPM's group-level posterior means shown in Figure \ref{fig:group-alpha-comparison}, and CDCM's data generating model; the true values are denoted by the stars. Posterior means and intervals for the diagonal parameters have been reparameterized according to the respective modeling definition to facilitate direct comparison. Diagonal parameters are in positions one, four, five, and eight within each panel. For the LR phase encoding panels, both models perform fairly consistently, even with the true value and DGM mismatch. For RL phase coding using the left hemisphere, this remains true for all parameters except for that of the third modulation parameter, in which both models incorrectly estimate the sign. Interestingly, both models struggle overall with RL phase encoding and right hemisphere, as both are very confident, but not as close to the truth as in the other session $\times$ hemisphere combinations. For the key parameter of interest, the second value in $\mathbf{C}$, both methods correctly estimate a positive value.}
\end{figure}

\begin{figure}[ht!]
    \centering
    \includegraphics[width=\linewidth]{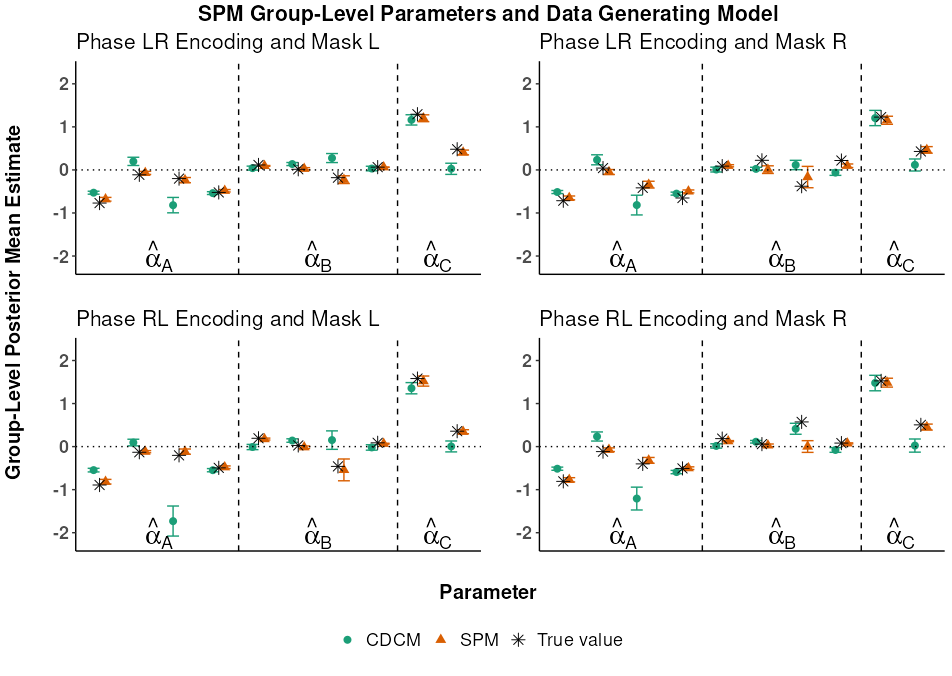}
    \caption{Group-level posterior mean estimates and 95\% HPD interval comparison between CDCM and SPM across all phase encoding and hemisphere combinations for the HCP follow-up simulation study. The data are simulated using SPM's group-level posterior means shown in Figure \ref{fig:group-alpha-comparison}, and SPM's data generating model; the true values are denoted by the stars. Posterior means and intervals for the diagonal parameters have been reparameterized according to the respective modeling definition to facilitate direct comparison. Diagonal parameters are in positions one, four, five, and eight within each panel. Unsurprisingly, SPM outperforms CDCM in terms of recovering the true parameters, but both methods are consistent in many cases. For the key parameter of interest, the second value in $\mathbf{C}$, SPM correctly estimates a positive value in all cases, and CDCM estimates values that are just barely positive, with all intervals including zero.}
\end{figure}

\end{document}